\documentclass[10pt,reqno,tbtags]{amsart}

\usepackage{caption}

\usepackage{tikz}
\usepackage{tikz-feynman}
\usetikzlibrary{calc,matrix,patterns}

\usepackage[top=3.5cm , bottom=3cm, left=3cm , right=3cm]{geometry}
\usepackage{amsmath,amssymb,amsfonts,amsthm} 
\usepackage{mathrsfs,latexsym} 
\usepackage{mathtools}
\usepackage{thmtools}

\numberwithin{equation}{section}

\usepackage{calrsfs}
\DeclareMathAlphabet{\pazocal}{OMS}{zplm}{m}{n}

\usepackage{color}
\usepackage{cite}

\usepackage{epsfig}
\usepackage{graphicx}
\usepackage[all]{xy}

\usepackage{bm}
\usepackage{hyperref}

\DeclareFontFamily{OT1}{pzc}{}
\DeclareFontShape{OT1}{pzc}{m}{it}{<-> s * [1.1500] pzcmi7t}{}
\DeclareMathAlphabet{\mathpzc}{OT1}{pzc}{m}{it}

\usepackage{float}
\usepackage{wasysym}

\usepackage{cite}\medskip
\usepackage{todonotes}
\roman{enumi}

\newtheorem{theorem}{Theorem}[section]
\newtheorem{definition}[theorem]{Definition}
\newtheorem{lemma}[theorem]{Lemma}
\newtheorem{proposition}[theorem]{Proposition}
\newtheorem{corollary}[theorem]{Corollary}
\newtheorem{axiom}{Axiom}

\declaretheorem[sibling=theorem,style=definition, qed=\bell]{remark}
\declaretheorem[sibling=theorem,style=definition, qed=\kreuz]{example}

\numberwithin{equation}{section}

\newcommand{\CC}{{\mathbb C}}
\newcommand{\RR}{{\mathbb R}}
\newcommand{\MM}{{\mathbb M}}
\newcommand{\NN}{{\mathbb N}}

\newcommand{\Oc}{{\mathcal{O}}}

\newcommand{\Dc}{{\mathcal{D}}}
\newcommand{\Ec}{{\mathcal{E}}}

\newcommand{\Fc}{{\mathcal{F}}}

\newcommand{\Lc}{{\mathcal{L}}}

\newcommand{\Rc}{{\mathcal{R}}}

\newcommand{\Ups}{{\Upsilon}}

\newcommand{\fA}{{\mathfrak A}}

\newcommand{\Floc}{\Fc_{\mathrm{loc}}}                   

\newcommand{\supp}{{\mathrm{supp} \, }}

\newcommand{\be}{\begin{equation}}
\newcommand{\ee}{\end{equation}}
\newcommand{\ox}{\otimes}

\def\eg{{\it e.g.\ }}
\def\ie{{\it i.e.\ }}

\def\etc{{\it etc.}}
\newcommand{\const}{\mathrm{const}}

\newcommand{\0}{\emptyset}
\newcommand{\g}{\mathpzc{g}}
\newcommand{\tildeg}{\mathpzc{\tilde{g}}}
\newcommand{\h}{\mathpzc{h}}
\newcommand{\tildeh}{\mathpzc{\tilde{h}}}
\newcommand{\kk}{\mathpzc{k}}
\newcommand{\e}{\mathpzc{e}}
\newcommand{\G}{\mathpzc{G}}
\newcommand{\HL}{\mathpzc{H}_L}
\newcommand{\HLz}{\mathpzc{H}_{L,\zeta}}
\newcommand{\LieGc}{\mathrm{Lie}\G_c(M)}
\newcommand{\LieRc}{\mathrm{Lie}\,\Rc_c}

\newcommand{\naturaltra}{\overset{\centerdot}{\longrightarrow}}
\begin{document} 
\title[The unitary Anomalous Master Ward Identity]{The unitary Master Ward Identity: Time slice axiom, Noether's Theorem and Anomalies}

\author[R. Brunetti]{Romeo Brunetti}
\address{Dipartimento di Matematica, Universit\`a di Trento, 38123 Povo (TN), Italy}
\email{romeo.brunetti@unitn.it}

\author[M. D\"utsch]{Michael D\"utsch}
\address{Institute f\"ur Theoretische Physik, Universit\"at G\"ottingen, 37077 G\"ottingen, Germany}
\email{michael.duetsch3@gmail.com}

\author[K. Fredenhagen]{Klaus Fredenhagen}
\address{II. Institute f\"ur Theoretische Physik, Universit\"at Hamburg, 22761 Hamburg, Germany}
\email{klaus.fredenhagen@desy.de}

\author[K. Rejzner]{Kasia Rejzner}
\address{Department of Mathematics, University of York, YO10 5DD York, UK}
\email{kasia.rejzner@york.ac.uk}

\dedicatory{Dedicated to Detlev Buchholz on the occasion of his $77^{\textrm{th}}$ birthday}

\begin{abstract}
 \noindent The C*-algebraic formulation of generic interacting quantum field theories, recently presented by Detlev Buchholz and one of the authors (KF), is enriched by a unitary version of the Master Ward Identity, which was postulated some time ago by Franz Marc Boas, Ferdinand Brennecke and two of us (MD,KF). It is shown that the corresponding axiom implies the validity of the time slice axiom. Moreover, it opens the way for a new approach to Noether's Theorem where it yields directly the unitaries implementing the symmetries. It also unravels interesting aspects of the role of anomalies in quantum field theory.
\end{abstract}

\maketitle

\tableofcontents

\section{Introduction}
The C*-algebraic formulation of quantum physics is well known for its rather unique combination of conceptual clarity and mathematical precision \cite{RevBook}. In quantum field theory, the algebraic approach of Araki, Borchers, Haag and Kastler \cite{Araki09,Borchers1,Haag64,Haag92} has led to deep insights into the structure of the theory. As particularly important examples one may mention the theory of superselection sectors of Doplicher, Haag and Roberts (see \eg \cite{Haag92}) which is at the basis of recent work in conformal field theory (see \eg Rehren's contribution in \cite{RevBook}), and Tomita-Takesaki modular theory \cite{StraZs19,Takesaki-book,BorchersT} which is instrumental \eg in Quantum Statistical Mechanics via the KMS conditions \cite{BR}, in the reconstruction of symmetries via Wiesbrock's half-sided modular condition/intersection \cite{Wies1,Wies2}, in localization properties of field theories with their particle interpretation \cite{BGL} and its mathematical ramifications \cite{NeebOlaf}, and which is fruitfully used nowadays to discuss entropy and entanglement in quantum field theories \cite{CLR19,HS18,LongoXu19}.

On the other hand, the axioms of algebraic quantum field theory have not yet been equally successful in fixing specific interacting theory models. As a \emph{bon mot}, Rudolf Haag used to ask colleagues ``What is the Lagrangian?'' and no answer could satisfy him.

Recently, Detlev Buchholz and one of us (KF) succeeded in finding a framework where the classical Lagrangian determines a net of C*-algebras \cite{BF19,BF20}. It is a fully non-perturbative construction. It was formulated for the case of a scalar field with polynomial self-interaction (for the incorporation of fermions, see \cite{BDFR21}). The algebras are generated by unitaries $S(F)$ labeled by local functionals $F$ of the (classical) scalar field configuration. $F$ is interpreted as a local variation of the dynamics, and $S(F)$ is the induced operation on the system (the scattering ``$S$-matrix,'' an interpretation as such is offered at the end of Sect.\ref{sec:tsa}). By definition, these unitaries satisfy a causality relation corresponding to the ordering in time when these operations are performed, and a relation which determines dynamics in terms of the classical Lagrangian and which is a unitary version of the Schwinger-Dyson equation known from perturbative quantum field theory \cite{IZ}. In the case of the free Lagrangian, the canonical commutation relations in form of the Weyl relations are a consequence of the formalism \cite{BF19,BF-time}. 

We discuss in this paper whether further relations can be added in order to enrich the structure and to bring the formalism nearer to standard quantum field theory. We thereby use inspiration from the path integral formulation and check whether the resulting relations are compatible \eg with perturbation theory.
An important relation is obtained by the \emph{field redefinition} (see \eg \cite{MYeats20, KreimerYeats,Balduf} and references therein). Its infinitesimal version appears in a somewhat different form in the quantum master equation of the BV formalism, and was precisely analyzed in renormalized perturbation theory under the name \emph{Master Ward Identity (MWI)} 
in \cite{DB02,DF03,Brennecke08,Due19,Hollands08,FredenhagenR13}. 
We postulate in this paper a unitary version of this identity (Axiom ``Symmetries'' in Sect.~\ref{sec:MWI})-- including anomalies -- and show that, 
in formal perturbation theory, it is essentially equivalent to the infinitesimal version (see Section~\ref{sec:pQFT} and Appendix \ref{app:Delta-X}). 

The crucial ingredient for our formulation is an intrinsic nonperturbative concept of the renormalization group (Definition~\ref{def1}). 
It corresponds to the St\"uckelberg-Petermann renormalization group \cite{BDF09,StuPet53,PopineauS16,DF04}
in formal perturbation theory. Here it is used to characterize possible anomalies of the unitary Master Ward Identity.

We then prove that the new axiom implies the time slice axiom \cite{HaagSchroer} (primitive causality in \cite{Haag64}), which states that each observable can be expressed in terms of observables in any neighbourhood of a Cauchy surface. The remarkable stability of the algebra under a large class of time evolutions 
compares well to a similar property in non-relativistic quantum field theory \cite{Bu} in the framework of Resolvent Algebras by Buchholz and Grundling \cite{BG}.

Moreover, the axiom also characterizes symmetries of the theory.  
One of the deepest structural results of classical physics is the intimate relation between symmetries and conservation laws which were uncovered by Emmy Noether \cite{Noether1,Noether2} about 100 years ago. 
A similar connection also holds in quantum physics, but since the Lagrangian is not directly used in canonical quantisation, it is less evident. In the path integral formulation, the situation looks better but mathematically rigorous results are rare in this formalism. Traditionally in quantum theory, one uses the classical form for the generators of symmetries and attempts finding the corresponding expressions in the quantized theory. There is no general unique procedure to do this and it might fail altogether. For example, in quantum field theory, one then has to rely on renormalized perturbation theory and has to check for anomalies. 

We show that in the absence of anomalies, symmetries can be locally unitarily implemented (a unitary version of Noether's Theorem\footnote{Another unitary version of Noether's Theorem was derived in \cite{BuDoLo} on the basis of the split property \cite{DoLo}.}). We also discuss the occurrence of anomalies and present two equivalent notions of the induced renormalization group flow; one related to anomalies of the Master Ward Identity (see formula \eqref{eq:flow:anomaly} and the discussion below it), and the other (the standard one) expressed in terms of a flow of Lagrangians together with an associated field renormalization (Theorem~\ref{thm:RG-flow}). 
A new, unexpected, result is an anomalous version of Noether's Theorem, Theorem~\ref{theorem:anomalousnoethertheorem}, in which we prove the unitary implementation of symmetries, modulo renormalization group transformations.

It turns out to be useful to generalize the framework from theories on Minkowski space to theories on generic globally hyperbolic spacetimes and to make use of the locally covariant
formulation \cite{BFV03,HW01}, extended in such a way that also interactions can be varied. As an immediate consequence, we find a closed expression for the \emph{algebraic adiabatic limit} by which the net of algebras with an additional interaction can be constructed within the original net. A complication, not present in formal perturbation theory, consists in the possible change of the causal structure by interactions of kinetic type, \ie quadratic functionals of first derivatives of the field (see \eg \cite{BF20}).
\section{Lagrangians, interactions and dynamical spacetimes}\label{sec:frame}
\subsection{Local functionals and observables}
We consider an $n$-component real scalar field $\phi$. The classical configuration space $\Ec(M,\RR^n)$ is the space of smooth functions on the manifold $M$ with values in $\RR^n$. \emph{Manifolds} for us are topological spaces that are connected, Hausdorff and locally homeomorphic to a Euclidean space, together with a smooth differentiable structure (\ie an atlas) and of generic dimension larger than $2$.\footnote{The case of dimension $2$ can be equally treated \emph{mutatis mutandis} on some of the crucial results, \eg those in Appendix~\ref{sec:interpolatingmetrics}.} We equip manifolds with Lagrangians $L=L_0+V_0$, 
which are density-valued local functionals on the configuration space and are of the form
\be\label{eq:L0,V}
L_0(x)[\phi]=\frac12 g^{-1}(d\phi(x),d\phi(x))d\mu_g(x)\ ,\ V_0(x)[\phi]=\hat{V}_0(x,\phi(x),d\phi(x))\ ,
\ee
with a metric $g$ for which the manifold is a globally hyperbolic spacetime\footnote{This entails that our manifolds are also paracompact, see \cite{Mara}.} and such that the linearized classical field equation is normally hyperbolic. Hence, specifying the kinetic part of the Lagrangian equips the manifold $M$ with a metric that makes it into a globally hyperbolic spacetime. Signatures for the metrics are taken as $(+,-,\dots,-)$, $\mu_g$ denotes the density  induced by $g$ and $\hat{V}_0$ is of 1st order with respect to $d\phi$. 
Globally hyperbolic manifolds are always time-orientable, and we choose a fixed time orientation. 

We label the fundamental observables by local functionals $F\in\Floc(M)$ on the configuration space, \ie functionals of the form
\begin{equation}
    F[\phi]=\int \hat{F}(x,j_x(\phi))\ ,
\end{equation}
with a smooth density-valued function $\hat{F}$ on the jet space of $\Ec(M,\RR^n)$ with compact support in $x$. 
A particularly simple local functional is the basic field integrated with a test density $h\in\Dc_{\mathrm{dens}}(M,\RR^n)$:
\be\label{eq:basic-field}
\langle\phi,h\rangle [\phi]\equiv \phi(h)\doteq\int\sum_{i=1}^n\phi_i(x) h_i(x) \ .
\ee

We refer to Appendix \ref{app:gen-field}, where one finds definitions and properties of functionals, including, for instance, the important notion of their support.
Next we want to add to a Lagrangian $L$ the interaction given by a local functional $F$.
In order to do this, we use the fact that a Lagrangian defines a family of local functionals $L(f)$, $f\in\Dc(M,[0,1])$, (named generalized Lagrangian in \cite{BDF09,BFR19}) by
\begin{equation}\label{eq:LagGenField}
    L(f)[\phi]\doteq \int L(x)[f\phi]\ .
\end{equation}
Lagrangians $L',L$ with
\begin{equation}
    \supp (L'(f)-L(f))\subset\supp (f-1)
\end{equation}
lead to the same equation of motion and are considered to be equivalent (notation  $L\sim L'$). The map $f\mapsto L(f)$ is
a generalized field in the sense of \cite{BDF09,BFR19} (again, details are given in Appendix \ref{app:gen-field}). In the following, we shall always mean a Lagrangian $L$ in the sense of a generalized field.

Given a local functional $F$, on the other hand, we can define a generalized field $A_F$ 
by
\begin{equation}
 A_F\,:\,\Dc(M)\to\Floc(M)\,;\,  A_F(f)[\phi]\doteq F[f\phi]\ .
\end{equation}
In Appendix \ref{app:gen-field}, the support of a generalized field is defined and it is shown that $\supp A_F=\supp F$. 
Therefore, a local functional $F$ corresponds  to a generalized field 
$A_F$ with compact support.  The addition of the interaction given by $F\in\Floc(M)$
to a Lagrangian $L$ is now expressed by $L+A_F$.

Next we construct the category $\mathfrak{Loc}$. We first describe its objects: they are \emph{dynamical spacetimes} \ie pairs $(M,L)$ where $M$ is a manifold and $L$ is a Lagrangian of the form \eqref{eq:L0,V}, considered as a generalized field in the form of \eqref{eq:LagGenField}, together with a choice of time orientation.
We consider local functionals $F$ for which the generalized field $A_F$ can be added to the Lagrangian $L$ as an interaction, such that the linearized equation of motion remains normally hyperbolic with respect to a possibly changed metric $g'$.  
This requires that $F$
is of the form 
\be
F[\phi]=F_{g',g}[\phi]+F_0[\phi]
\ee
with
\begin{equation}\label{eq:funct-metric-change}
    F_{g',g}[\phi]=\int\frac12\Bigl((g')^{-1}\bigl(d\phi(x),d\phi(x)\bigr)\,d\mu_{g'}(x)
-g^{-1}\bigl(d\phi(x),d\phi(x)\bigr)\,d\mu_g(x)\Bigr)\ ,
\end{equation}
where $g'$ is a Lorentz metric for which $M$ is globally hyperbolic with $\supp(g'-g)$ compact, 
and $F_0$ is a local functional of the form
\begin{equation}
    F_0[\phi]=\int  \bigl(h_0(x,\phi(x))+\langle h_1(x,\phi(x)),d\phi(x)\rangle\bigr)\, d\mu_g(x)\ ,
\end{equation}
with $h_0,h_1$ smooth  functions on the bundle $M\times \RR^n$, with real values and values in the tangent bundle, respectively, and compactly supported on the base space $M$. We denote the set of these functionals by $\Floc(M,L)$. For later purposes, we denote by $\mathrm{Q}(M,L)$ the subset of quadratic functionals.

For $F\in\Floc(M,L)$, we have $A_F(f)\in\Floc(M,L)$ if $f\equiv1$ on $\supp F$, hence the equivalence class of 
$L+A_F$ and the associated metric are well defined. The time orientation is obtained from the time orientation of the unperturbed Lagrangian by continuity. 
Notice that $\Floc(M,L)$ is \emph{not} a linear space, but has the property that for any $F\in\Floc(M,L)$
\begin{equation}\label{eq:Floc-add}
 F+G\in\Floc(M,L)\Longleftrightarrow G\in\Floc(M,L+A_F)\ .
\end{equation}
Moreover, in case $\supp F\cap \supp G=\0$, 
we have that with $F$ and $G$ also $F+G\in\Floc(M,L)$,
namely, if $g_1$ is the metric associated to $L+A_F$ and $g_2$ the metric associated to $L+A_G$, then $g_{12}\doteq g_1+g_2-g$ is the metric associated to $L+A_{F+G}$ which is easily seen to be also globally hyperbolic.

Complications arise from the fact that the set of Lorentz metrics is, in general, not convex, hence $A_F(f)$ for $0\le f\le1$ might not belong to $\Floc(M,L)$, for $F_{g',g}$ of the form in \eqref{eq:funct-metric-change} with $g\neq g'$. In Appendix~\ref{sec:interpolatingmetrics}, we solve this problem by introducing additional metrics $g_i$, $i=0,\dots,4$, with $g_0=g$, $g_4=g'$ such that, for each pointwise convex combination
\begin{equation}\label{eq:g-convex-comb}
    g_{\lambda}(x)=\lambda(x)g_{i-1}(x)+(1-\lambda(x))g_{i}(x)\ ,\ \lambda\in\mathcal{C}^{\infty}(M)\ ,\  0\le\lambda\le1 
\end{equation}
of $g_{i-1}$ and $g_i$, $i=1,\dots,4$, the manifold $M$ is globally hyperbolic.  We decompose
\begin{equation}\label{eq:partition(F)}
    F\equiv F_{g',g}+F_0=\sum_{i=0}^4 F_i
\end{equation}
with $F_0$ as above and $F_i\doteq F_{g_{i},g_{i-1}}$ for $i=1,\dots,4$. 
We thus get a quintuple of generalized fields $A_{F_i}$, $i=0,\dots,4$ with $A_{F_i}(f_i)\in\Floc(M,L+\sum_{j<i}A_{F_j})$ and $\sum_i A_{F_i}(f_i)\in\Floc(M,L)$ provided $f_i\in\Dc(M,[0,1])$,  with $\supp f_k\subset (f_{k-1})^{-1}(1)$, $k=1,\dots,4$.

Let $\Dc_s(M)$ denote the set of finite sequences $f=(f_1,\dots,f_n)$ of test functions $f_i\in \Dc(M,[0,1])$ subject to restrictions as above, and $f\equiv 1$ on a set $N$ means that $f_i\equiv1$ on $N$, $i=1\dots,n$.

\begin{definition}[\textbf{Interactions}] The set of finite sequences $V=(V_1,\dots,V_n)$ of generalized fields $V_i$ whose values, upon evaluation with test functions in $\Dc(M,[0,1])$, 
are local functionals in $\Floc(M,L+\sum_{j<i} V_j)$, shall be called \textbf{interactions} w.r.t the Lagrangian $L$ and denoted by the symbol $\mathrm{Int}(M,L)$. We use the notation $V(f)=\sum_iV_i(f_i)$, $f=(f_1,\dots,f_n)\in\Dc_s(M)$ for the local functionals which approximate $V$ for $f\equiv 1$ on a sufficiently large region.
\end{definition}
In the following we always understand Lagrangians as generalized fields in the sense that they might depend on a finite sequence $f\in\Dc_s(M)$.
For example, for $V=(V_1,\dots,V_n)\in\mathrm{Int}(M,L)$, $L+V$ is evaluated with an $(n+1)$-tuple $(f_0,f_1,\ldots,f_n)\in\Dc_s(M)$, and
$$(L+V)(f_0,\ldots,f_n)=L(f_0)+\sum_{i=1}^nV_i(f_i)\ .$$

We now turn to the morphisms. Morphisms of $\mathfrak{Loc}$, $\iota\in\mathrm{Hom}((M,L),(M',L'))$, are compositions of elementary morphisms of either of the following forms:
\begin{itemize}
\item $\iota_{\rho}$: Structure preserving embeddings, \ie smooth embeddings $\rho:M\to M'$  with $\rho_{\ast}L=L'$ which preserve the time orientation and with causally convex image $\rho(M)$. 

\noindent Here $(\rho_{\ast}L)(f)\doteq \rho_{\ast}(L(f\circ\rho))$, $f\in\Dc_s(M')$ and $\rho_{\ast}F[\phi]\doteq F[\phi\circ\rho]$,  $F\in\Floc(M,L)$.
\item $\iota_{\Phi}$: Affine field redefinitions $\Phi:\phi\mapsto \phi A+\phi_0$ with $(\phi A)_j(x)=\sum_{i=1}^n\phi_i(x)A_{ij}(x)$, $A(x)\in\mathrm{GL}(n,\RR)$,  where $A$ and $\phi_0$ are smooth functions of compact support, $M'=M$, $L'=\Phi_{\ast}L$.\footnote{More general field redefinitions could be considered \cite{KreimerYeats,MYeats20,Balduf}. We restrict ourselves here to affine field redefinitions, since they, together with their inverses, map polynomial functionals to polynomial functionals (this is important for the perturbative proof of the anomalous Master Ward Identity, see \eg Section \ref{sec:pQFT}); moreover, they map the set of Lagrangians of 2nd order in the field into itself, thus allowing analytic control.} 

\noindent Here $(\Phi_{\ast}L)(f)\doteq \Phi_{\ast}(L(f))$, $f\in\Dc_s(M)$ and $(\Phi_{\ast}F)[\phi]\doteq F[\Phi(\phi)]$, $F\in\Floc(M,L)$. 
%
\item $\iota_{V,+}$: Retarded interaction $M'=M$, $L'+V=L$, where $V\in\mathrm{Int}(M',L')$ with past compact support \cite{Ko13}.

\item $\iota_{V,-}$: Advanced interaction $M'=M$, $L'+V=L$, where $V\in\mathrm{Int}(M',L')$ with future compact support \cite{Ko13}.  
\end{itemize}
\subsection{Dynamical algebras} 
The associated quantum field theory is now a functor $\mathfrak{A}$ from $\mathfrak{Loc}$ to the category $\mathfrak{C}^*$ of unital C*-algebras with unital homomorphisms as arrows.
The \emph{dynamical algebra} $\mathfrak{A}(M,L)$ is a C*-algebra freely generated by unitaries $S_{(M,L)}(F)$, $F\in\Floc(M,L)$ with $S_{(M,L)}(c)=e^{ic}1$ for constant functionals $c$, $c\in\RR$,
modulo the following relations:
\begin{axiom}[\textbf{Causality Relation}] Let $G \in\Floc(M,L)$ and $F,H\in\Floc(M,L+A_G)$. Then
\be\label{Caus}
S_{(M,L)}(F+G+H)=S_{(M,L)}(F+G)S_{(M,L)}(G)^{-1}S_{(M,L)}(G+H)
\ee 
when $\supp F\cap J_-^{L+A_G}(\supp H)=\0$ where $J_-^{L+A_G}$ denotes the causal past with respect to the metric induced by $L+A_G$, and
\end{axiom}
\begin{axiom}[\textbf{Dynamical Relation}] For all $F\in\Floc(M,L)$ we require
\be\label{SD}
S_{(M,L)}(F)=S_{(M,L)}(F^\psi+\delta L(\psi))\ ,\quad  \psi\in\Dc(M,\RR^n)\ ,
\ee
where 
$$
F^\psi[\phi]\doteq F[\phi+\psi]\ ,\quad\delta L(\psi)\doteq L(f)^\psi-L(f)\ ,
$$
for any $f\in\Dc_s(M)$ satisfying $f \equiv 1$ on $\supp \psi$. This is the on-shell version of the Schwinger-Dyson relation in \cite{BF19}.
\end{axiom}
\begin{remark}
%
%
For the Dynamical Relation note that $\delta L(\psi)\in\Floc(M,L)$, and that with
$F\in\Floc(M,L)$ also $F^\psi\in\Floc(M,L)$ and $F^\psi+\delta L(\psi)\in\Floc(M,L)$; we wish also to point out that
$\supp\delta L(\psi)\subseteq\supp\psi$.

The ``on-shell algebra'' is distinguished from the ``off-shell algebra'' by the validity of the additional relation $S_{(M,L)}(\delta L(\psi))=1$ 
for all $\psi\in\Dc(M,\RR^n)$; in perturbation theory
this terminology agrees with the usual distinction between on-shell and off-shell time-ordered products, see \cite[Sect.~7]{BDFR21}. 
\end{remark}

The morphisms $\fA\iota_{\bullet}\equiv\alpha_{\bullet}$ associated to the various elementary morphisms are mono\-morphisms $\fA(M,L)\to\fA(M',L')$ which act on the generators of the algebra as
\begin{align}
\alpha_{\rho}(S_{(M,L)}(F))&=S_{(M',L')}(\rho_\ast  F)\ ,\\ 
 \alpha_{\Phi}(S_{(M,L)}(F))&=S_{(M',L')}(\Phi_\ast  F)\ ,\\ 
\label{Bog1}
\alpha_{V,+}(S_{(M,L)}(F))&=S_{(M',L')}(V(f))^{-1}S_{(M',L')}(F+V(f))\ ,\\
\label{Bog2}
\alpha_{V,-}(S_{(M,L)}(F))&=S_{(M',L')}(F+V(f))S_{(M',L')}(V(f))^{-1}\ .
\end{align}
In the last two equations $f\in\Dc_s(M)$
%
with $f\equiv 1$ on
a neighborhood of $J_-^L(\supp F)\cap\supp V$ for \eqref{Bog1} and of  $J_+^L(\supp F)\cap\supp V$ for \eqref{Bog2}. 
Due to the Causality Relation, $\alpha_{V,\pm}(S_{(M,L)}(F))$ do not depend on the 
remaining freedom in the choice of $f=(f_1,\dots,f_n)$. Namely, \eg for the retarded case, let $f_j'\in\Dc(M,[0,1])$ with $\supp f_j'\subset (f_{j-1})^{-1}(1)$ and $f_j'\equiv 1$ on $\supp f_{j+1}$ where we put $f_0\equiv 1$ and $f_{n+1}=0$, and  let $f'$ denote the sequence $f$ with $f_j$ replaced by $f_j'$. Then $\supp(V(f)-V(f'))\cap J_-^L(\supp F)\cap\supp V=\0$ and 
\begin{equation}\label{eq:Bog}
    S_{(M',L')}(V(f')+F)=S_{(M',L')}(V(f'))S_{(M',L')}(V(f))^{-1}S_{(M',L')}(V(f)+F)\ .
\end{equation}

We point out that $\alpha_\Phi$ and $\alpha_{V,\pm}$ are even surjective, that is, they are isomorphisms. 
For the latter two, there is a simple formula for the inverse, namely:
$$
(\iota_{V,\pm})^{-1}=\iota_{(-V),\pm}\ ,\quad \text{hence}\quad(\alpha_{V,\pm})^{-1}=\alpha_{(-V),\pm}\ ,
$$
by using that $(-V)\in\mathrm{Int}(M,L)$.

\begin{remark}\label{rm:Caus}
We return to the Causality Relation \eqref{Caus}: although it is a defining relation for $\fA(M,L)$, it has to be 
interpreted as ``causality'' in the dynamical spacetime $(M,L+A_G)$, {\it i.e.}, it is required that $\supp F$ is later than $\supp H$
with respect to the metric given by $L+A_G$; see \cite{BF20}. 

In causal perturbation theory \cite{EG73,BF00,Due19}, the factorization \eqref{Caus} is required for the $S$-matrix (describing perturbations of the free theory) w.r.t.~the causality structure of the \emph{free} theory; actually, in
that framework this factorization property is equivalent to its special case obtained by setting $G=0$.

In the nonperturbative framework at hand, this equivalence does not hold. However,
the Causality Relation \eqref{Caus} is equivalent to the just mentioned simpler 
causal factorization in dynamical spacetimes $(M,L+A_G)$, explicitly
\be\label{eq:Caus1}
S_{(M,L+A_G)}(F+H)=S_{(M,L+A_G)}(F)\,S_{(M,L+A_G)}(H)
\ee
if $\supp F\cap J_-^{L+A_G}(\supp H)=\0$, where $F,H\in\Floc(M,L+A_G)$. Namely, applying 
$\alpha_{A_G,+}$ to \eqref{eq:Caus1}, we indeed obtain \eqref{Caus}.
\end{remark}

We check that all these maps $\alpha_{\bullet}$
preserve the defining algebraic relations. For $\alpha_{\chi}$ this is the standard situation for local covariance. 
In detail, by using $(\rho_\ast F)^\psi=\rho_\ast(F^{\psi\circ\rho})$ (with $\psi\in\Dc(M',\RR^n)$) we obtain
\be
S_{(M',L')}\bigl((\rho_\ast F)^\psi+\delta L'(\psi)\bigr)=\alpha_\rho\bigl(S_{(M,L)}\bigl(F^{\psi\circ\rho}+\delta L(\psi\circ\rho)\bigr)\bigr)
=\alpha_\rho\bigl(S_{(M,L)}(F)\bigr)=S_{(M',L')}(\rho_\ast F)\,;
\ee
and for the Causality Relation the claim follows from the equivalence
\be\label{eq:J0}
\supp F\cap J_-^{L+A_G}(\supp H)=\0\quad\Leftrightarrow\quad\supp(\rho_\ast F)\cap J_-^{\rho_\ast L+A_{\rho_\ast G}}(\supp(\rho_\ast H))=\0\ .
\ee
The definition of $\alpha_{V,\pm}$ is just the Bogoliubov formula; for the Dynamical Relation the claim relies on
\begin{align}
S_{(M',L')}\bigl(V(f)\bigr)\,\alpha_{V,+}\bigl(S_{(M,L)}\bigl(F^\psi+\delta L(\psi)\bigr)\bigr)&=S_{(M',L')}\bigl(F^\psi+\delta L(\psi)+V(f)\bigr)
\nonumber\\
&=S_{(M',L')}\bigl((F+V(f))^\psi+\delta L'(\psi)\bigr)
\end{align}
and for the Causality Relation on
\be\label{eq:J1}
J_-^{L'+A_{G+V(f)}}=J_-^{L'+V+A_G}=J_-^{L+A_G}\ .
\ee
For $\alpha_{\Phi}$ we only have to check the Dynamical Relation, since $\supp (\Phi_{\ast}F)=\supp F$ and
\be\label{eq:J2}
J_-^{\Phi_\ast L+A_{\Phi_\ast G}}=J_-^{L+A_G}\ .
\ee
We use the formula
\be
(\Phi_{\ast}F)^{\psi'}[\phi]=F[(\phi+\psi')A+\phi_0]=(\Phi_{\ast}F^{\psi})[\phi]
\ee
with $\psi= \psi'A$, and the corresponding formula for $L$ and $L'=\Phi_{\ast}L$, and find 
\be
\alpha_{\Phi}\bigl(S_{(M,L)}(F^{\psi}+\delta L(\psi))\bigr)=
S_{(M',L')}\bigl(\Phi_{\ast}(F^{\psi}+\delta L(\psi))\bigr)
=S_{(M',L')}\bigl((\Phi_{\ast}F)^{\psi'}+\delta L'(\psi')\bigr)\ .
\ee

\noindent Due to the Dynamical Relation \eqref{SD} in $\fA(M,L)$, the l.h.s.~is equal to 
\be
\alpha_{\Phi}(S_{(M,L)}(F))=S_{(M',L')}(\Phi_{\ast}F)
\ee
so we see that $\alpha_\Phi$ 
maps this relation in $\fA(M,L)$ to the same relation in $\fA(M',L')$.

\begin{remark}[Net Structure] Given any dynamical spacetime $(M,L)$, we can restrict the functor $\fA$ to relatively compact, causally convex subregions $\Oc\subset M$ and the arrows to inclusions $\Oc_1\subset\Oc_2$. This  
subfunctor is the Haag-Kastler net associated to $(M,L)$ and will be denoted by $\fA_{(M,L)}$.
An isomorphism $\alpha:\fA(M,L)\to\fA(M',L')$ is called a \emph{net isomorphism} if it extends to an equivalence between the corresponding functors, \ie it extends to a \emph{bijective natural transformation}. In detail, for $\Oc_1\subset\Oc_2$ as
above, let $\iota_{(M,L)}^{\Oc_2,\Oc_1}:\fA_{(M,L)}(\Oc_1)\hookrightarrow\fA_{(M,L)}(\Oc_2)$ be the pertinent embedding of algebras.
That $\alpha$ extends to a bijective natural transformation means that there exists a unique bijection $\tilde\alpha:M\to M'$
fulfilling $\fA_{(M',L')}\bigl(\tilde\alpha(\Oc)\bigr)=\alpha\bigl(\fA_{(M,L)}(\Oc)\bigr)$ for all causally convex $\Oc\subset M$ 
and that
\be
\alpha\circ\iota_{(M,L)}^{\Oc_2,\Oc_1}=\iota_{(M',L')}^{\tilde\alpha(\Oc_2),\tilde\alpha(\Oc_1)}\circ\alpha \ .
\ee
\end{remark}
\section{Algebraic adiabatic limit}
A problem with Bogoliubov's formulae (\ref{Bog1},\ref{Bog2}) is that they do not apply for an interaction without support restriction. As remarked earlier \cite{IS78,BF00} this does not matter for the algebraic structure of the algebra of local observables for which it is sufficient to fix the interaction in a sufficiently large, however bounded, 
region (algebraic adiabatic limit), and it was also observed that the algebra of the interacting theory can be identified with the algebra of the free theory \cite{BF19}. The formalism described above allows an elegant formula for this identification 
by writing the interaction as a sum of a past compact and a future compact part and composing the arrows for retarded and advanced interactions. Namely we find:

\begin{proposition}
Let $V_+\in\mathrm{Int}(M,L)$ have past compact support and $V_-\in\mathrm{Int}(M,L+V_+)$ future compact support (for the respective causal structures).  Given the isomorphisms $\alpha_{V_-,-}:\fA(M,L+V)\to\fA(M,L+V_+)$ and $\alpha_{V_+,+}:\fA(M,L+V_+)\to\fA(M,L)$, with $V=V_++V_-$, then
\begin{equation}\label{eq:isochangeint}
    \alpha_{V_+,V_-}\doteq \alpha_{V_+,+}\circ\alpha_{V_-,-}
\end{equation}
defines an isomorphism from $\fA(M,L+V)$ to $\fA(M,L)$. 
It acts 
on the generators of the algebra as 
\be
\begin{split}
    &\alpha_{V_+,V_-}(S_{(M,L+V)}(F))\\
&=S_{(M,L)}(V_+(f))^{-1}S_{(M,L)}(F+V_+(f)+V_-(g))S_{(M,L)}(V_-(g)+V_+(f))^{-1}S_{(M,L)}(V_+(f))\  \label{eq:alpha-V}
\end{split}
\ee
where $g,f\in\Dc_s(M)$ with $g\equiv 1$ on a neighborhood of $J^{L+V}_+(\supp F)\cap\supp V_-$ and 
$f\equiv1$ on a neighborhood of $J^{L+V_+}_-\bigl(\supp F\cup\supp V_-(g)\bigr)\cap\supp V_+$.
The explicit expression \eqref{eq:alpha-V} for $\alpha_{V_+,V_-}(S(F))$ does not depend on the choices of $f$ and $g$.
\end{proposition}


\begin{proof}
We shall use the abbreviation $S=S_{(M,L)}$ for simplicity.
The formula \eqref{eq:alpha-V} is obtained straightforwardly by inserting the definitions \eqref{Bog1} and \eqref{Bog2}
of $\alpha_{V_+,+}$ and $\alpha_{V_-,-}$, respectively. Also the independence of the choice of $f$ and $g$ under the given conditions derives from the corresponding facts about the retarded and advanced isomorphisms.  
Nevertheless it is instructive to check the independence directly. 

By iteration of the argument given before equation \eqref{eq:Bog}, let $f'$ be the appropriate change for $f$ satisfying the conditions above for fixed $g$. Then we have that $\supp(V_+(f')-V_+(f))$ does not intersect $J^{L+V_+}_-\bigl(\supp F\cup\supp V_-(g)\bigr)$, 
and by causal factorization we get
\be
S(F+V_+(f') +V_-(g))=S(V_+(f'))S(V_+(f))^{-1}S(F+V_+(f) +V_-(g))
\ee
and
\be
S(V_+(f') +V_-(g))^{-1}=S(V_+(f) +V_-(g))^{-1}S(V_+(f))S(V_+(f'))^{-1}\ ,
\ee
hence
\be
    S(V_+(f'))^{-1}S(F+V_+(f') +V_-(g))S(V_+(f')+V_-(g))^{-1}S(V_+(f'))
    =\alpha_{V_+,V_-}(S_{(M,L+V)}(F))\ .
\ee
Let now $g'$ be another choice for $g$ satisfying $g'\equiv1$ on a neighborhood of $J^{L+V}_+(\supp F)\cap V_-$. 
Then $\supp(V_-(g')-V_-(g))$ does not intersect $J^{L+V}_+(\supp F)$ and we obtain
\be
 S(F+V_+(f) +V_-(g'))S(V_+(f) +V_-(g'))^{-1}=S(F+V_+(f) +V_-(g))S(V_+(f) +V_-(g))^{-1}\ ,   
\ee
thus again $\alpha_{V_+,V_-}(S_{(M,L+V)}(F))$ does not change.

To prove the joint independence of the choices we use an intermediate choice $f''$ for $f$ which satisfies the conditions above for both $g$ and $g'$. Then we first replace $f$ by $f''$, then $g$ by $g'$ and finally $f''$ by $f'$.
\end{proof}

The isomorphism $\alpha_{V_+,V_-}$ which embeds the original algebra (with Lagrangian $L+V$) into the algebra with Lagrangian $L$ depends, however, on the split $V=V_++V_-$.

This isomorphism changes the association of local subalgebras. But on the intersection $\supp V_+\cap\supp V_-$ the local subalgebras are only slightly deformed, as can be seen from the
conditions on the test functions $f$ and $g$ in \eqref{eq:alpha-V}. In particular, if $N$ is a globally hyperbolic neighbourhood of $\supp V_+\cap\supp V_-$, then
\begin{equation}\label{eq:interactionchange}
    \alpha_{V_+,V_-}(\fA(N,(L+V)\!\restriction_N))=\fA(N,L\!\restriction_N)
\end{equation}
where $L\!\restriction_{N}$ denotes the restriction of $f\mapsto L(f)$ to $\Dc_s(N)$. (To simplify the notation, this restriction is understood without mentioning
at some places.)

This observation allows an easy proof of the validity of a weak form of the time slice axiom. 
\begin{definition}[\textbf{Time slice in representation}]
We say that $\fA_{(M,L)}$ satisfies, in the representation $\pi$, the time slice axiom  for some Cauchy surface $\Sigma$ of $M$ whenever for any causally convex and globally hyperbolic neighbourhood $N$ of $\Sigma$ we have
\be\label{eq:timesliceaxiominrep}
\pi(\fA(N,L\!\restriction_{N}))=\pi(\fA(M,L))\ .
\ee 
\end{definition}
 Then it holds that 
\begin{theorem}\label{thm2} Let $\fA_{(M,L)}$ satisfy, in the representation $\pi$, the time slice axiom for some Cauchy surface $\Sigma$ which is also a Cauchy surface for the metric associated to $L+V$, and  let $\alpha\doteq\alpha_{V_+,V_-}: \fA(M,L+V)\to\fA(M,L)$  denote the isomorphism in \eqref{eq:isochangeint} for some decomposition $V=V_++V_-$. Then $\fA_{(M,L+V)}$  satisfies, in the representation $\pi\circ\alpha$, the time slice axiom  for  $\Sigma$.
%
\end{theorem}
\begin{proof} Let $N$ be a neighborhood of $\Sigma$. Choose a decomposition $V=V_+'+V_-'$ with 
\be\supp V_+'\cap\supp V_-'\subset N\ .\ee 
Then the induced isomorphism $\alpha'\doteq\alpha_{V_+',V_-'}$ differs from $\alpha$ by an inner automorphism 
$\mathrm{Ad}(U)(\bullet)\doteq U\bullet U^{-1}$, $\alpha=\alpha'\circ\mathrm{Ad}(U)$, with a unitary $U\in\fA(M,L+V)$.  
From the above discussion the following chain of identities holds true
\begin{equation}\label{eq:time-slice}
\begin{split}
    \pi\circ\alpha(\fA(N,L+V))&=\pi\circ\alpha'\circ\mathrm{Ad}(U)(\fA(N,L+V))\,\\
    &=\,\mathrm{Ad}(\pi\circ\alpha'(U))\circ\pi\circ\alpha'(\fA(N,L+V))\\
    &\stackrel{\mathclap{\eqref{eq:interactionchange}}}{=}\, \mathrm{Ad}(\pi\circ\alpha'(U))\circ\pi(\fA(N,L))\\
    &\stackrel{\mathclap{\eqref{eq:timesliceaxiominrep}}}{=}\,\mathrm{Ad}(\pi\circ\alpha'(U))\circ\pi(\fA(M,L))\\
    &\stackrel{\mathclap{\eqref{eq:isochangeint}}}{=}\, \mathrm{Ad}(\pi\circ\alpha'(U))\circ\pi\circ\alpha'(\fA(M,L+V))\\
    &=\,\pi\circ\alpha'(\mathrm{Ad}(U)(\fA(M,L+V)))\\
    &=\,\pi\circ\alpha(\fA(M,L+V))\ ,
\end{split}
\end{equation}
by which the theorem is proven.
\end{proof}

\begin{remark}
For the following particular case Theorem \ref{thm2} has already been proven in \cite{CF09}.
Looking at perturbation theory in the Fock representation,
let $L_0$ be the free Lagrangian (which may include a mass term). Then $\fA(M,L_0)$ is the algebra
consisting of Wick products of free fields; it is generated by the $S$-matrices 
$S_{(M,L_0)}(F), \, F\in\Floc(M,L_0)$, where $S_{(M,L_0)}(F)$ is the generating functional of the time-ordered products of the interaction $F$, fixed by suitable renormalization conditions. 
The time slice axiom holds in $\fA_{(M,L_0)}$, as shown in the first part of \cite{CF09}.

Given some Cauchy surface $\Sigma$ of $(M,L_0)$, the interacting fields in the future of $\Sigma$ depend only on the interaction 
in the future of $\Sigma$, up to unitary equivalence (see the algebraic adiabatic limit in \cite{BF00}).
We consider therefore an interaction $V\in\mathrm{Int}(M,L_0)$  with past compact support, 
hence $\alpha_{V,+}$ 
\eqref{Bog1} is an isomorphism from the interacting algebra $\fA(M,L_0+V)$ into the free algebra $\fA(M,L_0)$. 
For $G\in\Floc(M,L_0+V)$, we introduce 
\be\label{eq:intfield}
G_V\doteq\frac{d}{d\lambda}\Big\vert_{\lambda=0}\alpha_{V,+}(S_{(M,L_0+V)}(\lambda G))\in \fA(M,L_0);
\ee
$G_V$ is the interacting field corresponding to $G$ ({\it i.e.},~it agrees with $G$ for $V=0$), expressed in terms of the free theory. 
The relative $S$-matrix
$$
S_V(G)\doteq\alpha_{V,+}(S_{(M,L_0+V)}(G))=S_{(M,L_0)}(V(f))^{-1}\,S_{(M,L_0)}(G+V(f))\in \fA(M,L_0)\ ,
$$
(where still $G\in\Floc(M,L_0+V)$)
is the generating functional of the time-ordered products of the interacting fields $G_V$. 

By means of the Causality Relation \eqref{Caus}
for $S_V$, it is proven in the same reference -- in the second step -- that the validity of the time slice axiom in $\fA(M,L_0)$ implies 
that this axiom holds also for the net of algebras
$$
M\supset\Oc\mapsto\fA_V(\Oc)\doteq\bigl[\{S_V(G)\,\big\vert\,G\in\Floc(\Oc,L_0+V)\}\bigr] 
$$
(where $[\cdots]$ means the algebra generated by the elements of the indicated set). Actually, the proof of this second step
in \cite{CF09} is not limited to perturbation theory and agrees essentially with the proof of Theorem \ref{thm2} given above. 
Finally, the time slice axiom holds then also for the net 
$$
\Oc\mapsto\alpha_{V,+}^{-1}\bigl(\fA_V(\Oc)\bigr)=\fA(\Oc,L_0+V)=\bigl[\{S_{(M,L_0+V)}(G)\,\big\vert\,G\in\Floc(\Oc,L_0+V)\}\bigr]\ ,
$$ 
\ie, in $\fA_{(M,L_0+V)}$.
\end{remark}
Due to possible changes of the causal structure by kinetic terms in the interaction, the splitting into the past and future compact part is not always possible.
The isomorphy between algebras with different Lagrangians, however,  holds in general.
\begin{theorem}\label{thm:adiabatic}
For every interaction $V\in\mathrm{Int}(M,L)$ the algebras $\fA(M,L+V)$ and $\fA(M,L)$ are isomorphic.
\end{theorem}
\begin{proof}
By using the construction of Appendix B we decompose $V=\sum_{i=1}^5V_i$ such that $V_i\in\mathrm{Int}(M,L+\sum_{j<i}V_j)$ and such that the spacetimes corresponding to $L+\sum_{j<i}V_j$ and $L+\sum_{j\le i}V_j$ 
possess a joint foliation by Cauchy surfaces. For each $i=1,\dots,5$ we then decompose  $V_i=V_i^++V_i^-$ with $V_i^+\in\mathrm{Int}(M,L+\sum_{j<i}V_j)$ such that $\supp V_i^+$
is past compact (with respect to $L+\sum_{j<i}V_j$) and with $V_i^-\in\mathrm{Int}(M,L+\sum_{j<i}V_j+V_i^+)$ such that
$\supp V_i^-$ is future compact (with respect to $L+\sum_{j<i}V_j+V_i^+$). Then
\begin{equation}
    \alpha\doteq\prod_{i=1}^5\alpha_{V_i^+,+}\alpha_{V_i^-,-}
\end{equation}
is an isomorphism from $\fA(M,L+V)$ to $\fA(M,L)$.
\end{proof}
\begin{remark} The reader should not feel worried by this result that entails that the \emph{global} algebras of free and interacting theories are isomorphic. What matters is the locality structure, \ie the \emph{net structure}, and the isomorphism does not extend to a \emph{net isomorphism}, 
\ie the sub-algebras related to bounded regions are not mapped bijectively to each other, thus the theories are not equivalent.
\end{remark}

\section{Renormalization group}\label{sec:RG}

In perturbation theory, the axiom Dynamical Relation is equivalent to the renormalization condition ``Field Equation''
\cite[Sect.~7]{BDFR21}. But in contrast to the classical theory, in perturbative quantum field theory the field equation alone does not in general completely fix the dynamical evolution. There remains some freedom which has to be fixed by additional conditions (``renormalization conditions''). They can be classified in terms of the St\"uckelberg-Petermann renormalization group,
see \cite{StuPet53,PopineauS16,DF04,BDF09}.

\begin{example} We illustrate this problem on a simple example.
We assume that the free massive scalar field on $4$-dimensional Minkowski space is perturbed by an interaction $\phi^2(f)$. 
We use the axioms for the time-ordered (or retarded) products
given in \cite{Due19}. Then,
the Field Equation determines uniquely the time evolution of the interacting basic field (with derivatives), 
\eg $(\square\phi)_{\phi^2(f)}(x)$, because only tree diagrams contribute. But in general this does not hold for the
interacting \emph{composite} fields, e.g., $(\phi\square\phi)_{\phi^2(f)}(x)$, because of the contributions of divergent loop diagrams whose renormalization is unique only up to a real constant.
In this example, such indeterminacy can be removed by requiring the Master Ward Identity (MWI)
as an additional renormalization condition:%
\footnote{Note that there is no freedom of renormalization for the composite field $\phi^2_{\phi^2(f)}(x)$. Moreover, in classical 
field theory \eqref{eq:MWI:phi} is satisfied, as it is obtained by the pointwise product of $\phi^{\mathrm{class}}_{\phi^2(f)}(x)$ with
the Field Equation $(\square+m^2)\phi^{\mathrm{class}}_{\phi^2(f)}(x)=-2f(x)\,\phi^{\mathrm{class}}_{\phi^2(f)}(x)$.}
\begin{equation}\label{eq:MWI:phi}
    (\phi\square\phi)_{\phi^2(f)}(x)=-(m^2+2f(x))\,\phi^2_{\phi^2(f)}(x)\ ,
\end{equation}
which has the same form as the identity in the classical theory. 
\end{example}

In more complicated cases, the classical identities cannot always be preserved, and this is how the \emph{anomalies} arise.

Similarly to the perturbative case, also in the nonperturbative framework of this paper,
the axioms given so far do not yet fix completely the dynamics, that is, the time slice axiom does not hold for the algebra they define.
There exists a large group of isomorphisms, which prevent the existence of a dynamical law. Actually, the axiom Dynamical Relation fixes the dynamics (for the free Lagrangian) only for the subalgebra generated by $S$-matrices $S(F)$ with affine functionals $F$. For a given Lagrangian $L$ we consider the following group. 

\begin{definition}\label{def1} 
The \textbf{renormalization group} $\mathcal{R}(M,L)$ for a Lagrangian $L=L_0+V$, with $V\in\mathrm{Int}(M,L_0)$,
is the set of all bijections $Z$ of $\Floc(M,L)$ which satisfy the following conditions:
\begin{enumerate}
    \item[$(i)$] (Compact support) The support of $Z$
    \begin{align}\label{eq:suppZ}
&\supp Z\doteq \{x\in M\,|\, \text{ for every neighborhood }U\ni x\,\text{ there exist}\ G\in\Floc(M,L),\nonumber \\
              &\,F\in\Floc(M,L+A_G)\text{ with }\supp F\subset U\text{ such that } Z(F+G)\neq F+Z(G)\},
\end{align}

is compact.
\item[$(ii)$] (Locality)
Let $G\in\Floc(M,L)$ and $F,H\in\Floc(M,L+A_G)$ with the requirement that $\supp F\cap\supp H=\0$. Then
\be
Z(F+G+H)=Z(F+G)-Z(G)+Z(G+H)\ .
\ee
\item[$(iii)$] (Dynamics)
$Z$ preserves the dynamics, \ie\be
Z(F^\psi+\delta L(\psi))=Z(F)^{\psi}+\delta L(\psi)\ ,\ \psi\in\Dc(M,\RR^n)\ .
\ee
\item[$(iv)$] (Field shift)
Under shifts in configuration space, $Z$ transforms as
\begin{equation}\label{eq:shift-Z}
    Z(F^{\psi}-V(f))=Z(F-V(f))^{\psi}+\delta V(\psi)
\end{equation} 
with $\psi\in\Dc(M,\RR^n)$, $f\in\Dc_s(M)$,  $f\equiv 1\textrm{ on }\supp\psi$,
and $F\in\Floc(M,L-V)$ .
\item[$(v)$] (Causal Stability)
$Z$ does not change the causal structure, \ie
\begin{equation}
    J_{\pm}^{L+A_{Z(F)}}=J_{\pm}^{L+A_F}\ .
\end{equation}
\end{enumerate}
\end{definition}
For the free Lagrangian (\ie $L=L_0$), the condition $(iv)$ takes the simpler form
\be\label{eq:shift-RL0}
Z(F^\psi)=Z(F)^{\psi}\ ,\quad Z\in\Rc(M,L_0)\ .
\ee

\begin{remark}\label{rm:SP}
For the free Lagrangian $L_0$ the group $\Rc(M,L_0)$ can be compared with the St\"uckel\-berg-Petermann renormalization group
$\Rc_0$ as defined in \cite{DF04,BDF09,Due19}, see Appendix \ref{app:Delta-X}. 
The conditions  $(ii)$-$(iv)$ given above appear also in the definition of $\Rc_0$, but there is a main difference: 
the elements $Z_0$ of $\Rc_0$  may have 
non-compact support, which allows, \eg\!\!, to impose  translation invariance in Minkowski space as a renormalization condition.
\end{remark}

\begin{proposition}\label{prop:inv-suppZ}
Let $Z$ be in $\Rc(M,L)$, then it preserves the support of functionals in the sense that  for all $G\in\Floc(M,L)$ and $F\in\Floc(M,L+A_G)$ it holds
\be\label{eq:invar-suppZ} 
\supp \bigl(Z(F+G)-Z(G)\bigr)=\supp F \ .
\ee
\end{proposition}
\begin{proof} Let $\psi\in\Dc$ with $\supp\psi\cap\supp F=\0$. The assertion follows from
\be Z(F+G)^{\psi}-Z(G)^{\psi}=Z(F+G)-Z(G)\ ,
\ee
which can be verified as follows: let $f\equiv1$ on $\supp \psi$. Then by the conditions $(ii)$ (Locality) and $(iv)$ (Field shift)
\be
\begin{split}
Z(F+G)^{\psi}&=Z(F+G+V(f)-V(f))^{\psi}\\
&=Z((F+G)^{\psi}+\delta V(\psi))-\delta V(\psi)\\
&=Z(F+G+(G^{\psi}-G+\delta V(\psi))-\delta V(\psi)\\
&=Z(F+G)-Z(G)+Z(G^{\psi}+\delta V(\psi))-\delta V(\psi)\\
&=Z(F+G)-Z(G)+Z(G)^{\psi}\ .
\end{split}
\ee
\end{proof}

The renormalization group for a Lagrangian $L+W$ is obtained from $\mathcal{R}(M,L)$ by the following definition:
\begin{definition}\label{def:ZW}
Let $Z\in\Rc(M,L)$ and let $W\in\mathrm{Int}(M,L)$.
Then, for $F\in\Floc(M,L+W)$ we define:
\be\label{eq:renint}
Z^W(F)\doteq  Z(F+W(f))-W(f)\in\Floc(M,L+W)
\ee
with $f\in\Dc_s(M)$ and $f\equiv 1$ on the supports of $F$ and $Z$. 
\end{definition}
This definition%
\footnote{The notation $Z^W$ should not be confused with the notation $Z_W=Z(\bullet+W(f))-Z(W(f))$ used in \cite{BDF09}. The latter describes the renormalization of fields induced by a renormalization $W(f)\mapsto Z(W(f))$ of the interaction $W$.} 
is motivated by the following results:
\begin{proposition}\label{prop:ZW}
Let $Z^W$ be defined in terms of $Z\in\Rc(M,L)$ and $W\in\mathrm{Int}(M,L)$ by \eqref{eq:renint}. Then:
\begin{itemize}
    \item[($i$)] $Z^W(F)$ does not depend on the choice of the test function $f\in\Dc_s(M)$,
    \item[($ii$)] $Z^W\in\Rc(M,L+W)$.
\end{itemize}
\end{proposition}

\begin{proof}
($i$) Let $f_1,f_2\in\Dc_s(M)$ have the above mentioned property. Since $\supp(W(f_1)-W(f_2))\cap(\supp F\cup\supp Z)=\emptyset$
we may use Locality $(ii)$ and $Z(W(f_1))=Z(W(f_2))+(W(f_1)-W(f_2))$ in the following way (by abuse of notation):
\begin{align*}
    Z^{W(f_1)}(F)&=Z\bigl(F+W(f_2)+(W(f_1)-W(f_2))\bigr)-W(f_1)\\
    &=Z\bigl(F+W(f_2)\bigr)-Z\bigl(W(f_2)\bigr)+Z\bigl(W(f_1)\bigr)-W(f_1)\\
    &=Z\bigl(F+W(f_2)\bigr)+(W(f_1)-W(f_2))-W(f_1)\\
    &=Z^{W(f_2)}(F).
\end{align*}
($ii$) By \eqref{eq:Floc-add} we see that $Z^W$ maps $\Floc(M,L+W)$ into itself. $Z^W$ is also invertible with (bilateral) inverse $(Z^W)^{-1}=(Z^{-1})^W$ as may be seen from the equation
\begin{equation}\label{eq:inverse-zw}
    Z^W(Z^{-1})^W(F)=Z^W\bigl(Z^{-1}(F+W)-W\bigr)=Z\bigl(Z^{-1}(F+W)\bigr)-W=F\ .
\end{equation}
We are now going to verify that each of the properties $(i)$-$(v)$ in Definition \ref{def1} for $Z^W$ follows 
from the corresponding property of $Z\in\Rc(M,L)$. For Locality $(ii)$ this is 
straightforward; for the other properties we give the details.

To verify $(i)$, let $\supp F\cap\supp Z=\emptyset$. Then, we obtain
$Z^W(F+G)-Z^W(G)=F$, hence $\supp F\cap\supp Z^W=\emptyset$. It follows that $\supp Z^W\subset\supp Z$.

The property $(iii)$ can be checked as follows:
\begin{align*}
    Z^W\bigl(F^\psi+\delta(L+W)(\psi)\bigr)&=Z\bigl(F^\psi+W(f)^\psi+\delta L(\psi)\bigr)-W(f)\\
    &=Z\bigl(F+W(f)\bigr)^\psi+\delta L(\psi)-W(f)\\
    &=Z\bigl(F+W(f)\bigr)^\psi-W(f)^\psi+\delta (L+W)(\psi)\\
    &=Z^W(F)^\psi+\delta (L+W)(\psi).
    \end{align*}
    
To verify $(iv)$, we take into account that $L+W=L_0+(V+W)$ and obtain
\begin{align*}
    Z^W\bigl(F^\psi-(V+W)(f)\bigr)&=Z\bigl(F^\psi-V(f)\bigr)-W(f)\\
    &=Z\bigl(F-V(f)\bigr)^\psi+V(f)^\psi-V(f)-W(f)\\
    &=Z\bigl(F-(V+W)(f)+W(f)\bigr)^\psi-W(f)^\psi+\delta(V+W)(\psi)\\
    &=Z^W\bigl(F-(V+W)(f)\bigr)^\psi+\delta(V+W)(\psi).
\end{align*}
It remains to check $(v)$. We have for $f\in\Dc_s(M)$, $f\equiv 1$ on some relatively compact region $\Oc\supset\supp Z\cup\supp F$
\begin{equation}
    (L+W+A_{Z^W(F)})(f)=L(f)+W(f)+Z^W(F)=L(f)+Z(F+W(f))\ .
\end{equation}
But by assumption, $L+A_{Z(F+W(f))}$ induces the same causal structure as $L+A_{F+W(f)}$ which coincides on $\Oc$ with the causal structure of $L+W+A_{Z(F)}$. Since $\Oc$ can be arbitrarily large, condition $(v)$ is fulfilled.
\end{proof}

\begin{remark}\label{rem:Z-shift}
The defining condition $(iv)$ of $Z\in\Rc(M,L_0+V)$ in Definition \ref{def1} can be obtained from its simpler formulation in the particular case of 
$\Rc(M,L_0)$ (given in \eqref{eq:shift-RL0}) in the following way: looking at $Z^W$, we choose $W$ such that it compensates 
the interaction, that is, $W=-V$. Then, Proposition \ref{prop:ZW} states that $Z^{-V}\in\Rc(M,L_0)$, hence we may apply
\eqref{eq:shift-RL0}:
\be\label{eq:shift-Z-V}
Z^{-V}(F^\psi)=Z^{-V}(F)^\psi\ ,\,\, F\in\Floc(M,L_0)\ .
\ee
Inserting the definition of $Z^{-V}$, we precisely get the (general) formulation of $(iv)$ in $\Rc(M,L_0+V)$
(given in \eqref{eq:shift-Z}). Hence, \eqref{eq:shift-Z-V} is an equivalent reformulation of \eqref{eq:shift-Z}.
\end{remark}
We are now ready to prove that naming $\Rc(M,L)$ a renormalization \emph{group} is sound:
\begin{proposition}
$\Rc(M,L)$ is a group.
\end{proposition}
\begin{proof}
We first show that $Z_1Z_2\in\Rc(M,L)$ for $Z_1,Z_2\in\Rc(M,L)$ by verifying that the conditions $(i)-(v)$ of Definition~\ref{def1} hold.
\begin{itemize}
\item [$(ii)$](Locality) Let $G\in\Floc(M,L)$, $F,H\in\Floc(M,L+A_G)$ and $\supp F\cap\supp H=\0$.
Then
\begin{equation}
\begin{split}
    Z_1&Z_2(F+G+H)\\
    &=Z_1\bigl(Z_2(F+G)-Z_2(G)+Z_2(G+H)\bigr)\\
    &=Z_1\bigl(\underbrace{Z_2(F+G)-Z_2(G)}_{\supp F}+Z_2(G)+\underbrace{Z_2(G+H)-Z_2(G)}_{\supp H}\bigr)\\
    &=Z_1Z_2(F+G)-Z_1Z_2(G)+Z_1Z_2(G+H)\ .
    \end{split}
\end{equation}
\item[$(iii)$] (Dynamics) 
\begin{equation}
    Z_1Z_2\bigl(F^{\psi}+\delta L(\psi)\bigr)=Z_1\bigl(Z_2(F)^{\psi}+\delta L(\psi)\bigr)
    =Z_1Z_2(F)^{\psi}+\delta L(\psi)\ .
\end{equation}
\item[$(iv)$] (Field shifts) This is evident for the massless Lagrangian $L_0$. It remains true for the general case,
since $(iv)$ can equivalently be written as \eqref{eq:shift-Z-V}, and due to 
\begin{equation}\label{eq:comb-z}
\begin{split}
    Z_1^W Z_2^W(F)&=Z_1^W\bigl(Z_2(F+W(f))-W(f)\bigr)\\
    &=Z_1\bigl(Z_2(F+W(f))-W(f)+W(g)\bigr)-W(g)\\
    &=Z_1Z_2\bigl(F+W(f)\bigr)-W(f)+W(g)-W(g)\\
    &=(Z_1Z_2)^W(F)
    \end{split}
\end{equation}
where $f\equiv 1$ on $\supp Z_1\cup\supp Z_2\cup\supp F$ and $g\equiv 1$ on 
$\supp f$. In the third line we have used that $\supp\bigl(W(g)-W(f)\bigr)\cap\supp Z_1=\emptyset$.

\item[$(i)$](Compact support)
Let $\supp F\cap\supp Z_i=\0$, $i=1,2$. Then
\begin{equation}
    Z_1Z_2(F+G)=Z_1\bigl(F+Z_2(G)\bigr)
    =F+Z_1Z_2(G)\ ,
\end{equation}
hence $\supp Z_1Z_2\subset\supp Z_1\cup\supp Z_2$.
\item[$(v)$] (Causal Stability) evident.
\end{itemize}
We now show that $Z\in\Rc(M,L)$ implies $Z^{-1}\in\Rc(M,L)$.
\begin{itemize}
    \item [$(ii)$] (Locality) Let $G\in\Floc(M,L)$ and $F,H\in\Floc(M,L+A_G)$ with $\supp F\cap\supp H=\0$. Set
    $G'\doteq Z^{-1}(G)$, $F'\doteq Z^{-1}(F+G)-G'$ and $H'\doteq Z^{-1}(G+H)-G'$. 
    Then $F=Z(F'+G')-Z(G')$ and $H=Z(G'+H')-Z(G')$, thus 
    by \eqref{eq:invar-suppZ} we have $\supp F'=\supp F$ and $\supp H'=\supp H$, hence
    \begin{equation}
        F+G+H=Z(F'+G')-Z(G')+Z(G'+H')=Z(F'+G'+H')
    \end{equation}
     by Locality of $Z$, and therefore
    \be
    Z^{-1}(F+G+H)=Z^{-1}(F+G)-Z^{-1}(G)+Z^{-1}(G+H)\ .\ee
\item[$(iii)$](Dynamics) We apply $Z^{-1}$
to the equation
\begin{equation}
    Z\bigl(Z^{-1}(F)^{\psi}+\delta L(\psi)\bigr)=F^{\psi}+\delta L(\psi) 
\end{equation}
and get the wanted relation.
\item[$(iv)$](Field shifts) This is again evident for the Lagrangian $L_0$. The general case follows from 
$Z^{-V}\in\Rc(M,L_0)$ and the equality $(Z^W)^{-1}=(Z^{-1})^W$.  
\item[$(i)$](Compactness of support)
It is easy to check that $Z^{-1}$ verifies also \eqref{eq:invar-suppZ}. Indeed, let $F'=Z^{-1}(F+G)-Z^{-1}(G)$, then
\begin{equation}\label{eq:supp-inv}
    F=Z\bigl(F'+Z^{-1}(G)\bigr)-Z\bigl(Z^{-1}(G)\bigr)
\end{equation}
hence $\supp F'=\supp F$.

Now, let $\supp F\cap\supp Z=\0$. Then from \eqref{eq:supp-inv}
with $\supp F'=\supp F$, we get $F=F'=Z^{-1}(F+G)-Z^{-1}(G)$. We conclude that
$\supp Z^{-1}\subset\supp Z$, and therefore $\supp Z^{-1}=\supp Z$. 
\item[$(v)$] (Causal Stability) evident.
\end{itemize}
\end{proof}


\begin{corollary}
The map $\Rc(M,L)\to\Rc(M,L+W)\,;\,Z\mapsto Z^W$ \eqref{eq:renint} is a group isomorphism for any interaction $W\in\mathrm{Int}(M,L)$.
\end{corollary}
\begin{proof}
The claim is evident from \eqref{eq:inverse-zw} and \eqref{eq:comb-z}.
\end{proof}

The non-triviality of the group $\Rc(M,L)$ gives rise to a group of automorphisms of $\fA(M,L)$ preventing the validity
of the time slice axiom. In detail:

\begin{proposition}\label{prop:automorphism-by-Z}
The map
\be\label{eq:beta-Z} 
S(F)\mapsto S(Z(0))^{-1}S(Z(F)) \ ,\ 
\ee
induces a representation $\beta^{\mathrm{ret}}$(named retarded) of $\mathcal{R}(M,L)$ by automorphisms $\beta_Z^{\mathrm{ret}}$ of $\fA(M,L)$. The 
analogous advanced representation $\beta^{\mathrm{adv}}$ of $\mathcal{R}(M,L)$ is obtained by reversing the order of the factors in \eqref{eq:beta-Z}. 
\end{proposition}
\begin{proof} 
We first check that $\beta_Z\equiv \beta_Z^{\mathrm{ret}}$ preserves the defining relations of $\fA(M,L)$.
\begin{itemize}
\item
Causality Relation: Let $G\in\Floc(M,L)$ and $F,H\in\Floc(M,L+A_G)$ for which it holds that $J^{L+A_G}_+(\supp F)\cap\supp H=\0$. Then
\be
\begin{split}\label{eq:caus-SZ}
 &S(Z(0))\,\beta_Z(S(F+G+H))\\
 &=S(Z(F+G+H))\\
 &=S(Z(F+G)-Z(G)+Z(G+H))\\
&=S\bigl((Z(F+G)-Z(G))+Z(G)+(Z(G+H)-Z(G))\bigr)\ .
\end{split}
\ee
Since $\supp Z(F+G)-Z(G)=\supp F$ and $\supp Z(G+H)-Z(G)=\supp H$ and by Definition \ref{def1} $(v)$ $J^{L+A_{Z(G)}}_+=J^{L+A_G}_+$ we obtain
\be
\begin{split}
 &S\bigl(((Z(F+G)-Z(G))+Z(G)+(Z(G+H)-Z(G)))\bigr)\\
 &=S(Z(F+G))S(Z(G))^{-1}S(Z(G+H))\\&
 =S(Z(0))\,\beta_Z(S(F+G))\,\beta_Z(S(G))^{-1}\,\beta_Z(S(G+H))\ .
 \end{split}
\ee
\item Dynamical Relation:
\be
\begin{split}
 &S(Z(0))\,\beta_Z(S(F^{\psi}+\delta L(\psi)))=S(Z(F^{\psi}+\delta L(\psi)))\\
 &=S(Z(F)^{\psi}+\delta L(\psi))=S(Z(F))=S(Z(0))\,\beta_Z(S(F))\ .   
\end{split}
\ee
\end{itemize}
We see that $\beta_Z$ is an endomorphism. For $Z_1,Z_2\in\mathcal{R}(M,L)$ we find
\begin{align}
    \beta_{Z_1}\beta_{Z_2}(S(F))&=\beta_{Z_1}\bigl(S(Z_2(0))^{-1}S(Z_2(F))\bigr)\\
    &=\bigl(\beta_{Z_1}(S(Z_2(0)))\bigr)^{-1}\beta_{Z_1}\bigl(S(Z_2(F))\bigr)\nonumber\\
    &=\bigl(S(Z_1(0))^{-1}S(Z_1Z_2(0))\bigr)^{-1}S(Z_1(0))^{-1}S(Z_1Z_2(F))\nonumber\\
    &=\beta_{Z_1Z_2}(S(F))
\end{align}
and with $\beta_{\mathrm{id}_{\mathcal{R}(M,L)}}=\mathrm{id}_{\fA(M,L)}$ we conclude that $Z\mapsto\beta_Z$ is a representation by automorphisms.
\end{proof}

\begin{remark}\label{rm:main-thm} 
Compared with the perturbative $S$-matrix and the St\"uckelberg-Petermann renormalization group $\mathcal{R}_0$,
Proposition \ref{prop:automorphism-by-Z} corresponds to the part of the Main Theorem of perturbative renormalization, 
stating that if $S$ is an 
admissible $S$-matrix ({\it i.e.}, satisfies the axioms for the time-ordered product) this holds also for $S\circ Z$ for all 
$Z\in\mathcal{R}_0$, see \cite{DF04,BDF09,Due19}. 
\end{remark}
\begin{remark}\label{rem:beta}
According to \eqref{eq:invar-suppZ}, the support of $Z(F)$ does not necessarily coincide with the support of $F$ but only the support of $Z(F)-Z(0)$ does. Hence the automorphism $\beta^{\mathrm{ret}}_Z$ does not in general preserve the local subalgebras, thus is not an automorphism of the net. Instead its image can be interpreted as a theory with an additional interaction $Z(0)$ (in the sense of a generalized field) together with a field renormalization $F\mapsto Z(F)-Z(0)$, see \eqref{eq:betaZ-alpha}. This will be crucial in the discussion of the renormalization group flow 
in Section \ref{sec:RGF}.
\end{remark}
The automorphisms induced by the renormalization group transform under the addition of an interaction
in the following way:
\begin{proposition}
Let $V\in\mathrm{Int}(M,L)$ have past compact support  and $Z\in\Rc(M,L)$. Then
\begin{equation}
    \beta_Z^{\mathrm{ret}}\circ\alpha_{V,+}=\alpha_{V,+}\circ\beta_{Z^V}^{\mathrm{ret}}\ .
\end{equation}
Analogously, if $V$ is future compact, then it holds that 
$\beta_Z^{\mathrm{adv}}\circ\alpha_{V,-}=\alpha_{V,-}\circ\beta_{Z^V}^{\mathrm{adv}}$.
\end{proposition}
\begin{proof}
Let $F\in\Floc(M,L+V)$. Then if $f\in\Dc_s(M)$ with $f\equiv 1$ on a sufficiently large region and with the abbreviations $S_{(M,L)}\equiv S$, $S_{(M,L+V)}\equiv S'$, and $\beta_{\bullet}\equiv\beta_{\bullet}^{\mathrm{ret}}$
the claim follows from
\begin{equation}
    \begin{split}
        \alpha_{V,+}\circ\beta_{Z^V}(S'(F))&=\alpha_{V,+}\left(S'(Z^V(0))^{-1}S'(Z^V(F))\right)\\
                                      &=S(V(f)+Z^V(0))^{-1}S(V(f)+Z^V(F))\\
                                      &=S(Z(V(f))^{-1}S(Z(V(f)+F))\\
                                      &=\beta_Z\left(S(V(f))^{-1}S(V(f)+F)\right)\\
                                      &=\beta_Z\circ\alpha_{V,+}(S'(F))\ .
    \end{split}
\end{equation}
\end{proof}

We collect some properties of renormalization group elements. 
\begin{proposition}\label{prop:Z} Let $Z\in\mathcal{R}(M,L)$, $L=L_0+V$ with the free massless Lagrangian $L_0$ and $V\in\mathrm{Int}(M,L_0)$.
\begin{enumerate}
    \item[$(i)$] For an affine functional $F=\langle\phi,h\rangle+c$, $c\in\RR$, $h\in\Dc_{\mathrm{dens}}(M,\RR^n)$ the relation
    \begin{equation}\label{eq:Z(F+c)}
        Z(F+G)=F+Z(G),\ G\in\Floc(M,L)
    \end{equation}
    holds true. Therefore the addition of a source term $\langle \phi,q\rangle$ to the Lagrangian, with a smooth density $q\in\Ec_{\mathrm{dens}}(M,\RR^n)$,
    does not change $Z$, \ie $Z^{\langle \phi,q\rangle}=Z$, and  $\Rc(M,L+\langle \phi,q\rangle)=\Rc(M,L)$.
    \item[$(ii)$] 
    We have
    \begin{equation}\label{eq:suppZ(0)}
        \supp Z(0)\subset\supp Z\cap\supp V\ 
    \end{equation}
    and, more generally,
    \be\label{eq:suppZ(F)}
    \supp Z(F)\subset\supp F\cup (\supp Z\cap\supp V).
    \ee
    In addition it holds that
    \be\label{eq:supp(Z(F)-F)}
    \supp \bigl(Z(F)-F\bigr)\subset\supp Z\ ,\quad  F\in\Floc(M,L)\ .
    \ee
    \item[$(iii)$] 
    %
    Let $V(f)\in\mathrm{Q}(M,L_0)$ for $f\in\Dc_s(M)$.
    Then $Z$ is invariant under shifts,
    \begin{equation}\label{eq:shift-Z2}
       Z(F^{\psi})=Z(F)^{\psi}\ ,\  \psi\in\Dc(M,\RR^n)\ . 
    \end{equation}
    Moreover, for $F\in\mathrm{Q}(M,L)$ we find
    \begin{equation}\label{eq:quadratic}
        Z(F)=F+c_F
    \end{equation}
    with a constant functional $c_F$. In particular $Z(0)$ is a constant functional.
\end{enumerate}
\end{proposition}

\begin{proof}
\begin{itemize}
\item[$(i)$] 
Let $h=h_0+K\psi$, $h_0\in\Dc_{\mathrm{dens}}(M,\RR^n)$, $\psi\in\Dc(M,\RR^n)$ where $\supp h_0\cap\supp Z=\0$ and $K$ is minus the d'Alembertian, considered as a map from functions to densities. Then
\begin{equation}
    \langle \phi,K\psi\rangle=\delta L(\psi)-\frac12\langle\psi,K\psi\rangle-V(f)^{\psi}+V(f)
\end{equation}
with $f\in\Dc_s(M)$, $f\equiv 1$ on $\supp\psi$. Since constant functionals have empty support, we obtain, from the definition of the support of $Z$, 
\begin{equation*}
    \begin{split}
        Z(F+G)&=Z\bigl(G+\delta L(\psi)+V(f)-V(f)^{\psi}\bigr)+\langle\phi,h_0\rangle-\frac12\langle\psi,K\psi\rangle+c\\
        &=Z\bigl((G+V(f))^{-\psi}-V(f)\bigr)^{\psi}+\delta L(\psi)+\langle\phi,h_0\rangle-\frac12\langle\psi,K\psi\rangle+c\\
        &=Z(G)-\delta V(\psi)+\delta L(\psi)+\langle\phi,h_0\rangle-\frac12\langle\psi,K\psi\rangle+c\\
        &=Z(G)+F
    \end{split}
\end{equation*}    
    where the second equality follows from $(iii)$ and the third one from $(iv)$ in Definition \ref{def1}.
\item[$(ii)$]
Let $\psi\in\Dc(M,\RR^n)$ with $\supp\psi\cap\supp Z\cap\supp V=\0$. We decompose $\psi=\psi_1+\psi_2$ with $\supp\psi_1\cap\supp Z=\0$, $\supp\psi_2\cap\supp V=\0$ and $\supp\psi_1\subset\supp\psi$. From Definition \ref{def1}$(iv)$ we have with $f\equiv1$ on $\supp\psi$, $f\in\Dc_s(M)$
\begin{equation}
    Z(0)^{\psi}=Z(V(f)^{\psi}-V(f))-V(f)^{\psi}+V(f)=Z(\delta V(\psi))-\delta V(\psi)\ .
\end{equation}
With $\delta V(\psi_1+\psi_2)=\delta V(\psi_1)+\delta V(\psi_2)^{\psi_1}$ and $\delta V(\psi_2)=0$ (Prop.~\ref{prop:suppA=suppdA}) 
we get
\begin{equation}
    Z(0)^{\psi}=Z(\delta V(\psi_1))-\delta V(\psi_1)=Z(0)
\end{equation}
since $\supp\delta V(\psi_1)\cap\supp Z=\0$.
 Hence $\supp Z(0)\subset \supp Z\cap\supp V$ by definition of the support of functionals.

The statement \eqref{eq:suppZ(F)} follows immediately from \eqref{eq:invar-suppZ} and \eqref{eq:suppZ(0)}:
$$    
\supp Z(F)=\supp\bigl((Z(F)-Z(0))+Z(0)\bigr)\subset\supp F\cup \supp Z(0)\ .
$$
To prove \eqref{eq:supp(Z(F)-F)} 
we decompose $F=F_1+F_2$ with $\supp F_2\cap\supp Z=\emptyset$. 
By the definition of $\supp Z$ we obtain
\begin{align*}
    Z(F)&=Z(F_1)+F_2\ ,
\end{align*}
hence $Z(F)-F=Z(F_1)-F_1$. Taking into account \eqref{eq:suppZ(F)}, we conclude that $\supp(Z(F_1)-F_1)\subset\supp F_1\cup\supp Z$. Since for any $x\not\in\supp Z$ we can find a decomposition with $x\not\in\supp F_1$ we arrive at the conclusion.

\item[$(iii)$] 
To prove \eqref{eq:shift-Z2} we use again Definition \ref{def1}$(iv)$ and get for $f\equiv 1$ on $\supp\psi$, $f\in\Dc_s(M)$
\begin{equation}
    \begin{split}
        Z(F^{\psi})&=Z((F+V(f)^{-\psi})^{\psi}-V(f))\\
        &=Z(F+V(f)^{-\psi}-V(f))^{\psi}+\delta V(\psi)\\
        &=Z(F)^{\psi}+(V(f)^{-\psi}-V(f))^{\psi}+\delta V(\psi)\\
        &=Z(F)^{\psi}
    \end{split}
\end{equation}
where we used \eqref{eq:Z(F+c)} in the third equality.

In particular, for a quadratic functional $F\in\mathrm{Q}(M,L)$ the difference $F-F^\psi$ is affine, hence
\begin{equation}
    Z(F)-F=Z(F^\psi+(F-F^\psi))-F=Z(F^\psi)-F^\psi=(Z(F)-F)^\psi,
\end{equation}
from which we conclude that $Z(F)-F$ is a constant functional.
\end{itemize}
\end{proof}

\section{Master Ward Identity}\label{sec:MWI}
%
 In certain cases several elementary morphisms exist between objects of $\mathfrak{Loc}$. Namely let $M'=M$, let $\chi$ be  a diffeomorphism of $M$ with compact support, let $L'=\chi_{\ast}L$ and assume that the time orientations coincide outside of the support of $\chi$. Set 
\be
 \delta_{\chi}L \doteq (L'-L)(f)\in\Floc(M,L)\ ,\quad f\in\Dc_s(M)\ \textrm{with}\ f\equiv 1\ \textrm{on}\ \supp\chi\ .
\ee

$\delta_{\chi}L$ can be interpreted as a retarded or an advanced interaction.\footnote{For the simplicity of notation, we will write $\delta_{\chi}L$ instead of $A_{\delta_{\chi}L}$.} Hence we obtain the arrow  
$\iota_\chi$ from $(M,L)\to(M,L')$ and the arrows $\iota_{\delta_{\chi}L,+},\iota_{\delta_{\chi}L,-}$ from $(M,L')\to(M,L)$ and can build the automorphisms
\be
\alpha^{\chi}_{\pm}\doteq 
\alpha_{\delta_{\chi}L,{\pm}}\circ\alpha_{\chi}
\ee
of $\fA(M,L)$. We compute
\be
\alpha^{\chi}_+(S_{(M,L)}(F))=S_{(M,L)}(\delta_{\chi}L)^{-1}S_{(M,L)}(\chi_*F+\delta_{\chi}L)
\ee
and
\be
\alpha^{\chi}_-(S_{(M,L)}(F))=S_{(M,L)}(\chi_*F+\delta_{\chi}L)S_{(M,L)}(\delta_{\chi}L)^{-1}\ .
\ee
If the support of $\chi$ does not intersect the past (future) of $\supp F$, 
we have $\chi_*F=F$ and from the Causality Relation we conclude that $\alpha^{\chi}_{\pm}$ acts trivially on $S_{(M,L)}(F)$.


An analogous discussion can be performed for compactly supported affine field redefinitions $\Phi$. Let again $M'=M$, and let $L'=\Phi_{\ast}L$. Set $\delta_{\Phi}L=L'-L$ and define
\be 
\alpha^{\Phi}_{\pm}
\doteq\alpha_{\delta_{\Phi}L,\pm}\circ\alpha_{\Phi}\ .
\ee
In the case of a translation $\Phi:\phi\to\phi+\psi$ we see that the Dynamical Relation fixes these automorphisms to the identity.

The diffeomorphisms $\chi$ act on affine field redefinitions $\Phi=(A,\phi_0)$ by
\begin{equation}\label{eq:semi}
    \chi\Phi\doteq (A\circ\chi^{-1},\phi_0\circ\chi^{-1})\ .
\end{equation}
We find the following commutation relation
\begin{equation}\label{eq:chi-Phi}
    \chi_{\ast }\Phi_\ast =(\chi\Phi)_\ast \chi_\ast \ .
\end{equation}
Motivated by this property, we define:

\begin{definition}[\textbf{Symmetry Transformations}]\label{def:G}
The group of (compactly supported) \textbf{symmetry transformations} $$\mathpzc{G}_c(M)\doteq\mathcal{C}^{\infty}_c(M,\mathrm{Aff}(\RR^n))\rtimes\mathrm{Diff}_c(M)$$ 
is the semidirect product of the group of compactly supported diffeomorphisms of $M$, $\mathrm{Diff}_c(M)$, with the group of compactly supported  smooth and affine field redefinitions,  $\mathcal{C}^{\infty}_c(M,\mathrm{Aff}(\RR^n))$, with respect to the action \eqref{eq:semi},
{\it i.e.}, the product in $\mathpzc{G}_c(M)$ is defined by
\footnote{Hence, we obtain the same commutation law as in \eqref{eq:chi-Phi}:
$(\mathrm{id}_\Phi,\chi)\cdot(\Phi,\mathrm{id}_\chi)=((\chi\Phi),\mathrm{id}_\chi)\cdot(\mathrm{id}_\Phi,\chi)$\ .}
\begin{equation}
(\Phi_1,\chi_1)\cdot (\Phi_2,\chi_2)\doteq(\Phi_1\circ(\chi_1\Phi_2),\chi_1\chi_2)\ .
\end{equation}
The support of such a transformation $\g\doteq(\Phi,\chi)\in \mathpzc{G}_c(M)$ is
\begin{equation}
    \supp{(\Phi,\chi)}\doteq\supp\Phi\cup\supp\chi\ .
\end{equation}
\end{definition}

The group $\mathpzc{G}_c(M)$ acts on $\Floc(M,L)$ by
\begin{equation}
    (\Phi,\chi)_{\ast }\doteq\Phi_{\ast }\chi_{\ast }\ .
\end{equation}
By using \eqref{eq:chi-Phi} we verify that
\be
(\g_1\g_2)_\ast  = \g_{1\ast }\,\g_{2\ast }\ .
\ee
We also see that the support of $\g\in\mathpzc{G}_c(M)$ is the smallest closed subset $N$ of $M$ such that 
$\supp F\cap N=\emptyset$ implies $\g_\ast  F=F$.
The \emph{diffeomorphism part} $\tildeg$ of $\g\in\mathpzc{G}_c(M)$ can be recovered from $\g_{\ast }$
by
\be\label{eq:tilde-g}
\supp(\g_\ast  F)=\tildeg(\supp F)\ ,\quad
F\in\Floc(M,L)\ .
\ee
We set 
$\delta_\g L\doteq\g_{\ast}L(f)-L(f)$ with $f\in\Dc_s(M)$ and $f\vert_{\supp\g}=1$ and observe the relation
\begin{equation}
    \delta_{\g\h}L=\g_{\ast }\delta_\h L+\delta_{\g}L\ ,\ 
    \g,\h\in\mathpzc{G}_c(M)\ .
\end{equation}
We also set 
\begin{equation}
    \alpha_{\pm}^{(\Phi,\chi)}\doteq\alpha_{\pm}^{\Phi}\circ\alpha_{\pm}^{\chi}\ ,
\end{equation}
with $\g=(\Phi,\chi)$; this can equivalently be written as
\be\label{eq:alpha+g}
\alpha_+^{\g}\bigl(S_{(M,L)}(F)\bigr)=S_{(M,L)}(\delta_\g L)^{-1} S_{(M,L)}(\g_\ast  F+\delta_\g L)\ ,
\ee
and analogously for $\alpha_-^{\g}\bigl(S_{(M,L)}(F)\bigr)$; finally we find 
\be\label{eq:alpha-gh}
\alpha_{\pm}^{\g\h}=\alpha_{\pm}^\g\circ\alpha_{\pm}^\h\ . 
\ee

In a first step we compute these automorphisms for the case of a linear field redefinition, a quadratic Lagrangian $L=\frac12\langle\phi,K\phi\rangle$, with a normally hyperbolic formally selfadjoint differential operator $K$ (again considered as a map from functions to densities), and affine functionals $F_h=\langle\phi,h\rangle+\frac12\langle h,\Delta_{\mathrm{adv}} h\rangle =
\langle\phi,h\rangle+\frac12\langle h,\Delta_{\mathrm{ret}} h\rangle $, $h\in\Dc_{\mathrm{dens}}(M,\RR^n)$ with the retarded ($\Delta_{\mathrm{ret}}$) and advanced ($\Delta_{\mathrm{adv}}$) inverses of $K$.
Let 
\be
\Phi[\phi](x)=A^t(x)\phi(x)
\ee
where $A$ is a $\mathrm{GL}(n,\RR)$-valued smooth function with compact support and the upper index $t$ denotes the transpose. Then
\be
\delta_{\Phi}L(x)=\frac12 \langle \phi,((AKA^t-K)\phi\rangle\ .
\ee
We use the formula
\be
S_{(M,L)}(G)S_{(M,L)}(\langle\phi,h\rangle)=S_{(M,L)}(G^{\Delta_{\mathrm{adv}}h}+\langle\phi,h\rangle)
\ee
derived in \cite{BF20} for test densities $h$ and local functionals $G$.
We have
\be
(\delta_{\Phi}L)^{\Delta_{\mathrm{adv}}h}+\langle\phi,h\rangle
=\delta_{\Phi}L+\langle\phi,(AKA^t\Delta_{\mathrm{adv}} h\rangle +\frac12\langle\Delta_{\mathrm{adv}} h,((AKA^t-K)\Delta_{\mathrm{adv}} h\rangle
\ee
and hence with $h'\doteq K(A^t)^{-1}\Delta_{\mathrm{adv}} h$ we get
\be
S_{(M,L)}(\delta_{\Phi}L+\Phi_{\ast}F_h)=S_{(M,L)}(\delta_{\Phi}L)S_{(M,L)}(\langle\phi,h'\rangle)e^{\frac{i}{2}\langle h',\Delta_{\mathrm{adv}} h'\rangle}
\ee
hence 
\be
\alpha^{\Phi}_+(S_{(M,L)}(F_h))=S_{(M,L)}(F_{h'})\ .
\ee
Now, the elements $W(h)\doteq  S(F_h)$ satisfy the Weyl relations
\be
W(h_1)W(h_2)=W(h_1+h_2)e^{-\frac{i}{2}\langle h_1,\Delta h_2\rangle}\ ,\ W(K\psi)=1\ ,\ \psi\in\Dc(M,\RR^n)\ ,
\ee
where $\Delta=\Delta_{\mathrm{ret}}-\Delta_{\mathrm{adv}}$, as derived in \cite{BF19}. Since
\be
h'=h+K((A^t)^{-1}-1)\Delta_{\mathrm{adv}}h
\ee
we conclude that $\alpha^{\Phi}_+$ acts trivially on $S_{(M,L)}(F_h)$. 
This remains true when we replace $F_h$ by a general affine functional $F=F_h+c,\,\,c\in\RR$, due to $S_{(M,L)}(G+c)=S_{(M,L)}(G)\,e^{ic}$. 
Analogously we find $\alpha^{\Phi}_-(W(h))=W(h)$.
This implies in particular that $S_{(M,L)}(\delta_{\Phi}L)$ commutes with all Weyl operators and is thus a multiple of 1 in an 
irreducible representation of the Weyl algebra.
Indeed, in a perturbative calculation only numerical terms can contribute to $S_{(M,L)}(\delta_{\Phi}L)$.

An analogous calculation can be done for a generic element of $\mathpzc{G}_c(M)$. To summarize,
we have found cases in which $S_{(M,L)}(\delta_{\g}L+\g_{\ast }F)$ is equal to $S_{(M,L)}(F)$ up to a 
constant phase factor $S_{(M,L)} (\delta_{\g} L)$.

In a formal path integral calculation the transformation $\g\in\mathpzc{G}_c(M)$ corresponds to a change of variables. 
There the time ordered exponential of some local functional $F$ corresponds to the time ordered exponential of 
$\delta_{\g} L+\g_{\ast }F$ up to the Jacobian of the transformation. In renormalized perturbation theory this is taken into account in the \emph{BV-Laplacian}, which is a term in the \emph{quantum master equation}. Within causal perturbation theory it results in the anomalous Master Ward Identity. See \cite{FredenhagenR13} for the detailed account of the BV formalism and its connection to the anomalous Master Ward Identity.

In our axiomatic framework we incorporate it in the following axiom: 
We introduce an $L$-dependent action of $\mathpzc{G}_c(M)$ on $\Floc(M,L)$
\begin{equation}\label{eq:g_L}
    (\g,F)\mapsto \g_LF\doteq \delta_{\g} L+\g_{\ast } F\ .
\end{equation}
Obviously $\mathpzc{e}_L=\mathrm{id}_{\Floc(M,L)}$ for the unit $\mathpzc{e}\in \mathpzc{G}_c(M)$ and, since 
$(\g\h)_\ast =\g_\ast  \h_\ast $, we easily verify that 
$(\g\h)_L=\g_{L} \h_{L}$. The possible deviation of the MWI is described in terms of a map $$\zeta:\mathpzc{G}_c(M)\rightarrow \Rc(M,L)\,,$$ from the group of symmetry transformations $\mathpzc{G}_c(M)$ (Def. \ref{def:G}) to the renormalization group of the Lagrangian $L$, $\Rc(M,L)$ (Def. \ref{def1}),
satisfying $\zeta_\e=\mathrm{id}_{\Floc(M,L)}$, $\supp\zeta_{\g}\subset\supp \g$ and the \emph{cocycle relation}
\begin{equation}\label{eq:cocycle}
 \zeta_{\g\h}=\zeta_{\h}(\zeta_{\g})^{\h}\quad\text{where}\quad
(\zeta_\g)^\h\doteq\h_L^{-1}\zeta_{\g} \h_L\ ,\quad \g,\h\in\G_c(M)\ ,
\end{equation}
The set of these cocycles is denoted by $\mathfrak{Z}(M,L)$.

\begin{axiom}[\textbf{Symmetries}] 
Let $(M,L)$ be some dynamical spacetime.
A representation $\pi$ of $\fA(M,L)$ satisfies the \textbf{unitary anomalous Master Ward Identity} (unitary AMWI) if there exists
some $\zeta\in\mathfrak{Z}(M,L)$ 
such that
\begin{equation}\label{eq:uni-anom-MWI}
    \pi\circ S_{(M,L)}\circ\g_L=\pi\circ S_{(M,L)}\circ\zeta_\g\ ,\quad \g\in\mathpzc{G}_c(M) \ .
\end{equation}
\end{axiom}

Let $I_{\zeta}$ denote the intersection of all ideals annihilated by representations which satisfy the unitary AMWI for a specific $\zeta$, and let $\fA(M,L,\zeta)$ denote the quotient $\fA(M,L)/I_{\zeta}$ and $\fA_{(M,L,\zeta)}$ the corresponding net. The C*-algebra $\fA(M,L,\zeta)$ is generated by the unitaries $S_{(M,L,\zeta)}(F)=S_{(M,L)}(F)+I_{\zeta}$. 

For any $F$ with $\supp F\cap\supp\g =\emptyset$ the identity \eqref{eq:uni-anom-MWI} reads
\be\label{eq:clear}
S_{(M,L,\zeta)} (\delta_{\g} L+F)=S_{(M,L,\zeta)}(\zeta_{\g}(0)+F)\ ,\,\,\text{in particular}\quad
S_{(M,L,\zeta)} (\delta_{\g} L)=S_{(M,L,\zeta)}\bigl(\zeta_{\g}(0)\bigr)\ ,
\ee
by using Definition \ref{def1}$(i)$. 

\begin{remark}\label{rk:Z(0)}  By \eqref{eq:clear} we see the reason why we do not require that the elements $Z$ of the renormalization
group $\Rc(M,L)$ satisfy $Z(0)=\mathrm{constant}$ (as it is done for the St\"uckelberg-Petermann group $\Rc_0$, cf.~Remark \ref{rm:SP}). Actually, terms like $Z(0)$, besides Prop.~\ref{prop:automorphism-by-Z}, will play a prominent role later (see Sect.~\ref{sec:RGF}).
\end{remark}

In Section~\ref{sec:pQFT} we show that in perturbation theory the unitary AMWI
is essentially equivalent to the AMWI as established in \cite{Brennecke08}.
If $\zeta_{\g}=\mathrm{id}_{\Floc(M,L)}$, we call \eqref{eq:uni-anom-MWI} the \emph{unitary Master Ward Identity}
(unitary MWI); in perturbation theory the latter is equivalent to the on-shell MWI introduced in
\cite{DB02,DF03} and \cite[Chap.~4]{Due19}.%
\footnote{For perturbative, scalar QED a proof of the equivalence of the unitary MWI to the on-shell MWI is given in \cite{DPR21}.} 
The (on-shell) MWI is a universal formulation of symmetries valid in classical field theory \cite{DF03}. For this reason, we call 
$\zeta:\g\mapsto \zeta_{\g}$ the \emph{anomaly map}.


The cocycle relation \eqref{eq:cocycle} is motivated by the following result:

\begin{proposition}\label{prop:AMWI-cocyclerelation}
Let $\Ups:\mathpzc{G}_c(M)\to\Rc(M,L)$ and let $\pi$ be a representation of $\fA(M,L)$ such that 
\be\label{eq:amwi}\pi\circ S_{(M,L)} \circ\g_L=\pi\circ S_{(M,L)} \circ \Ups_{\g}\ .
\ee
Then it holds that
\be\label{eq:S(cocycle)}
\pi\circ S_{(M,L)}\circ\Ups_{\g\h}
=\pi\circ S_{(M,L)}\circ\Ups_{\h}\circ(\Ups_{\g})^{\h}\quad\text{with}\quad
(\Ups_\g)^\h\doteq\h_L^{-1}\Ups_{\g} \h_L\ .
\ee
\end{proposition}
\begin{proof} 
We shall write $S\equiv \pi\circ S_{(M,L)}$ for simplicity. A straightforward application of \eqref{eq:amwi} and the definition of 
$(\Ups_\g)^\h$ gives
\be
S\circ\Ups_{\g\h}=S\circ(\g\h)_L=S\circ\g_L\circ\h_L
=S\circ\Ups_{\g}\circ\h_L
=S\circ\h_L\circ(\Ups_{\g})^{\h}=S\circ\Ups_{\h}\circ(\Ups_{\g})^{\h}\ .
\ee
\end{proof}
Note that in causal perturbation theory with $S$-matrices evaluated on off-shell field configurations \cite{DF04}, the
relation \eqref{eq:S(cocycle)} even implies the cocycle relation \eqref{eq:cocycle}. In contrast, the $S$-matrices in Proposition \ref{prop:AMWI-cocyclerelation} would correspond to perturbative $S$-matrices evaluated on \emph{on-shell configurations} and are not injective maps on $\Floc(M,L)$, so the cocycle relation does not automatically follow.


In the formulation of the cocycle condition, we used an action of
the group $\mathpzc{G}_c(M)$ on  renormalization group elements $Z\in\mathcal{R}(M,L)$ by
\be \label{eq:Zg}
Z\mapsto Z^\g \doteq \g_L^{-1}Z\g_L\quad\text{for}\quad \g\in\mathpzc{G}_c(M).
\ee
This is well defined due to the following result:

\begin{lemma}\label{lem:Zg-in-R}
Let $Z\in\Rc(M,L)$ and $\g\in\mathpzc{G}_c(M)$ and let $Z^\g$ be defined by \eqref{eq:Zg}. Then, $Z^\g \in\Rc(M,L)$.
\end{lemma}

\begin{proof}
We have to check the defining conditions of $\Rc(M,L)$ (Definition \ref{def1}) for $Z^{\g}$: Locality $(ii)$ 
is obtained straightforwardly by using that $\supp F\cap\supp H=\emptyset$ implies 
$\supp \g_\ast  F\cap\supp \g_\ast  H=\emptyset$.

To verify $(i)$ for $Z^{\g}$ let $\supp \g_{\ast }F\cap\supp Z=\emptyset$. Then
we obtain
\begin{equation*}
    \g_L Z^{\g}(F+G)=Z(\g_{\ast }F+g_L G)=\g_{\ast }F+Z(\g_L G)=\g_L(F+Z^{\g}(G))\ ,
\end{equation*}
hence $\supp Z^{\g}\subset (\tilde{ \g})^{-1}(\supp Z)$, where $\tildeg$ is defined in \eqref{eq:tilde-g}.


To check $(iii)$, we set $\psi_\ast  F\doteq F^{\psi}$ and $\psi_L F\doteq\delta L(\psi)+\psi_\ast  F$.
With this, the condition $(iii)$ says that $\psi_L$ commutes with elements of the renormalization group. We have
\be\label{eq:gL-psiL}
\g_L\psi_L=(\g'\psi)_L\g_L\ ,\quad\text{where}\quad \g'\psi\doteq\g(\phi+\psi)-\g\phi\in\Dc(M)\ ,
\ee
and find
\begin{equation}\label{eq:Zg-psiL}
    \begin{split}
        Z^{\g}\psi_L
        &=\g_L^{-1}Z\g_L\psi_L
        =\g_L^{-1}Z(\g'\psi)_L\g_L\\
        &=\g_L^{-1}(\g'\psi)_LZ\g_L
        =\psi_L\g_L^{-1}Z\g_L\\
        &=\psi_LZ^{\g}\ ,
        \end{split}
\end{equation}
by using $(iii)$ for $Z$ in the third equality sign.

$(iv)$ We first consider the case of a quadratic Lagrangian. Then
\begin{equation}\label{eq:shift}
    \begin{split}
        Z^{\g}\psi_{\ast }(F)&=Z^{\g}\bigl(\psi_L(F)-\delta L(\psi)\bigr)\\
        &=\g_L^{-1}Z\bigl(\g_L\psi_L(F)-\g_{\ast }\delta L(\psi)\bigr)\\
        &=\g_L^{-1}Z\bigl(\g_L\psi_L(F)\bigr)-\delta L(\psi)\\
        &=\psi_L\g_L^{-1}Z\g_L(F)-\delta L(\psi)\\
        &=\psi_{\ast }Z^{\g}(F)
    \end{split}
\end{equation}
where we used in the third line that $\g^{-1}_{\ast }\delta L(\psi)$ is an affine functional 
(hence we may apply Prop.~\ref{prop:Z}(ii)) and that $\g_L$ is an affine map, and in the 4th line the result \eqref{eq:Zg-psiL}.

For the general case we show that the transformation with $\g$ commutes with the addition of an interaction
according to \eqref{eq:renint}, namely we have
\begin{equation}
\begin{split}
    (Z^{W})^{\g}(F)&=\g^{-1}_{L+W}Z^W\g_{L+W}(F)\\
             &=\delta_{\g^{-1}}(L+W)+\g^{-1}_{\ast }Z^W\bigl(\delta_{\g}(L+W)+\g_{\ast }F\bigr)\\
             &=\delta_{\g^{-1}}L-W+\g^{-1}_{\ast }Z\bigl(\delta_{\g}L+\g_{\ast }(F+W)\bigr)\\
             &=(Z^{\g})^W(F)\ .
\end{split}
\end{equation}
As explained in Remark \ref{rem:Z-shift}, we have to show that
$(Z^{\g})^{-V}\,\psi_\ast  =\psi_\ast \,((Z^{\g})^{-V})$. From \eqref{eq:shift} we know that $(Z^{-V})^{\g}$ commutes with $\psi_\ast $,
therefore, this holds also for $(Z^{\g})^{-V}$. 

$(v)$ With simplified notation we have
\be
\begin{split}
L+Z^{\g}(F)&=L+\delta_{\g^{-1}}L+\g^{-1}_{\ast}Z(\delta_{\g}L+\g_{\ast}F)\\
&=\g^{-1}_{\ast}(L+Z(\delta_{\g}L+\g_{\ast}F))\ .
\end{split}
\ee
But by assumption, $L+Z(\delta_{\g}L+\g_{\ast}F)$ induces the same causal structure as $L+\delta_{\g}L+g_{\ast}F=\g_{\ast}(L+F)$. Hence also $Z^{\g}$ preserves the causal structure.
\end{proof}

For later purpose we prove the following results:

\begin{proposition}\label{prop:ZW-cocycle}
Let $\zeta$ be as introduced above. Then the following properties hold:
\begin{itemize}
    \item[($i$)] Let $\zeta\in\mathfrak{Z}(M,L)$ and $W\in\mathrm{Int}(M,L)$. Then, $\zeta^W:\g\mapsto(\zeta_\g)^W$ 
(Def.~\ref{def:ZW}) is an element of $\mathfrak{Z}(M,L+W)$, {\it i.e.},
\be
(\zeta_{\g\h})^W=(\zeta_\h)^W\,\h_{L+W}^{-1}(\zeta_\g)^W \h_{L+W}.
\ee
\item[$(ii)$] The advanced and retarded maps $\alpha_{V,\pm}$ \eqref{Bog1} and \eqref{Bog2} induce isomorphisms $\overline{\alpha}_{V,\pm}$ from
$\fA(M,L+V,\zeta^V)$ to $\fA(M,L,\zeta)$.
\end{itemize}
\end{proposition}

\begin{proof} $(i)$ According to Proposition~\ref{prop:ZW}, $\zeta^W$ takes values in $\Rc(M,L+W)$.
To shorten the notations we write $W$ in place of $W(f)$:

\be
\begin{split}
  (\zeta_\h)^W\,\h_{L+W}^{-1}(\zeta_\g)^W \h_{L+W}(F)&= 
  (\zeta_\h)^W\,\h_{L+W}^{-1}\bigl(\zeta_\g(\h_L F+\h_\ast W)-W\bigr)\\
  &= (\zeta_\h)^W\bigl(\delta_{\h^{-1}}L+(\delta_{\h^{-1}}W-\h_\ast^{-1} W)+\h_\ast^{-1}\zeta_\g(\h_L F+\h_\ast W)\bigr)\\
  &= (\zeta_\h)^W\bigl(\h_L^{-1}\zeta_\g(\h_L F+\h_\ast W)-W\bigr)\\
  &= \zeta_\h\h_L^{-1}\zeta_\g\h_L(F+ W)-W\\
  &= (\zeta_{\g\h})^W(F)\ ,
\end{split}
\ee
where the cocycle relation for $\zeta$ is used in the last step.

$(ii)$ The assertion follows from the observation that the unitary AMWI in $\fA(M,L+V,\zeta^V)$ is mapped to the
unitary AMWI in $\fA(M,L,\zeta)$. To wit, writing again $V$ instead of $V(f)$, we obtain

\begin{align}
\alpha_{V,+}\bigl(S_{(M,L+V)}(\g_{L+V}F)\bigr)&=S_{(M,L)}(V)^{-1}\,S_{(M,L)}\bigl(\delta_\g (L+V)+\g_\ast F+V\bigr)\nonumber\\
&=S_{(M,L)}(V)^{-1}\,S_{(M,L)}\bigl(\delta_\g L+\g_\ast (F+V)\bigr)
\intertext{and}
\alpha_{V,+}\bigl(S_{(M,L+V)}(\zeta_\g^V (F))\bigr)&=S_{(M,L)}(V)^{-1}\,S_{(M,L)}\bigl(\zeta_\g(F+V)\bigr)\ ,
\end{align}

thus the ideal $I_{\zeta^V}$ is mapped onto the ideal $I_\zeta$.  The same applies for $\alpha_{V,-}$.
\end{proof}
\begin{definition}[\textbf{Cocycle Equivalence}]\label{def:Cocycle Equivalence}
Two cocycles $\zeta,\zeta'\in\mathfrak{Z}(M,L)$ are equivalent, if there exists some $Z\in\Rc(M,L)$ with 
\be \label{eq:cohomology}
Z\zeta'_{\g}=\zeta_{\g}Z^{\g}\ ,\quad \g\in\G_c(M)\ .
\ee 
\end{definition}
One easily verifies that this is indeed an equivalence relation: symmetry follows from $(Z^\g)^{-1}=(Z^{-1})^\g$ 
and transitivity from $(Z_1Z_2)^\g=Z_1^\g\,Z_2^\g$.
Equivalences between cocycles lead to the following relations: 
\begin{proposition}\label{prop:betaZ-cocycle}
Let $\zeta\in\mathfrak{Z}(M,L)$ be a cocycle and $Z\in\Rc(M,L)$. Then
\begin{itemize}
    \item[($i$)] The map $\zeta': \g\mapsto \zeta_{\g}'$ with
    \begin{equation}
  \zeta_{\g}'\doteq Z^{-1}\zeta_{\g}Z^{\g}\in\Rc(M,L)  
\end{equation} 
is a cocycle, {\it i.e.}, it satisfies the relation \eqref{eq:cocycle}.
\item[($ii$)] Both $\beta_Z^{\mathrm{ret}}$ and  $\beta_Z^{\mathrm{adv}}$ (Prop.~\ref{prop:automorphism-by-Z})
induce isomorphisms $\overline{\beta}_Z^{\,\mathrm{ret/adv}}$ from $\fA(M,L,\zeta')$ to $\fA(M,L,\zeta)$, where $\zeta'$ is defined as in (i).

\item[$(iii)$] The isomorphism $\epsilon_Z\doteq\alpha_{A_{Z(0)},+}^{-1}\circ\beta^{\mathrm{ret}}_Z=\alpha_{A_{Z(0)},-}^{-1}\circ\beta^{\mathrm{adv}}_Z$ extends to a net isomorphism 
$$
\fA_{(M,L,\zeta')} \to \fA_{(M,L+A_{Z(0)},\zeta^{A_{Z(0)}})}\ .
$$ 

\end{itemize}
\end{proposition}
\begin{proof}
\begin{itemize}
    \item[$(i)$] Inserting the definitions and the cocycle relation for $\zeta$, we obtain
$$
\zeta'_{\g\h}=Z^{-1}\zeta_{\g\h}Z^{\g\h}=Z^{-1}\zeta_{\h}(\zeta_{\g})^{\h}(Z^{\g})^{\h}=Z^{-1}\zeta_{\h}Z^{\h}(Z^{-1}\zeta_{\g} Z^{\g})^{\h}=\zeta'_{\h}(\zeta'_{\g})^{\h}\ .
$$
\item[$(ii)$] The automorphisms $\beta_Z^{\mathrm{ret/adv}}$ map the ideals defined by the unitary AMWI into each other, as may be seen from the equation (denoting $S_{(M,L)}=S$)
\begin{equation}
        \beta^{\mathrm{ret}}_Z(S(\g_LF)^{-1}S(\zeta'_{\g}F))=S(Z\g_LF)^{-1}S(Z\zeta'_{\g}F)
        =S(\g_LZ^{\g}F)^{-1}S(\zeta_{\g}Z^{\g}F)\ ,
\end{equation}
for the retarded case, the advanced case follows  analogously.

\item[$(iii)$] Since $\supp(Z(F)-Z(0))=\supp F$, the isomorphism $\epsilon_Z$,
\begin{equation}\label{eq:eps-Z}
    \epsilon_Z(S_{(M,L,\zeta')}(F))=S_{(M,L+A_{Z(0)},\zeta^{A_{Z(0)}})}(Z(F)-Z(0))\ ,
\end{equation}
preserves the local subalgebras.
\end{itemize}
\end{proof}

We can use the cocycle relation \eqref{eq:cocycle} for the map $\g\mapsto \zeta_{\g}$ to define a modified action of $\G_c(M)$ on $\Floc(M,L)$ by
\begin{equation}\label{eq:g-zeta}
    \g_{\zeta}F\doteq  \g_L\zeta_\g^{-1}(F)
\end{equation}
namely
\begin{equation}\label{eq:g-zeta-1}
(\g\h)_{\zeta}=(\g\h)_L\zeta_{\g\h}^{-1}=\g_L \h_L (\zeta_{\g}^{\h})^{-1}\zeta_{\h}^{-1}=\g_L\zeta_{\g}^{-1}{\h}_L\zeta_{\h}^{-1}=\g_{\zeta}\h_{\zeta}\ .
\end{equation}
\begin{remark}\label{rm:uAMWI-alternative}
Note that the unitary AMWI can be written in terms of $\g_\zeta$ as $$S=S\circ \g_\zeta\,.$$
\end{remark}
We now consider the case that the cocycle $\zeta_{\g}$  in the  unitary AMWI is trivial for $\g\in \mathpzc{H}$ with $\mathpzc{H}\subset \mathpzc{G}_c(M)$. We have the following result:
\begin{proposition}
The set $\mathpzc{H}=\{\h\in \mathpzc{G}_c(M)\, | \, \zeta_{\h}=\mathrm{id}_{\Floc(M,L)}\}$ is a subgroup of $\mathpzc{G}_c(M)$. 
\end{proposition}
\begin{proof}
Let $\h_1,\h_2\in \mathpzc{H}$ then $(\h_i)_{\zeta}=(\h_i)_L$, $i=1,2$, hence
\begin{equation}
    (\h_1\h_2)_{\zeta}=(\h_1)_{\zeta}(\h_2)_{\zeta}=(\h_1)_L(\h_2)_L=(\h_1\h_2)_L
\end{equation}
thus $\zeta_{\h_1\h_2}=\mathrm{id}$ and $\h_1\h_2\in \mathpzc{H}$. Moreover, for $\h\in \mathpzc{H}$
\begin{equation}
    (\h^{-1})_{\zeta}=(\h_{\zeta})^{-1}=(\h_L)^{-1}=(\h^{-1})_L\ ,
\end{equation}
so also $\h^{-1}\in \mathpzc{H}$.
\end{proof}

\section{Time slice Axiom and the relative Cauchy evolution}\label{sec:tsa}
We now want to prove that the unitary AMWI implies the time slice axiom.
\begin{theorem}[Time slice]\label{th:TSA}
Let $N\subset M$ be a causally convex and globally hyperbolic neighbourhood of a Cauchy surface of $M$ with respect to the causal structure induced by the Lagrangian $L$
and let $\zeta\in\mathfrak{Z}(M,L)$ be a cocycle. Then
\begin{equation}\label{eq:TSA}
    \fA(N,L\!\restriction_{N},\zeta\!\restriction_{N})=\fA(M,L,\zeta) \ .
\end{equation}
\end{theorem}
Here, $\zeta\!\restriction_{N}\,:\,\mathpzc{G}_c(N)\to\Rc(N,L\!\restriction_{N})$ is obtained by the restriction
of $\zeta\,:\,\mathpzc{G}_c(M)\to\Rc(M,L)$.

\begin{proof}
We shall use the notation $S\equiv S_{(M,L,\zeta)}$ for simplicity throughout the proof and do not display explicitely the restrictions as in \eqref{eq:TSA}. We first restrict ourselves to a
 stationary spacetime $\RR\times\Sigma$ with time coordinate $t$ and with a time independent quadratic Lagrangian $L$ and a neighbourhood $N$ of the submanifold $t=0$.

Let $F\in\Floc(\RR\times \Sigma,L)$. We show that $S(F)$ can be written as a finite product of $S$-matrices whose arguments have supports in $J_-(N)\cap\supp F$. 

If $\supp F\subset J_-(N)$ this is trivial. If $\supp F\not \subset J_-(N)$, we argue as follows:
Let $a=\sup\{t|(t,x)\in\supp F\}$. Choose $\tau\in(0,a/2)$ and a compactly supported diffeomorphism $\chi$ with support in $J_+(\Sigma)$ which acts on the set $\{(t,x)\in J_-(\supp F)|t\in[a/2,a]\}$ as a translation of the time variable by $-\tau$ and maps the set $\{(t,x)\in J_-(\supp F)|t\in[0,a/2]\}$ into the set $\{(t,x)\in J_-(\supp F)|t\in[0,a/2-\tau]\}$. An example for such a diffeomorphism is
\be
\chi(t,x)=(t+\rho(t)\sigma (x),x)
\ee
with test functions $\rho\in\Dc(\RR_+), \sigma\in\Dc(\Sigma)$ with $\dot{\rho}>-1$, $\rho(t)=-\tau$, $t\in[a/2,a]$ and $0\le\sigma\le1$, $\sigma\equiv 1$ on $\{x|(t,x)\in J_-(\supp F)\text{ for some }t\in[a/2,a]\}$; see Figure \ref{fg:TSA}.

\begin{figure}[ht]
\centering
\begin{tikzpicture}[scale=0.7]
\coordinate (A) at (0,0) ; \coordinate (B) at (0,4) ;
\coordinate (C) at (0,8) ; \coordinate (C1) at (0,6.5) ;\coordinate (C2) at (0,2.5) ;\coordinate (C3) at (0.5,2.5) ;
\coordinate (D) at (0,8.5) ; \coordinate (E) at (0,-2) ; 
\coordinate (F) at (4.5,6) ;
\coordinate (G) at (4.5,1) ; 
\coordinate (H) at (4.5,4.5) ;
\coordinate (I) at (2.5,6) ; \coordinate (J) at (2.5,1) ;
\coordinate (Ia) at (2.5,4.5) ; \coordinate (Ka) at (6.5,4.5) ;
\coordinate (Ja) at (4.5,-1) ; \coordinate (Kb) at (4.5,6.5) ;
\coordinate (K) at (6.5,6) ; \coordinate (L) at (6.5,1) ; \coordinate (M) at (4.5,7.5) ;
\coordinate (N) at (0.5,1.5) ; \coordinate (O) at (0.5,-1.5) ;
\coordinate (P) at (16,-1.5) ; \coordinate (Q) at (16,1.5) ;\coordinate (R) at (0.3,0) ; \coordinate (T) at (16.5,0) ;
\coordinate (U) at (3.09,7.41) ; \coordinate (V) at (3.09,5.91) ;\coordinate (W) at (5.91,5.91) ;\coordinate (U1) at (3,6.9) ;
\coordinate (X) at (0.5,4) ; \coordinate (Y) at (9.32,4) ;\coordinate (Z) at (9.32,2.5) ;
\coordinate (X1) at (0.5,3.32) ; \coordinate (Y1) at (9.2,3.25) ; \coordinate (Y2) at (10.82,2.5) ;
\coordinate (U2) at (0.5,4.82) ;\coordinate (W1) at (5.91,7.41) ;
\coordinate (Z1) at (13.32,0) ;\coordinate (Z2) at (14.82,-1.5) ; 
\coordinate (Z3) at (11,1.5) ;\coordinate (Z4) at (7.1,1.5) ;
\fill[black!5] (N) -- (O) -- (P) -- (Q) -- (N);
\fill[black!15] (Ia) -- (J) to[out=-90, in=180] (Ja) to[out=0, in=-90] (L) -- (Ka) to[out=90, in=0] (Kb) to[out=180, in=90] (Ia);
\draw[->] (E) -- (D) ;
\draw (A) node {$\scriptstyle-$} node[left] {$\scriptstyle 0$} ;
\draw (B) node {$\scriptstyle-$} node[left] {$\scriptstyle a/2$} ;
\draw (C) node {$\scriptstyle-$} node[left] {$\scriptstyle a$} ;
\draw (C1) node {$\scriptstyle-$} node[left] {$\scriptstyle a-\tau$} ;
\draw (C2) node {$\scriptstyle-$} node[left] {$\scriptstyle a/2-\tau$} ;
\draw (D)  node[above right] {$t$}; 
\draw[dashed] (N) -- (Z4); \draw[dashed] (Z3) -- (Q);
\draw[dashed] (O) -- (P);
\draw (F) ++(0:2cm) arc(0:180:2cm) ;
\draw[dashed] (H) ++(0:2cm) arc(0:180:2cm) ;
\draw (G) ++(0:2cm) arc(0:-180:2cm) ;
\draw (I) -- (J) ; \draw (K) -- (L) ;
\draw (M)  node {$\supp F$};
\draw (F)  node[below] {$\chi(\supp F)$};
\draw[thick] (R) -- (T);
\draw (N)  node[below right] {$N$};
\draw (T)  node[right] {$\{0\}\times\Sigma$};
\draw[thick,->] (U) -- (V) ; \draw[thick,->] (Y) -- (Z) ;
\draw (U1)  node[right] {$\tau$}; \draw (Y1)  node[right] {$\tau$};
\draw[dashed] (V) -- (X1); \draw[dashed] (W) -- (Z) to[out=-45, in=155] (Z3) to[out=-25, in=135] (Z1);
\draw (X) -- (Y);\draw[dashed] (C2) -- (Z); 
\draw(U) -- (U2);\draw[thick,dotted] (Y) -- (Z2); \draw (W1) -- (Y);
 \draw (Y2)  node[right] {$J_-(\supp F)$};
 \draw (10.85,1.90)  node[below left] {$\chi(J_-(\supp F))$};
\end{tikzpicture}
\caption{\small Illustration of the proof of Theorem \ref{th:TSA}: $N$ is filled by a light gray,
$\supp F$ and $J_-(\supp F)\cap\{t\ge a/2\}$ are sketched by a solid line, the remaining part of $J_-(\supp F)$ 
is sketched by a dotted line, $\chi(\supp F)$ is filled 
by a  dark gray, $\chi\bigl(J_-(\supp F)\bigr)=J_-^{L+\delta_{\chi}L}\bigl(\chi(\supp F)\bigr)$ and 
$\chi\bigl(J_-(\supp F)\cap\{t\ge a/2\}\bigr)$ are 
sketched by a dashed line. Note that the latter set is obtained from $J_-(\supp F)\cap\{t\ge a/2\}$ by a translation of the time 
variable by $(-\tau)$; however, $J_-(\supp F)$ and $\chi(J_-(\supp F))$ join at $t=0$, since there they leave the support of $\chi$.}
\label{fg:TSA} 
\end{figure}
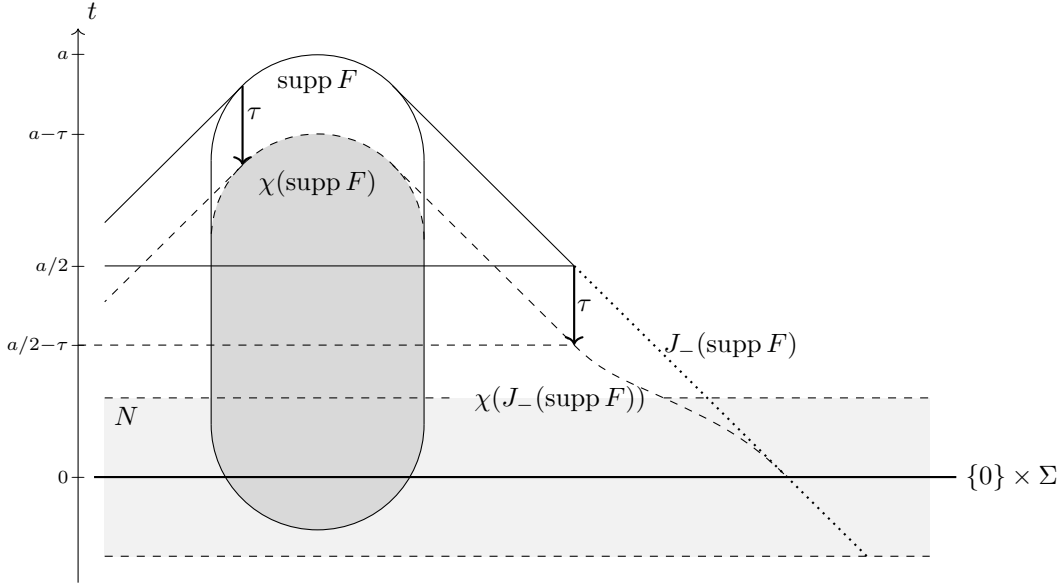

Next we prove the following auxiliary result:

\begin{lemma}
Let $(M,L)$ be a stationary dynamical spacetime, \ie the Lagrangian $L$ is time independent. In addition,
let $U\subset M$ be compact and $\chi(t,x)=(t+v,x)$ for all $(t,x)\in U$, where $v\in\RR$ is constant.
Then the variation $\delta_\chi L\in\Floc(M,L)$ satisfies $\supp\delta_\chi L\cap\chi(U)=\emptyset$.
\end{lemma}

\begin{proof}
For any $\psi,\phi\in\Ec(M)$ with $\supp\psi\subset U$ we have to show that $\delta_\chi L[\phi+\psi]=\delta_\chi L[\phi]$,
where $\delta_\chi L[\phi]=L(f\circ\chi)[\phi\circ\chi]-L(f)[\phi]$ with $f\in\Dc(M)$ with $f\equiv 1$ on $\supp\chi$. 
Since $L$ does not explicitly depend on $t$,
we may introduce $\Lc\bigl(\phi(t,x),d\phi(t,x)\bigr)$ by
$$
\int\Lc\bigl(\phi(t,x),d\phi(t,x)\bigr)\,f(t,x)\,d\mu_g(t,x)\doteq L(f)[\phi]\quad\forall f,\phi\ .
$$
By setting $y\doteq\chi(t,x)$ we obtain
\begin{align}
  L(f\circ\chi)&[(\phi+\psi)\circ\chi]- L(f\circ\chi)[\phi\circ\chi]\nonumber\\ 
=&\int\Bigl(\Lc\bigl((\phi+\psi)(y),d(\phi+\psi)(y)(D\chi)(\chi^{-1}(y))\bigr)\nonumber\\
&\quad -\Lc\bigl(\phi(y),d\phi(y)(D\chi)(\chi^{-1}(y))\bigr)\Bigr)f(\chi^{-1}(y))\,d\mu_g(\chi^{-1}(y))\ .
\label{eq:translation}
\end{align}
Since $\supp\psi\subset\chi(U)$, we may restrict the integration to $y\in\chi(U)$,
hence $f(\chi^{-1}(y))=1$. In addition, since $\chi\vert_U$ is a time translation,
we have $(D\chi)(\chi^{-1}(y))=\mathbf{1}$ and $d\mu_g(\chi^{-1}(y))=d\mu_g(y)$ for $y\in\chi(U)$.
Hence, $\chi$ disappears from \eqref{eq:translation}, that is,
\eqref{eq:translation}$=L(f)[\phi+\psi]- L(f)[\phi]$.
\end{proof}

\noindent\emph{Continuation of the proof of Theorem \ref{th:TSA}.}
From the unitary AMWI we obtain
\be
S(F)=S(\delta_{\chi}L+\chi_{\ast }\zeta_{\chi}^{-1}(F))\ .
\ee
Since $\zeta_{\chi}\in\Rc(M,L)$ and $L$ is quadratic, $\zeta_{\chi}^{-1}(0)$ is a constant (Prop.~\ref{prop:Z}(ii)) and therefore $\supp\zeta_{\chi}^{-1}(F)=\supp F$ (by using \eqref{eq:invar-suppZ}). 

Then $\supp\chi_{\ast }\zeta_{\chi}^{-1}(F)=\chi(\supp F)$
and
\begin{equation}
    \chi(\supp F)\cap J_+(N)\subset J_-(\supp F)\cap\{t\le a-\tau\}\ ,
\end{equation}
and since by assumption $L$ is stationary, 
\be\supp\delta_{\chi}L\subset J_+(\Sigma)\setminus\chi\bigl( J_-(\supp F)\cap\{t\ge a/2\}\bigr)
\subset J_+(\Sigma)\setminus\bigl(\chi(J_-(\supp F))\cap\{t\ge a/2-\tau\}\bigr)
\ee 
by using the Lemma.
We decompose $\delta_{\chi}L=H_++H_-$ 
such that 
\be
\supp H_-\subset J_+(\Sigma)\cap\{t<a/2\}
\ ,\  
\supp H_+\cap \chi(J_-(\supp F))=\0\ .
\ee
Since $\chi(J_-(\supp F))=J_-^{L+\delta_{\chi}L}(\chi(\supp F))$ and $H_+$ does not change the metric in this region, we can replace $\delta_{\chi}L$ by $H_-$ and get
$\supp H_+\cap J_-^{L+A_{H_-}}(\supp \chi_{\ast }\zeta_{\chi}^{-1}(F))=\0$, thus we obtain the factorization
\be\label{eq:time-slice2}
S(F)=S(\delta_{\chi}L)S(H_-)^{-1}S(H_-+\chi_{\ast }\zeta_{\chi}^{-1}(F))\ .
\ee
Again by the unitary AMWI, $S(\delta_{\chi}L)=S(\zeta_{\chi}(0))$. $\zeta_{\chi}(0)$ is a constant functional and thus has empty support. Hence the arguments of all $S$-matrices on the right hand side have compact supports contained in 
$(J_+(N)\cap \{t\leq a-\tau\})\cup (J_-(N)\cap\supp F)$.

We may now iterate the argument and find after a finite number of steps, that $S(F)$ can always be written as a product of $S$-matrices whose arguments have support in $J_-(N)\cap\supp F$.

We then repeat the procedure for the time reversed situation and arrive at the desired result for the special case of a free stationary Lagrangian.

For a generic choice of a dynamical spacetime $(M,L)$, a cocycle $\zeta$ and a neighborhood $N$ of a Cauchy surface $\Sigma$ we choose a foliation $M=\RR\times \Sigma$ with $\{0\}\times \Sigma\equiv\Sigma$ where the metric associated to $L$ assumes the form
\begin{equation}
    g(t,x)=a(t,x)\,dt^2-h_t(x)
\end{equation}
with a smooth positive function $a$ and a smooth family $h_t$ of Riemann metrics on $\Sigma$. We then set
\begin{equation}
    L_0=\frac12 g_0^{-1}(d\phi,d\phi)d\mu_{g_0}
\end{equation}
with
\begin{equation}
    g_0(t,x)=a(0,x)dt^2-h_0(x)\ .
\end{equation}
Let us define $V\doteq L-L_0\in\mathrm{Int}(M,L_0)$, hence $-V\in\mathrm{Int}(M,L)$.  
According to Proposition~\ref{prop:ZW-cocycle}, $\zeta^{-V}$ is an $\Rc(M,L_0)$-valued cocycle.
We now use Theorem \ref{thm2} for a faithful representation $\pi_0$ of $\fA(M,L_0,\zeta^{-V})$. Since $L_0$ is stationary and quadratic, we know from the first part of the proof that the time slice axiom holds for $\pi_0$. Since $\Sigma$ is a Cauchy surface for both, $g$ and $g_0$, we conclude that, for any splitting $V=V_++V_-$ into a past compact and a future compact part, the representation $\pi\doteq\pi_0\circ\alpha_{V_+,V_-}$ has the property
\begin{equation}
    \pi(\fA(N,L,\zeta))=\pi(\fA(M,L,\zeta))\ .
\end{equation}
Since with $\pi_0$ also $\pi$ is faithful, we conclude that 
\begin{equation}
    \fA(N,L,\zeta)=\fA(M,L,\zeta)\ .
\end{equation}
\end{proof}

Granted the time slice axiom, we can briefly discuss the relative Cauchy evolution that was defined and implemented in the more general framework of locally covariant quantum field theory \cite{BFV03}, to which we refer for more details. 
The basic advanced and retarded maps in \eqref{Bog1} and \eqref{Bog2} interpolate between the algebra $\fA(M,L',\zeta)$ and the \emph{perturbed} algebra $\fA(M,L,\zeta^V)$, with $L=L'+V$ (Prop.~\ref{prop:ZW-cocycle}),
where the interaction $V$ changes also the metric $g$ in a compact region of $M$, but the resulting transformed metric $g+h$ remains globally hyperbolic. Let $N_{\pm}$ be two neighborhoods of Cauchy surfaces in $(M,L')$, one in the past and the other in the future of the support of the perturbation $V$, w.r.t. the causal structure induced by $L'$. Let $\chi_{\pm}$ denote the isometric embeddings of $N
_{\pm}$ into $M$. We combine the basic retarded and advanced maps \eqref{Bog1} and \eqref{Bog2} with the monomorphisms $\alpha_{\chi_{\pm}}$ induced by these embeddings. Since Theorem \ref{th:TSA} holds, these combinations give rise to bijective maps.  Let us call them  as $\tilde{\alpha}_{V,-} \doteq \alpha_{\chi_-}\circ \overline{\alpha}_{V,-}^{-1}$ and $\tilde{\alpha}_{V,+} \doteq \alpha_{\chi_+}\circ \overline{\alpha}_{V,+}^{-1}$, where the extensions $\overline{\alpha}_{V,\pm}$ are the ones from Proposition \ref{prop:ZW-cocycle}(ii). We obtain:
\begin{align*}
\fA(M,L',\zeta)\,\,\xrightarrow{\overline{\alpha}_{V,\pm}^{-1}}&\,\,\fA(M,L'+V,\zeta^V)\\
&=\fA(N_\pm,(L'+V)\!\restriction_{N_\pm},\zeta^V\!\restriction_{N_\pm})\\
&=\fA(N_\pm,L'\!\restriction_{N_\pm},\zeta\!\restriction_{N_\pm})\xrightarrow{\alpha_{\chi_\pm}}\,\,\fA(M,L',\zeta)\ .
\end{align*}
Hence, it is straightforward to compute that the relative Cauchy evolution automorphism $\beta_V$ of $\fA(M,L',\zeta)$ defined via these maps is given by
$$
\beta_V \doteq \tilde{\alpha}_{V,+}\circ \tilde{\alpha}_{V,-}^{-1}= \mathrm{Ad}\bigl(S_{(M,L',\zeta)}(V(f))^{-1}\bigr)\ .
$$
\begin{remark}
As expected, the relative Cauchy evolution \ie \emph{scattering morphism} as also called in \cite{BFV03}, is indeed implemented by the scattering matrices. This also justifies, \emph{a posteriori}, the interpretation of the symbols $S_\bullet(F)$ as scattering matrices, inducing a perturbation of the theory by an intermediate change in the interaction and metric in a compact region.
\end{remark}

Our new formulation of the relative Cauchy evolution also allows one to obtain the stress-energy tensor. If $\delta_s L$ is the change in the Lagrangian induced by a change in the metric from $g$ to $g_s$, the S-matrix $S_{(M,L,\zeta)}(\delta_s L)$ may be interpreted as the time-ordered \emph{exponentiated} stress-energy tensor, so to obtain the stress-energy tensor, we need to be able to differentiate $S_{(M,L,\zeta)}(\delta_s L)$ in an appropriately regular Hilbert space representation $\pi$ for a smooth one-parameter family of metrics $g_s$ with , $g_0=g$,  $s\in[0,1]$ and $(g_s-g)$ compactly supported, satisfying the geometric assumptions as in Section 4.1 of \cite{BFV03}. 
Assume that $\pi$ is such that $\pi(S_{(M,L,\zeta)}(\delta_sL))$ is differentiable as a function of $s$ in the sense of quadratic forms on a dense domain $\mathcal{V}$, \ie we require that
\[
\frac{d}{ds}\langle\theta,\pi\bigl(S_{(M,L,\zeta)}(\delta_s L)\bigr)\theta\rangle\Big\vert_{s=0}=\int_M t^{\mu\nu}(x)h_{\mu\nu}(x)\sqrt{-g}dx
\]
with $\theta\in\mathcal{V}$ and $h_{\mu\nu}=\frac{d}{ds}(g_s)_{\mu\nu}\big\vert_{s=0}$,
and that the right-hand side defines an operator-valued distribution  $T^{\mu\nu}$ with matrix elements $t^{\mu\nu}$ and domain $\mathcal{V}$.

One should choose this domain according to the physical problem at hand, so that it remains invariant under some appropriate class of observables. Note that our current formulation is an improvement over the results of \cite{BFV03}, since there one could only reconstruct the derivation given by the commutator with the stress-energy tensor, while here we obtained the stress-energy tensor itself.

Covariant conservation of the stress-energy tensor can be obtained as follows. 
Let $(\chi_s)$ be a 1-parameter group of compactly supported diffeomorphisms and consider the induced change $(g_s-g)$ of the metric and $\delta_sL$ of the Lagrangian.  We interpret the equation
\be\label{eq:cov:cons}
S_{(M,L,\zeta)}\bigl(\delta_sL\bigr)=1\,,
\ee
for all $\chi_s \in \mathrm{Diff}_c(M) $, as the finite version of the covariant conservation of the stress-energy tensor. Namely, differentiating this property in the representation $\pi$
results in $\nabla_\mu T^{\mu\nu}=0$, in the sense of operator-valued distributions (see, {\it e.g.}, \cite{Hawking-Ellis}) . 

To see when the stress energy tensor is covariantly conserved, note that
by the unitary AMWI, we have that
\[
S_{(M,L,\zeta)}\bigl(\delta_s L\bigr)=S_{(M,L,\zeta)}\bigl( \zeta_{\chi_s}(0)\bigr)\,,
\]
so \eqref{eq:cov:cons} holds if the cocycle $\zeta$ is trivial for the subgroup $\mathrm{Diff}_c(M)$ of $\G_c(M)$, 
\ie in the absence of anomalies for diffeomorphisms.

\section{Symmetries of the Lagrangian and the anomalous Noether theorem}\label{sec:symL}
In this section we will show how 
unbroken symmetries of the Lagrangian (those that survive quantization and don't lead to non-trivial anomalies) give rise to the unitary action on our algebras that can be interpreted as the quantum version of Noether's theorem (Corollary \ref{cor:unit:noether}), which follows from our main result: the Anomalous Noether Theorem (Theorem \ref{theorem:anomalousnoethertheorem}).

Let $\G(M)$ denote the unit component of the semidirect product of the diffeomorphism group of $M$ and the group of affine field redefinitions, as described in Definition \ref{def:G}, but without the restriction to compact support. Note that $\G_c(M)$ is a normal subgroup of $\G(M)$. Each element $\g\in\G(M)$ has an action on $\Floc(M,L)$ that can be approximated \emph{locally} by elements of $\G_c(M)$ in the sense of the following definition.
\begin{definition}\label{def:GgO}
Let $\Oc\subset M$ be a relatively compact subset.  $\G(\g,\Oc)$ denotes the set of all  $\g'\in\G_c(M)$ such that $\g'_{\ast}F=\g_{\ast}F$ for all $F\in\Floc(M,L)$ with $\supp F\subset\Oc$. 
\end{definition}
Note that this set is non-empty, since one can always use partitions of unity to construct these local approximations.

For a given Lagrangian $L$ on $M$, we consider the subgroup of dynamical symmetries
\begin{equation}
    \HL=\{\h\in\G(M)|\h_{\ast}L\sim L\}\ .
\end{equation}
The group of dynamical symmetries $\h$ is represented by automorphisms $\gamma_{\h}$ of $\fA(M,L)$, induced by their action on the generators,
\begin{equation}
    \gamma_{\h}\left(S(F)\right)\doteq S(\h_{\ast}F)\ .
\end{equation}
The ideal $I_{\zeta}$ of $\fA(M,L)$ induced by the cocycle $\zeta$ is generated by the elements
\begin{equation}
    P_\zeta(\g,F)\doteq S(\g_LF)S(\zeta_{\g}F)^{-1}-1\ ,\ \g\in\G_c(M)\ ,\ F\in\Floc(M,L)\ .
\end{equation}
We apply the automorphism $\gamma_{\h}$ of $\fA(M,L)$ to these elements and obtain
\begin{align}
    \gamma_{\h}(P_\zeta(\g,F))&=S(\h_{\ast }\g_LF)S(\h_{\ast }\zeta_{\g}F)^{-1}-1=
    S((\h\g \h^{-1})_L\h_{\ast }F)S((\h\zeta)_{\h\g \h^{-1}}\h_{\ast }F)-1\nonumber\\
    &=P_{\h\zeta}(\h\g\h^{-1},\h_\ast F)\label{eq:gamma-P}
\end{align}
with the action 
\begin{equation}\label{eq:sym-acts-cocy}
    (\h,\zeta)\mapsto \h\zeta\ ,\ (\h\zeta)_\g\doteq\h_{\ast }\zeta_{\h^{-1}\g \h}
    \h^{-1}_{\ast }
\end{equation}
of $\HL$ on $\mathfrak{Z}(M,L)$. Here we used that $\G_c(M)$ is a normal subgroup of $\G(M)$ and that
\begin{equation}
    \h_{\ast}\delta_{\g}L=\delta_{\h\g\h^{-1}}\h_{\ast}L
\end{equation}
as well that $\h_{\ast}L\sim L$ by assumption.

To verify the consistency of the definition \eqref{eq:sym-acts-cocy}, first note that $(\h\zeta)_\g\in\Rc(M,L)$ as one finds
straightforwardly by checking the defining properties of $\Rc(M,L)$; in addition, one verifies that $(\h\zeta)$ satisfies again the cocycle relation,
\be\label{eq:cocycle-hzeta}
(\h\zeta)_{\g_1\g_2}=(\h\zeta)_{\g_2}\,\g_{2L}^{-1}\,(\h\zeta)_{\g_1}\,\g_{2L}\ ,\quad\g_1,\g_2\in\G_c(M),\,\h\in\HL,
\ee
and associativity:
\be
\h_1(\h_2\zeta)=(\h_1\h_2)\zeta \ ,\quad \h_1,\h_2\in\HL\ .
\ee
The latter is straightforward, while  \eqref{eq:cocycle-hzeta} is obtained by inserting the definitions and the
cocycle relation for $\zeta$ and by using
\be
\g_L\h_\ast =h_\ast(\h^{-1}\g\h)_L\ ,\quad\h\in\HL,\,\g\in\G_c(M)\ ,
\ee
which relies on $\h_\ast L\sim L$.

From \eqref{eq:gamma-P} we conclude that $\gamma_{\h}(I_{\zeta})=I_{\h\zeta}$ and, therefore, 
$\gamma_{\h}$ induces an isomorphism $\overline{\gamma}_{\h}$ of the quotient algebras
\begin{equation}\label{eq:flow:anomaly}
    \overline{\gamma}_{\h}:\fA(M,L,\zeta)\to\fA(M,L,\h\zeta)\ .
\end{equation}
Hence $\HL$ induces a flow on the space of theories with a given Lagrangian $L$, but possibly different cocycles $\zeta$.
This is quite analogous to the flow of the renormalization group under scalings for scale-invariant Lagrangians as we will discuss in the next section.

We want to understand better the action of $\HL$ on cocycles. We use the fact that the symmetries of the Lagrangian can be locally approximated by compactly supported symmetries. 
We find the following relations:
\begin{proposition}\label{prop:h-zeta} 
Let $\Oc\subset M$ be an open relatively compact region and $\h\in\HL$. Then for $\supp \g,\supp F\subset\tildeh(\Oc)$, with $\tildeh$ defined as in \eqref{eq:tilde-g} and $\h'\in\G(\h,\Oc)$ as in Definition \ref{def:GgO}, we have:
\begin{align}
    (\h\zeta)_{\g}F&=(\h'\zeta)_{\g}F\ , \\
    (\h'\zeta)_{\g}&=Z^{-1}\zeta_{\g}Z^\g\quad\text{with}\quad Z\doteq\zeta_{(\h')^{-1}}\ , \end{align}
where $(\h'\zeta)$ is defined by the formula \eqref{eq:sym-acts-cocy}.
\end{proposition}
\begin{proof} 
$\h'\in\G(\h,\Oc)$ is equivalent to $(\h')^{-1}\in\G(\h^{-1},\tildeh(\Oc))$, hence $(\h')^{-1}_{\ast} F=\h^{-1}_{\ast}F$. Moreover, $\widetilde{\h'\h^{-1}}$ is a diffeomorphism of $M$ which is the identity on $\tildeh(\Oc)$, hence due to
$\supp\g\subset\tildeh(\Oc)$
\begin{equation}\label{eq:h'-h}
    \g\h'\h^{-1}=\h'\h^{-1}\g 
\end{equation}
and thus $(\h')^{-1}\g\h'=\h^{-1}\g\h$. Now $\supp \h^{-1}\g\h\subset\Oc$, hence $\supp\zeta_{\h^{-1}\g\h}\h^{-1}_{\ast}F\subset\Oc$.
Inserting these relations into \eqref{eq:sym-acts-cocy} yields the first equality.

We now use that $\h_{\ast}L\sim L$ for $\h\in\HL$. Therefore $\supp\delta_{\h'}L\cap\tildeh(\Oc)=\0$ for $\h'\in\G(\h,\Oc)$ and with
\begin{equation}
\delta_{\h'}L+\h'_{\ast}\delta_{(\h')^{-1}}L=0
\end{equation}
we get $\supp\delta_{(\h')^{-1}}L\cap\Oc=\0$. Hence
$\supp\delta_{(\h')^{-1}}L\cap\supp\zeta_{(\h')^{-1}\g\h'}=\0$. 
Then with
\begin{equation}
    \h'_{\ast}G=\h'_L G-\delta_{\h'}L\ ,\ (\h')^{-1}_{\ast}G=(\h')^{-1}_L G-\delta_{(\h')^{-1}}L
\end{equation}
 and
\begin{equation}
    \zeta_{(\h')^{-1}\g\h'}\bigl((\h')^{-1}_{\ast}G\bigr)=\zeta_{(\h')^{-1}\g\h'}\bigl((\h')^{-1}_L G\bigr)-\delta_{(\h')^{-1}}L
\end{equation}
for $G\in\Floc(M,L)$ (by using \eqref{eq:suppZ}) 
we get
\begin{equation}\label{eq:h'zeta1}
    (\h'\zeta)_{\g}=\h'_L\zeta_{(\h')^{-1}\g\h'}(\h')^{-1}_L\ .
\end{equation}
Using the cocycle identity we finally get the second equation
\begin{equation}\label{eq:h'zeta2}
        (\h'\zeta)_{\g}=\h'_L\zeta_{\h'}(\h')^{-1}_L\zeta_{\g}\zeta_{(\h')^{-1}}^{\g}
        =Z^{-1}\zeta_{\g}Z^\g\quad\text{with}\quad Z\doteq\zeta_{(\h')^{-1}}\ .
\end{equation}
\end{proof}
We observe that, locally, the flow of $\HL$ is induced by the anomaly map $\zeta$ restricted to local approximations of 
$\HL$. We now study the consequences of the unitary AMWI.



\begin{theorem}[Anomalous Noether Theorem]\label{theorem:anomalousnoethertheorem}
Let $(M,L)$ be a dynamical spacetime, equipped with a cocycle $\zeta\in\mathfrak{Z}(M,L)$. For $\h\in\HL$ and  any choice of $\h'\in\G(\h,\Oc)$, with $\Oc\subset M$  relatively compact and causally convex,  there exists a unitary $U\in\fA(M,L,\h\zeta)$ such that
\begin{equation}\label{eq:anomalousnoether}
    \overline{\gamma}_{\h}(S_{(M,L,\zeta)}(F))=\mathrm{Ad}(U)\bigl(\overline{\beta}^{\mathrm{ret}}_{(\h\zeta)_{\h'}}(S_{(M,L,(\h\zeta)')}(F))\bigr)
\end{equation}
for all $F\in\Floc(M,L)$ for which $\supp F\subset\Oc$,
 with
 \begin{equation}\label{eq:hzeta'}
     (\h\zeta)'_{\g}\doteq(\h\zeta)_{\h'}^{-1}(\h\zeta)_{\g}(\h\zeta)_{\h'}^{\g}=((\h')^{-1}\h\zeta)_{\g}^{-\delta_{(\h')^{-1}}L}\quad
     \ .
 \end{equation}
 and $(\h\zeta)'_{\g}=\zeta_{\g}$ if $\supp\g\subset\Oc$.
\end{theorem}


\begin{proof}  The unitary AMWI yields for $\h\in\HL$
\begin{equation}
    \overline{\gamma}_{\h}(S_{(M,L,\zeta)}(F))=S(\h'_{\ast}F)=S\bigl((\h\zeta)_{\h'}(F+\delta_{(\h')^{-1}}L)\bigr)\ ,
\end{equation}
with $S\doteq S_{(M,L,\h\zeta)}$. We have $\supp\delta_{(\h')^{-1}}L\cap\Oc=\0$. Since $\Oc$ is causally convex, we decompose (see the corresponding procedure in Theorem \ref{thm:adiabatic})
\begin{equation}
 \delta_{(\h')^{-1}}L=\sum_{i=1}^5\sum_{\pm}Q_i^{\pm}   
\end{equation}
with $\supp Q_i^+\cap J_-^{L+A_{Q_{<i}}}(\Oc)=\0$ and $\supp Q_i^-\cap J^{L+A_{Q_{<i}+Q_i^+}}_+(\Oc)=\0$, $Q_{<i}\doteq\sum_{j<i,\pm}Q_j^{\pm}$, in particular $Q_{<1}=0$ and $Q_{<6}=\delta_{(\h')^{-1}}L$. 
We then use stepwise causal factorization for $S'(F+Q_{<6})$ with $S'\doteq\overline{\beta}^{\mathrm{ret}}_{(\h\zeta)_{\h'}}\circ S_{(M,L,(\h\zeta)')}$ (noticing that by Prop.~\ref{prop:automorphism-by-Z} and \ref{prop:betaZ-cocycle}, we have 
 $\overline{\beta}^{\mathrm{ret}}_{(\h\zeta)_{\h'}}(S_{(M,L,(\h\zeta)')}(F))\in\fA(M,L,\h\zeta)$) 
starting with $i=5$ and obtain 
\begin{align*}
  &\overline{\gamma}_{\h}(S_{(M,L,\zeta)}(F))=S((\h\zeta)_{\h'}(0))\,S'(F+Q_{<6})\\
  &=S'(Q_{<6})^{-1}\,S'(F+Q_{<5}+Q_5^+)\,S'(Q_{<5}+Q_5^+)^{-1}\,S'(Q_{<6})\\
  &=S'(Q_{<6})^{-1}\,S'(Q_{<5}+Q_5^+)\,S'(Q_{<5})^{-1}\,S'(F+Q_{<5})\,S'(Q_{<5}+Q_5^+)^{-1}\,S'(Q_{<6})\\
  &=\ldots\\
  &=\ldots S'(Q_{<2}+Q_2^+) S'(Q_{<2})^{-1}S'(Q_1^+)S'(F)S'(Q_1^+)^{-1}S'(Q_{<2})S'(Q_{<2}+Q_2^+)^{-1}S'(Q_{<3})\ldots,
\end{align*}
where we have also taken into account that
$$
S\bigl((\h\zeta)_{\h'}(0)\bigr)S'(Q_{<6})=S\bigl((\h\zeta)_{\h'}(\delta_{(\h')^{-1}}L)\bigr)=
\gamma_{\h}(S_{(M,L,\zeta)}(0))=1
$$
by the unitary AMWI for $F=0$. We end up with \eqref{eq:anomalousnoether}, where
\begin{equation}
    U=\prod_{i=0}^{4}S'(Q_{<6-i})^{-1}S'(Q_{<5-i}+Q_{5-i}^+)\ .
\end{equation}
In the formula for $U$ one may replace $S'$ by $S\circ(\h\zeta)_{\h'}$.

The second equation in \eqref{eq:hzeta'} is verified as follows
\begin{equation*}
\begin{alignedat}{2}
    (\h\zeta)'_g\quad&=\quad &&(\h\zeta)_{\h'}^{-1}(\h\zeta)_{\g}(\h\zeta)_{\h'}^{\g}\\
    &\overset{\mathclap{\eqref{eq:cocycle}}}{=}\quad\,&&(\h\zeta)_{(\h')^{-1}}^{\h'}(\h\zeta)_{\h'\g}\\
    &=\quad\, &&(\h')_L^{-1}(\h\zeta)_{(\h')^{-1}}(\h\zeta)_{\h'\g}^{(\h')^{-1}}\h'_L\\
    &\overset{\mathclap{\eqref{eq:cocycle}}}{=}\quad\, &&(\h')_L^{-1}(\h\zeta)_{\h'\g(\h')^{-1}}\h'_L\\
    &\overset{\mathclap{\eqref{eq:sym-acts-cocy}}}{=}\quad\, &&(\h')_L^{-1}\h_{\ast}\zeta_{\h^{-1}\h'\g(\h')^{-1}\h}\h_{\ast}^{-1}\h'_L\\
    &\overset{\mathclap{\eqref{eq:g_L}}}{=}\quad\, &&\delta_{(\h')^{-1}}L+
    ((\h')^{-1}\h\zeta)_{\g}(\bullet-\delta_{(\h')^{-1}}L)\\
    &=\quad\, &&((\h')^{-1}\h\zeta)_{\g}^{-\delta_{(\h')^{-1}}L}\ .
\end{alignedat}
\end{equation*}
To prove the last statement, let $\supp\g\subset\Oc$. We set $\kk\doteq(\h')^{-1}\h$ and use that $\supp\kk\cap\Oc=\0$. 
Then $\kk\g=\g\kk$ and
$\supp (\kk\zeta)_{\g}\subset\supp\g\subset\Oc$. 
We split
\be
F=\kk_{\ast}F-\delta_\kk F
\ee
and use that $\supp \delta_\kk F\cap\Oc=\0$. 
Then
\begin{equation}
    \begin{split}
    (\kk\zeta)_{\g}^{-\delta_{(\h')^{-1}}L}(F)&=(\kk\zeta)_{\g}(\kk_{\ast}F-\delta_{(\h')^{-1}}L)-\delta_\kk F+\delta_{(\h')^{-1}}L\\
    &=(\kk\zeta)_{\g}\bigl(\kk_{\ast}F\bigr)-\delta_\kk F\\
    &=\kk_{\ast}\zeta_{\g}(F)-\delta_\kk F\\
    &=\kk_{\ast}(\zeta_{\g}(F)-F)+F\\
    &=\zeta_{\g}(F)-F+F\\
    &=\zeta_{\g}(F)\ ;
    \end{split}
\end{equation}
in the second step we have taken into account that $\supp\delta_{(\h')^{-1}}L\cap\Oc=\0$ (see the proof of Prop.~\ref{prop:h-zeta})
and in the second last step that $\supp(\zeta_{\g}(F)-F)\subset\supp\zeta_\g\subset\Oc$ \eqref{eq:supp(Z(F)-F)}. This finishes the proof.
\end{proof}


\begin{remark} A simple example of implementation of the previous theorem can be found in Subsection~\ref{subsect:comp-anomalies}.
\end{remark}
We now turn to the group of unbroken symmetries, $\HLz\doteq\{\h\in\HL|\h\zeta=\zeta\}$. Elements of this group induce automorphisms 
$\overline{\gamma}_\h$ of $\fA(M,L,\zeta)$, but may still have nontrivial anomaly.
A group which has locally trivial anomalies can be defined by 
\begin{equation}
\begin{split}
    \HLz^0\doteq\{\h\in\HL\,|\,&\text{for all relatively compact }\Oc\text{ there exists }\h'\in\G(\h,\Oc)\\
    &\ \text{with }\zeta_{\h'}(F)=F\text{ if }\supp F\subset\Oc\} \ .
\end{split}
\end{equation}
Actually, $\HLz^0$ is a subgroup of $\HLz$. Namely, let $\h\in\HLz^0$.  By using Prop.~\ref{prop:h-zeta} we obtain 
\be
\begin{split}
    (\h\zeta)_\g F=(\h'\zeta)_\g (F)&=(\zeta_{(\h')^{-1}})^{-1}\,\zeta_\g\,(\zeta_{(\h')^{-1}})^\g (F)\\
&=(\zeta_{(\h')^{-1}})^{-1}\,\zeta_\g(F)\\
&=\zeta_\g(F)
\end{split}
\ee
for a suitable choice of $\h'$; in the last step we have taken into account that
$\supp\zeta_\g(F)
\subset(\supp F\cup\supp\g)$ by using \eqref{eq:suppZ(F)}.

As a straightforward application of Theorem~\ref{theorem:anomalousnoethertheorem} to what just discussed, we find
\begin{corollary}[Unitary Noether Theorem]\label{cor:unit:noether} For any $\h\in\HLz^0$ and a causally convex and relatively compact region $\Oc$ there exists a unitary $U\in\fA(M,L,\zeta)$ such that
 \begin{equation}\label{eq:unitarynoether}
    \overline{\gamma}_{\h}(S_{(M,L,\zeta)}(F))=\mathrm{Ad}(U)\bigl(S_{(M,L,\zeta)}(F)\bigr)
\end{equation}
for all $F$ with $\supp F\subset\Oc$.
\end{corollary}
This entails that
symmetries $\h\in\HLz^0$ are locally implemented by unitaries, and the particular case of the unitary MWI 
$$
S_{(M,L,\zeta)}(\delta_{(\h')^{-1}}L)=1\ ,\ \h'\in\G(\h,\Oc)
$$ 
is a unitary version of a conservation law. This constitutes a \emph{unitary version of the Noether theorem}. We see, however, that, dependent on the anomaly $\zeta$, the Noether theorem applies only to a subgroup of unbroken symmetries. 

\begin{example}
To give a concrete example for $\delta_{(\h')^{-1}}L$ in the case that $\zeta_{\h'}=\mathrm{id}$, we look at scalar QED 
(for details see \cite[Sect.~4.2]{DPR21}): for $L=L_0$ the free Lagrangian (including the mass term for the complex 
scalar field $\phi$) and for the affine field redefinition $\iota_{\h'}$ given by
$$
\phi(x)\mapsto\phi(x)\,e^{i\alpha(x)}\ ,\quad \phi^*(x)\mapsto\phi^*(x)\,e^{-i\alpha(x)}\ ,\quad A^\mu(x)\mapsto A^\mu(x),
$$
where $\alpha\in\Dc(M,\RR)$, one obtains
$$
\delta_{(\h')^{-1}}L_0=(\partial j)(\alpha)+(\phi^*\phi)\bigl((\partial\alpha)^2\bigr)\ ,
$$
where $j$ is the electromagnetic current of the free theory ({\it i.e.}, the Noether current pertaining to the invariance of $L_0$ 
under the global $U(1)$-transformation 
$\phi(x)\mapsto\phi(x)\,e^{ia}$ with $a\in\RR$).
\end{example}

\section{Renormalization group flow}\label{sec:RGF}
We have seen in the previous section that the presence of anomalies in form of a cocycle $\zeta$ might induce a nontrivial action of the symmetry group $\HL$ of the Lagrangian on the theory by the action of $\HL$ on $\zeta$. We also observed that, locally, the action of $\h\in\HL$ can be understood as an action of the renormalization group on the theory (Theorem \ref{theorem:anomalousnoethertheorem}) where the renormalization group element is given by $\zeta_{\h'}$, for a local approximation $\h'\in\G(\h,\Oc)$.
We now want to see whether also a global interpretation is possible.
For this purpose we ask in which sense our cocycles can be extended to not necessarily compactly supported symmetry transformations $\h\in\HL$.
A direct extension of $\zeta$, however, might not map compactly supported functionals to compactly supported functionals. 

We therefore define for $F$ with $\supp F\subset\Oc$, $\h'\in\G(\h,\Oc)$ 
\begin{equation}\label{eq:theta-h}
 \theta_{\h}(F)\doteq\zeta_{\h'}(F)-\zeta_{\h'}(0)\ . 
 \end{equation}
$\theta$ is well defined in view of the following proposition:
\begin{proposition}\label{prop:indip1}
$\zeta_{\h'}(F)-\zeta_{\h'}(0)$ does not depend on the choice of $\h'$.
\end{proposition}  
\begin{proof}
Let $\h''\in\G(\h,\Oc)$. Then $\h''=\jmath\h'$ with $\jmath\in\G_c(M)$ and $\supp \jmath\cap\tildeh(\Oc)=\0$. 
The cocycle relation yields
 \begin{equation}
   \zeta_{\h''}= \zeta_{\h'}\zeta_{\jmath}^{\h'}\ .
\end{equation}
We have
\begin{equation}
    \begin{split}
     \zeta_{\jmath}^{\h'}(F)&=(\h')^{-1}_L\zeta_{\jmath}(\h'_{\ast}F+\delta_{\h'}L)\\
                        &=(\h')^{-1}_{\ast}\zeta_{\jmath}\delta_{\h'}L+F+\delta_{(\h')^{-1}}L\\
                        &=(\h')^{-1}_{\ast}(\zeta_{\jmath}\delta_{\h'}L-\delta_{\h'}L)+F\ .
    \end{split}
\end{equation}
Since $\supp(\zeta_{\jmath}\delta_{\h'}L-\delta_{\h'}L)\subset\supp\jmath$ (by \eqref{eq:supp(Z(F)-F)})
we get from Locality of $\zeta_{\h'}$ the relation
\begin{equation}
    \zeta_{\h''}(F)=\zeta_{\h'}(\h')^{-1}_{\ast}(\zeta_{\jmath}\delta_{\h'}L-\delta_{\h'}L)-\zeta_{\h'}(0)+\zeta_{\h'}(F)\ .
\end{equation}
But the first term on the right hand side is equal to $\zeta_{\h''}(0)$. This yields the claim.
\end{proof}






In the next step we show that the family $\zeta_{\h'}(0),\,\h'\in\G(\h,\Oc),\Oc\subset M\ $ relatively compact, 
defines a generalized field $\delta_{\zeta,\h}L$. Namely, given $f\in\Dc(M)$, we choose $\Oc$ such that $\supp f\subset\Oc$ 
and set
\begin{equation}\label{eq:d-zeta-h-L}
   \delta_{\zeta,\h}L(f)[\phi]\doteq\zeta_{\h'}(0)[f\phi]-\zeta_{\h'}(0)[0]\ ,\  
   \h'\in\G(\h,\Oc)\ . 
\end{equation}

Analogously to \eqref{eq:AF} the defining properties of a generalized field (Def.~\ref{eq:def-F-A(f)}) are satisfied.
Moreover, we prove that it is also well defined, namely
\begin{proposition}\label{prop:indip2}
$\zeta_{\h'}(0)[f\phi]-\zeta_{\h'}(0)[0]$ does not depend on the choice of $\h'$.
\end{proposition}
\begin{proof}
Choosing $\h''$ as in the previous proposition, we have
\begin{equation}
    \zeta_{\h''}(0)=\bigl(\zeta_{\h'}(\h')^{-1}_{\ast}(\zeta_{\jmath}\delta_{\h'}L-\delta_{\h'}L)-\zeta_{\h'}(0)\bigr)+\zeta_{\h'}(0)\ .
\end{equation}
The support of the first term on the right hand side is equal to 
$\supp(\h')^{-1}_{\ast}(\zeta_{\jmath}\delta_{\h'}L-\delta_{\h'}L)\subset(\mathpzc{\tilde{h'}})^{-1}(\supp\jmath)$, 
hence its evaluation on configurations $f\phi$ with $\supp f\subset\Oc$ is independent of $\phi$,  and we arrive at
\begin{equation}
   \zeta_{\h''}(0)[f\phi]-\zeta_{\h''}(0)[0]=\zeta_{\h'}(0)[f\phi]-\zeta_{\h'}(0)[0]\ . 
\end{equation}
\end{proof}


We now can construct the flow of theories under the action of the symmetry group in terms of the renomalization group. By the Anomalous Noether Theorem, Theorem~\ref{theorem:anomalousnoethertheorem}, we saw that, up to an inner automorphism, the action of $\overline{\gamma}_{\h}$ can locally be replaced by an isomorphism induced by a renormalization group transformation $\overline{\beta}^{\mathrm{ret}}_Z$.

We recall that a renormalization group element $Z\in\Rc(M,L)$ induces an isomorphism $\beta^{\mathrm{ret}}_Z$ and, from Remark \ref{rem:beta}
(see also \eqref{eq:eps-Z}), that $\beta^{\mathrm{ret}}_Z$ can be interpreted as the composition of two actions: one
adding an interaction $Z(0)$ and the other transforming the 
local functionals by $F\mapsto Z(F)-Z(0)$, \ie
\be\label{eq:betaZ-alpha}
\beta^{\mathrm{ret}}_Z(S_{(M,L)}(F))=\alpha_{A_{Z(0)},+}\bigl(S_{(M,L+A_{Z(0)})}(Z(F)-Z(0))\bigr)\ .
\ee
Passing to the quotient algebras depending on a cocycle $\zeta$, notice that for both previous actions (extended to the quotients, with $Z\doteq\zeta_{\h'}+c$, where $c\doteq -\zeta_{\h'}(0)[0]$  is a constant%
\footnote{Here we use that for any $Z\in\Rc(M,L)$ and any $c\in\RR$ it holds that also $Z+c\in\Rc(M,L)$ and that 
$\beta^{\mathrm{ret}}_Z=\beta^{\mathrm{ret}}_{(Z+c)}$. The latter relies on causal factorization: 
$S_{(M,L)}((Z+c)(F))=S_{(M,L)}(Z(F))\,e^{ic}$.})
we give meaningful expressions,  which depend only on $\h$, but neither on the choice of $\h'\in\G(\h,\Oc)$ nor of $\Oc$ (by Propositions \ref{prop:indip1} and \ref{prop:indip2}). 
Hence we \emph{interpret the $\h$-transformed theory as a theory with an additional interaction $A_{Z(0)}=\delta_{\h,\zeta}L$ and a 
field transformation $F\mapsto Z(F)-Z(0)=\theta_{\h}(F)$}. 
This corresponds nicely to the standard description of anomalies 
(as \eg the scaling anomaly) by running coupling constants (\ie addition of terms to the Lagrangian) and renormalizations  of composite fields (\ie field transformation).
We formulate our findings in the following theorem:
\begin{theorem}
\label{thm:RG-flow}
Symmetries $\h\in\HL$ of a Lagrangian $L$ induce, in the presence of a nontrivial anomaly $\zeta$, a flow of the associated quantum field theory $\fA_{(M,L,\zeta)}$ which can be described in two equivalent ways: either as an action of $\h^{-1}$ on the anomaly leading to the net $\fA_{(M,L,h^{-1}\zeta)}$  $($see \eqref{eq:flow:anomaly}$)$ or as a change $L\mapsto\tilde{L}\doteq L+\delta_{\h\zeta}L$ of the Lagrangian, followed by a transformation $\theta_{\h}$ of the fields. The equivalence follows from the fact that
the map
\begin{equation}\label{eq:flow}
S_{(M,L,\h^{-1}\zeta)}(F)\longmapsto S_{(M,\tilde{L},\tilde\zeta)}(\theta_{\h} F)
\end{equation}
where $\tilde\zeta\doteq\zeta^{\delta_{\h,\zeta}L}$, 
induces a net isomorphism $\fA_{(M,L,\h^{-1}\zeta)}\to\fA_{(M,\tilde L,\tilde \zeta)}$.
\end{theorem}
\begin{remark}
Note that if we restrict ourselves to scaling transformations, then Theorem \ref{thm:RG-flow} corresponds to the Theorem 6.2 of \cite{BDF09} (algebraic Callan-Symanzik equation). In the same paper also the relations to other renormalization group flow equations (e.g. Polchinski) are discussed.
\end{remark}

\begin{remark} 
The map \eqref{eq:flow} corresponds to the net isomorphism $\epsilon_Z$ (Prop.~\ref{prop:betaZ-cocycle}$(iii)$) with
$Z=\zeta_{\h'}-\zeta_{\h'}(0)[0]$, since (by Prop.~\ref{prop:h-zeta}) locally it holds that
$(\h^{-1}\zeta)_\g=(\h'^{-1}\zeta)_\g=Z^{-1}\zeta_\g Z^\g$. However, since $\h$ has non-compact support, it does not suffice to refer to this result,
instead we give a direct proof.
\end{remark}
\begin{proof}
Since $\theta_{\h}$ preserves the support due to condition $(ii)$ for elements of the renormalization group, the map preserves the local subalgebras. We check whether it also maps the axioms for $S\doteq S_{(M,L,\h^{-1}\zeta)}$ into the corresponding axioms for $\tilde S\doteq S_{(M,\tilde L,\tilde\zeta)}$ and vice versa.

We use the fact that $\theta_{\h}$ interpolates between the respective Lagrangians, namely we have
\begin{equation}\label{eq:theta1}
\begin{split}
    \theta_{\h}(F^{\psi}+\delta L(\psi))&=\zeta_{\h'}(F^{\psi}+\delta L(\psi))-\zeta_{\h'}(0)\\
    &=\zeta_{\h'}(F)^{\psi}+\delta L(\psi)-\zeta_{\h'}(0)\\
    &=\theta_{\h}(F)^\psi+\zeta_{\h'}(0)^{\psi}+\delta L(\psi)-\zeta_{\h'}(0)\\
    &=\theta_{\h}(F)^{\psi}+\delta \tilde L(\psi)
\end{split}
\end{equation}
as well as
\begin{equation}\label{eq:theta2}
\begin{split}
    \theta_{\h}(\h^{-1}\zeta)_{\g}(F)
    &=\zeta_{\h'}(\h^{-1}\zeta)_g(F)-\zeta_{\h'}(0)\\
    &=\zeta_{\g}\zeta_{\h'}^{\g}(F)-\zeta_{\h'}(0)\\
    &=\tilde\zeta_{\g}\bigl(\zeta_{\h'}^{\g}(F)-\zeta_{\h'}(0)\bigr)\\
    &=\tilde\zeta_{\g}\bigl(\g^{-1}_{\ast}\theta_{\h}\g_L(F)+\g^{-1}_{\ast}\zeta_{\h'}(0)+\delta_{\g^{-1}}L-\zeta_{\h'}(0)\bigr)\\
    &=\tilde\zeta_{\g}\theta_{\h}^{\g}(F)
\end{split}
\end{equation}
where $\theta_{\h}^{\g}\doteq\g_{L'}^{-1}\theta_{\h}\,\g_L$, by using Prop.~\ref{prop:h-zeta} in the second step.

Since $\theta_{\h}$ also satisfies the Locality condition of Definition \ref{def1} $(ii)$, the Causality relation (Axiom 1) 
for $S$ is mapped into 
\begin{equation}
\begin{split}
     \tilde S&\bigl((\theta_\h(G+F)-\theta_\h(F))+\theta_\h(F)+(\theta_\h(H+F)-\theta_\h(F))\bigr)\\
&=\tilde S(\theta_\h(G+F))\tilde S(\theta_\h(F))^{-1}\tilde S(\theta_\h(H+F))
\end{split}
\end{equation}
if $\supp G\cap J_-^{L+A_F}(\supp H)=\emptyset$. 
This is the same axiom for $\tilde S$ due to $\supp(\theta_\h(G+F)-\theta_\h(F))=\supp G$ 
and since $\tilde{L}+A_{\theta_{\h}(F)}$ induces the same causal structure as $L+A_F$. 
The Dynamical relation (Axiom 2) for $S$ is mapped into 
$\tilde S\bigl(\theta_{\h}(F)^{\psi}+\delta \tilde L(\psi)\bigr)=\tilde S(\theta_{\h}(F))$; this follows straightforwardly from \eqref{eq:theta1}.
Finally the unitary AMWI for $S$ is mapped into $\tilde S(\g_{\tilde L}(\theta_{\h}^{\g}F))=\tilde S({\tilde\zeta
}_\g(\theta_{\h}^{\g}F))$, 
as a consequence of \eqref{eq:theta2}. This proves that the map is homomorphic. But since $\theta_{\h}$ is invertible, as a consequence of the invertiblity of $\zeta_{\h'}$ for all $\h'\in\G(\h,\Oc)$, the map is also bijective. 
\end{proof}

\section{Covariance}
Up to now we did not impose any covariance conditions on the choice of the anomaly. In perturbative locally covariant quantum field theories, renormalization is restricted by the requirement that the quantities of interest (fields, time ordered products, \etc) can be understood as natural transformations between certain functors on the category of spacetimes. In our present paper, we enlarged this category. We will now consider several subcategories and natural transformations between functors on them and discuss the resulting restrictions
for the anomaly.

We start with the space of local functionals. $\Floc$ associates to any dynamical spacetime $(M,L)$ the space $\Floc(M,L)$. Arrows $\iota_\rho:(M,L)\to(M',L')$ for diffeomorphisms $\rho:M\to M'$ with $\rho^{\ast}L'=L$ induce maps $\Floc\iota_{\rho}\equiv\rho_{\ast}:\Floc(M,L)\to\Floc(M',L')$ with
\begin{equation}
    (\rho_{\ast})F[\phi]=F[\phi\circ\rho]\ .
\end{equation}
Arrows $\iota_{\g}:(M,L)\mapsto(M',L')$ for symmetry transformations $\g\in\G_c(M)$ with $M'=M$ and $L'=\g_{\ast}L$ induce the maps $\Floc\iota_{\g}\equiv\g_{\ast}:\Floc(M,L)\to\Floc(M',L')$. 
We do not introduce the maps $\Floc\iota_{V,\pm}$ corresponding to interactions. 

The functor of observable algebras $\fA:(M,L)\to\fA(M,L)$ is defined on the whole category. The $S$-matrix is a natural transformation 
\begin{equation}
    S:\Floc\naturaltra\fA
\end{equation}
for the subcategory where the interaction arrows $\iota_{W,\pm}$ are removed, as illustrated in Figure \ref{fg:Covariance1}.

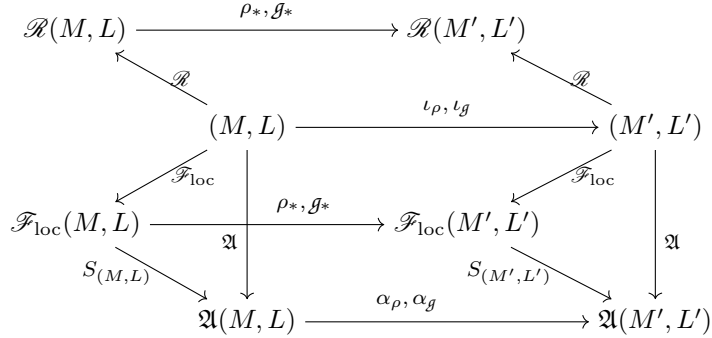
\begin{figure}[ht]
\centering
\begin{tikzpicture}
\matrix(M)[matrix of math nodes, row sep=2em, column sep=1.5 em]{
\Rc(M,L) &&& \Rc(M',L') &\\
& (M,L) &&& (M',L') \\
\Floc(M,L) &&& \Floc(M',L') &\\
& \fA(M,L) &&& \fA(M',L') \\
   };
\begin{scope}[every node/.style={midway,font=\footnotesize}]
\draw[->] (M-2-2) -- node[above] {$\iota_\rho,\iota_\g$} (M-2-5) ;
\draw[->] (M-2-2) -- node[below left] {$\fA$} (M-4-2) ;
\draw[->] (M-2-2) -- node[right] {$\Floc$} (M-3-1);
\draw[->] (M-3-1) -- node[left] {$S_{(M,L)}$} (M-4-2) ;
\draw[->] (M-2-2) -- node[right] {$\Rc$} (M-1-1) ;
\draw[->] (M-2-5) -- node[right] {$\Floc$} (M-3-4);
\draw[->] (M-2-5) -- node[right] {$\Rc$} (M-1-4);
\draw[->] (M-3-4) -- node[left] {$S_{(M',L')}$} (M-4-5) ;
\draw[->] (M-2-5) -- node[below right] {$\fA$} (M-4-5) ;
\draw[->] (M-1-1) -- node[above] {$\rho_\ast,\g_\ast$} (M-1-4) ;
\draw[->] (M-3-1) -- node[above right] {$\rho_\ast,\g_\ast$} (M-3-4) ;
\draw[->] (M-4-2) -- node[above left] {$\alpha_\rho,\alpha_\g$} (M-4-5) ;
\end{scope}
\end{tikzpicture}
\caption{\small Covariance with respect to diffeomorphisms $\rho:M\to M'$ and symmetry transformations $\g\in\G_c(M)$ (for the latter
we assume $M'= M$). The diagram describes the functors $\Floc$, $\Rc$ and $\fA$ and the natural transformation $S$.}
\label{fg:Covariance1}
\end{figure}

The interactions transform under symmetries and embeddings as
\begin{equation}
    (\g_{\ast}V)(f)\doteq\g_{\ast}(V(f\circ\tildeg))
\end{equation}
In order to include the arrows corresponding to interactions we introduce a bundle over the space of past compactly supported interactions with fibers of associated local functionals,
\begin{equation}
    \mathrm{Int}_+^2(M,L)=\bigsqcup_{V\in\mathrm{Int}(M,L),\supp V\text{ past compact}}\{V\}\times\Floc(M,L+V)\ .
\end{equation}
In addition to the arrows for field redefinitions and embeddings which just act for both components as before, we also introduce the action of arrows $\iota_{W,+}$ for interactions $W$ with past compact support
\begin{equation}
    \mathrm{Int}_+^2\iota_{W,+}:\begin{cases}\mathrm{Int}_+^2(M,L+W)\longrightarrow &\mathrm{Int}_+^2(M,L)\\
    (V,F)\qquad\qquad\,\,\,\longmapsto &(V+W,F)\end{cases}\ ,
\end{equation}
where $W\in\mathrm{Int}(M,L)$, $V\in\mathrm{Int}(M,L+W)$ (hence $V+W\in\mathrm{Int}(M,L)$) and $F\in\Floc(M,L+W+V)$.
We then consider the relative $S$-matrices
\begin{equation}
\begin{split}
   S^{\mathrm{rel}}_{(M,L)}(V,F)&\doteq \alpha_{V,+}\bigl(S_{(M,L+V)}(F)\bigr)=S_{(M,L)}(V(f))^{-1}\,S_{(M,L)}(F+V(f))\\
\end{split}
\end{equation}
and find that they define a natural transformation $S^{\mathrm{rel}}:\mathrm{Int}_+^2\to\fA$, \ie
$$
S^{\mathrm{rel}}_{(M,L)}(V+W,F)=\alpha_{W,+}\bigl(S^{\mathrm{rel}}_{(M,L+W)}(V,F)\bigr)\ ,
$$
if only the arrows for retarded interactions are considered, see Figure \ref{fg:Covariance2}. 
An analogous construction can be done for the advanced case.

\begin{figure}[ht]
\centering
\begin{tikzpicture}[scale=0.5]
\matrix(M)[matrix of math nodes, row sep=2em, column sep=1 em]{
\Rc(M,L+W) &&& \Rc(M,L)&\\
& (M,L+W) &&& (M,L) \\
\mathrm{Int}_+^2(M,L+W) &&& \mathrm{Int}_+^2(M,L) &\\
& \fA(M,L+W) &&& \fA(M,L) \\
   };
\begin{scope}[every node/.style={midway,font=\footnotesize}]
\draw[->] (M-2-2) -- node[above] {$\iota_{W,+}$} (M-2-5) ;
\draw[->] (M-2-2) -- node[below left] {$\fA$} (M-4-2) ;
\draw[->] (M-2-2) -- node[left] {$\mathrm{Int}_+^2$} (M-3-1);
\draw[->] (M-3-1) -- node[left] {$S^{\mathrm{rel}}_{(M,L+W)}$} (M-4-2) ;
\draw[->] (M-2-2) -- node[right] {$\Rc$} (M-1-1) ;
\draw[->] (M-2-5) -- node[left] {$\mathrm{Int}_+^2$} (M-3-4);
\draw[->] (M-2-5) -- node[right] {$\Rc$} (M-1-4);
\draw[->] (M-3-4) -- node[left] {$S^{\mathrm{rel}}_{(M,L)}$} (M-4-5) ;
\draw[->] (M-2-5) -- node[below right] {$\fA$} (M-4-5) ;
\draw[->] (M-1-1) -- node[above] {$Z\mapsto Z^{-W}$} (M-1-4) ;
\draw[->] (M-3-1) -- node[above right] {$\mathrm{Int}_+^2\iota_{W,+}$} (M-3-4) ;
\draw[->] (M-4-2) -- node[above left] {$\alpha_{W,+}$} (M-4-5) ;
\end{scope}
\end{tikzpicture}
\caption{\small Covariance with respect to the subtraction of an interaction $W$ with past compact support, given by the morphism $\iota_{W,+}$.}
\label{fg:Covariance2}
\end{figure}

The functor $\Rc:(M,L)\mapsto \Rc(M,L)$ of renormalization groups maps the interaction arrows $\iota_{V,\pm}$ to the maps 
$Z\mapsto Z^{-V}$ (or equivalently: $Z^V\mapsto Z$). For
symmetries we get $(\g_{\ast}Z)(F)\doteq\g_{\ast}(Z(\g_{\ast}^{-1}F)$. For embeddings $\rho:(M,L)\to (M',L')$
there is the difficulty that (for $F\in\Floc(M',L')$) $\rho^{-1}_\ast F$ is only defined if $\supp F\subset\rho(M)$, 
hence we decompose $F=F_0+F_1$ with $\supp F_0\subset\rho(M)$ and $\supp F_1\cap\rho(\supp Z)=\0$ and set
\begin{equation}\label{eq:rho*Z}
    (\rho_{\ast}Z)(F)\doteq\rho_{\ast}(Z(\rho_{\ast}^{-1}F_0))+F_1\ .
\end{equation}
The right hand side is independent of the split $F=F_0+F_1$, since for another split $F=F_0'+F_1'$ the functional $G=F_0'-F_0=F_1-F_1'$ satisfies $\supp G\cap\rho(\supp Z)=\0$, hence 
\begin{equation}
   \rho_{\ast}(Z(\rho_{\ast}^{-1}F_0'))=\rho_{\ast}(Z(\rho_{\ast}^{-1}F_0))+G \ .
\end{equation}

We finally discuss the covariance properties of the cocycles $\zeta$. The symmetry groups $\G_c(M)$ do not depend on the Lagrangian and their elements $\g=(\Phi,\chi)$ do not transform under 
changes of the interaction. Under symmetry transformations $\h\in\G_c(M)$ they transform by conjugation:
$\g\mapsto \h\g\h^{-1}$. Under embeddings we have to use that they are compactly supported. So we set
\begin{equation}
  \rho_{\ast}(\Phi,\chi)\doteq (\rho\Phi,\rho_{\ast}\chi)
\end{equation}
with $\rho\Phi\restriction_{\rho}$ as defined in \eqref{eq:semi} and extended to the identity outside of $\rho(M)$, and with
\begin{equation}
    \rho_{\ast}\chi(x)=\left\{
    \begin{array}{ccc}
    \rho(\chi(\rho^{-1}(x))) &,& x\in\rho(M) \\
       x&,&\text{ else}
    \end{array}\right.
\end{equation} 
The cocycles then transform by
\begin{equation}
    \zeta\mapsto \zeta^{-V}
\end{equation}
under interactions $\iota_{V;\pm}$, by $\zeta\to\g\zeta$ under symmetry tranformations 
$\iota_\g$, (here we use the definition \eqref{eq:sym-acts-cocy} for $\g\in\G_c(M)$),  and by
\begin{equation}
(\rho_{\ast}\zeta)_{\rho_{\ast}\g}=\rho_{\ast}\zeta_{\g}\rho_{\ast}^{-1}
\end{equation}
under embeddings $\iota_\rho$ by using \eqref{eq:rho*Z}. Note that $\rho_{\ast}\zeta$ is defined only for symmetry transformations with support in the image of $\rho$. 

We can now discuss the appropriate naturality conditions on $\zeta$.
The covariance under adding interactions is already taking into account: $\zeta:\G_c\times\Floc\to\Floc$ is a natural transformation
with respect to $\iota_{V,\pm}$, \ie $\zeta^{-V}_\g(F+V(f))=\zeta_g(F)+V(f)$, as illustrated in Figure \ref{fg:Covariance3}.
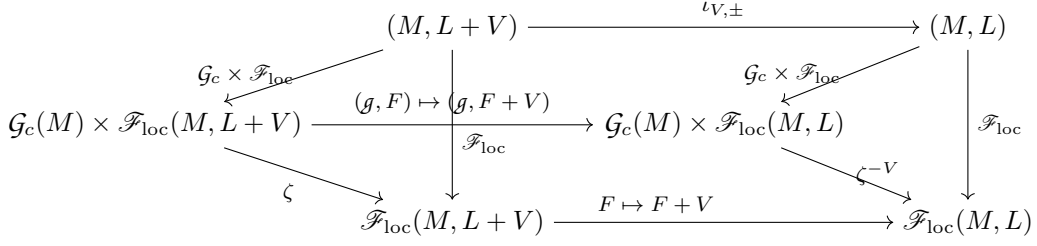
\begin{figure}[ht]
\centering
\begin{tikzpicture}
\matrix(M)[matrix of math nodes, row sep=2em, column sep=1.5 em]{
& (M,L+V) && (M,L) \\
\G_c(M)\times\Floc(M,L+V) && \G_c(M)\times\Floc(M,L) &\\
& \Floc(M,L+V) && \Floc(M,L) \\
   };
\begin{scope}[every node/.style={midway,font=\footnotesize}]
\draw[->] (M-1-2) -- node[above] {$\iota_{V,\pm}$} (M-1-4) ;
\draw[->] (M-1-2) -- node[below right] {$\Floc$} (M-3-2) ;
\draw[->] (M-1-2) -- node[left] {$\G_c\times\Floc$} (M-2-1);
\draw[->] (M-2-1) -- node[below left] {$\zeta$} (M-3-2) ;
\draw[->] (M-1-4) -- node[left] {$\G_c\times\Floc$} (M-2-3);
\draw[->] (M-2-3) -- node[right] {$\zeta^{-V}$} (M-3-4) ;
\draw[->] (M-1-4) -- node[right] {$\Floc$} (M-3-4) ;
\draw[->] (M-2-1) -- node[above] {$(\g,F)\mapsto (\g,F+V)$} (M-2-3) ;
\draw[->] (M-3-2) -- node[above left] {$F\mapsto F+V$} (M-3-4) ;
\end{scope}
\end{tikzpicture}
\caption{\small Naturality of $\zeta$ with respect to the subtraction of the interaction $V$. For shortness we write $V$ for $V(f)$.}
\label{fg:Covariance3}
\end{figure}
As a consequence it is sufficient to define the cocycle for a specific Lagrangian. 

Actually, in perturbation theory one discusses the renormalization within the free theory (including its time ordered product) which determines then also the interacting theory. The covariance under embeddings is the crucial condition for local covariance. Due to the work of Hollands and Wald \cite{HW01} (see also \cite{KhavMor}) this condition is satisfied in perturbation theory, hence we may impose it in our framework:
$\zeta:\G_c\times\Floc\to\Floc$ is a natural transformation
with respect to $\iota_\rho$, \ie $(\rho_\ast\zeta)_{\rho_\ast\g}(\rho_\ast F)=\rho_\ast\bigl(\zeta_g(F)\bigr)$, see Figure
\ref{fg:Covariance4}.
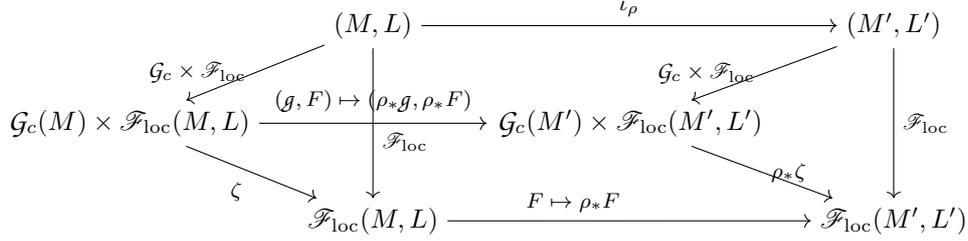
\begin{figure}[ht]
\centering
\begin{tikzpicture}
\matrix(M)[matrix of math nodes, row sep=2em, column sep=1.5 em]{
& (M,L) && (M',L') \\
\G_c(M)\times\Floc(M,L) && \G_c(M')\times\Floc(M',L') &\\
& \Floc(M,L) && \Floc(M',L') \\
   };
\begin{scope}[every node/.style={midway,font=\footnotesize}]
\draw[->] (M-1-2) -- node[above] {$\iota_\rho$} (M-1-4) ;
\draw[->] (M-1-2) -- node[below right] {$\Floc$} (M-3-2) ;
\draw[->] (M-1-2) -- node[left] {$\G_c\times\Floc$} (M-2-1);
\draw[->] (M-2-1) -- node[below left] {$\zeta$} (M-3-2) ;
\draw[->] (M-1-4) -- node[left] {$\G_c\times\Floc$} (M-2-3);
\draw[->] (M-2-3) -- node[right] {$\rho_\ast\zeta$} (M-3-4) ;
\draw[->] (M-1-4) -- node[right] {$\Floc$} (M-3-4) ;
\draw[->] (M-2-1) -- node[above] {$(\g,F)\mapsto (\rho_\ast\g,\rho_\ast F)$} (M-2-3) ;
\draw[->] (M-3-2) -- node[above left] {$F\mapsto \rho_\ast F$} (M-3-4) ;
\end{scope}
\end{tikzpicture}
\caption{\small Naturality of $\zeta$ with respect to diffeomorphisms $\rho:M\to M'$.}
\label{fg:Covariance4}
\end{figure}

Analogously to the latter result,
we can also impose the naturality condition for the arrows coming from symmetry transformations $\g\in\G_c(M)$, 
since due to the assumption on global hyperbolicity no compactly supported symmetry transformation leaves the Lagrangian invariant.

There remains, however, the relation between the arrow corresponding to the application of a symmetry $\h\in\G_c(M)$ 
and the arrow corresponding to the induced change of the Lagrangian: $L\to\h_\ast L=L+\delta_\h L$, 
details are given in Figure \ref{fg:Covariance5}.
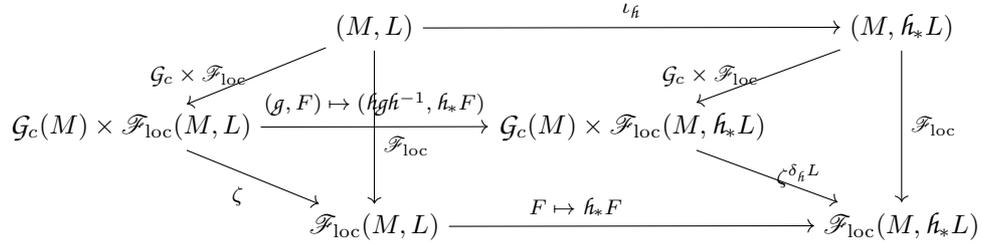
\begin{figure}[ht]
\centering
\begin{tikzpicture}
\matrix(M)[matrix of math nodes, row sep=2em, column sep=1.5 em]{
& (M,L) && (M,\h_\ast L) \\
\G_c(M)\times\Floc(M,L) && \G_c(M)\times\Floc(M,\h_\ast L) &\\
& \Floc(M,L) && \Floc(M,\h_\ast L) \\
   };
\begin{scope}[every node/.style={midway,font=\footnotesize}]
\draw[->] (M-1-2) -- node[above] {$\iota_\h$} (M-1-4) ;
\draw[->] (M-1-2) -- node[below right] {$\Floc$} (M-3-2) ;
\draw[->] (M-1-2) -- node[left] {$\G_c\times\Floc$} (M-2-1);
\draw[->] (M-2-1) -- node[below left] {$\zeta$} (M-3-2) ;
\draw[->] (M-1-4) -- node[left] {$\G_c\times\Floc$} (M-2-3);
\draw[->] (M-2-3) -- node[right] {$\zeta^{\delta_\h L}$} (M-3-4) ;
\draw[->] (M-1-4) -- node[right] {$\Floc$} (M-3-4) ;
\draw[->] (M-2-1) -- node[above] {$(\g,F)\mapsto (\h\g\h^{-1},\h_\ast F)$} (M-2-3) ;
\draw[->] (M-3-2) -- node[above left] {$F\mapsto \h_\ast F$} (M-3-4) ;
\end{scope}
\end{tikzpicture}
\caption{\small Behaviour of $\zeta$ under the application of a symmetry $\h\in\G_c(M)$ and the induced change of the Lagrangian.}
\label{fg:Covariance5}
\end{figure}

\noindent Full naturality of $\zeta$ would yield $\zeta^{\delta_{\h}L}_{\h\g\h^{-1}}(\h_\ast F){=}\h_\ast\zeta_g(F)$, that is,
\begin{equation}\label{eq:cocycle-naturality}
    \h\zeta{=}\zeta^{\delta_{\h}L} \ .
\end{equation}
But for such a $\zeta$ the renormalization group flow would be trivial:
\begin{proposition}
Let $\zeta\in\mathfrak{Z}(M,L)$.
\begin{itemize}
    \item[$(i)$] The cocycles $\h\zeta$ and $\zeta^{\delta_{\h}L}$ are equivalent for every $\h\in\G_c(M)$,
    \begin{equation}\label{eq:hzeta-zetadL}
    (\h\zeta)_{\g}=Z^{-1}\zeta_\g^{\delta_{\h}L}Z^{\g}\ ,\quad \g\in\G_c(M)\ ,
\end{equation}
with $Z\doteq\zeta_{\h^{-1}}^{\delta_{\h}L}$.
    \item[$(ii)$] If  $\h'\zeta{=}\zeta^{\delta_{\h'}L}$ for all $\h'\in\G_c(M)$,
then $\h\zeta=\zeta$ for all $\h\in\HL$.
\end{itemize}
\end{proposition}
\begin{proof}
\begin{itemize}
    \item[$(i)$] 
We use the calculations in Section \ref{sec:symL} leading to Proposition \ref{prop:h-zeta}, but now for a generic element $\h\in\G_c$. The transformed cocycle $\h\zeta$ takes values in $\Rc(M,\h_{\ast}L)$. The second equation in the proposition then assumes the form
\eqref{eq:hzeta-zetadL}.
This follows by the following chain of identities:
\begin{equation}\label{eq:equivalence}
    \begin{split}
        (\h\zeta)_{\g}(F)\quad&\overset{\mathclap{\eqref{eq:sym-acts-cocy}}}{=}\quad\h_{\ast}\zeta_{\h^{-1}\g\h}\h_{\ast}^{-1}(F)\\
        &=\quad\h_{\h_{\ast}L}(\h_{\h_{\ast}L}^{-1}\h_{\ast})\zeta_{\h^{-1}\g\h}(\h_{\ast}^{-1}\h_{\h_{\ast}L})\h_{\h_{\ast}L}^{-1}(F)\\
        &\overset{\mathclap{\eqref{eq:g_L}}}{=}\quad
        \h_{\h_{\ast}L}\bigl(\delta_{\h^{-1}}\h_{\ast}L+\zeta_{\h^{-1}\g\h}(\delta_{\h}L+\h_{\h_{\ast}L}^{-1}(F))\bigr)\\
        &\overset{\mathclap{\eqref{eq:renint}}}{=}\quad\h_{\h_{\ast}L}\zeta^{\delta_{\h}L}_{\h^{-1}\g\h}\h_{\h_{\ast}L}^{-1}(F)\\
        &\overset{\mathclap{\mathrm{(cr)}}}{=}\quad\h_{\h_{\ast}L}(\zeta^{\delta_{\h}L}_{\h}\h_{\h_{\ast}L}^{-1}
        \zeta^{\delta_{\h}L}_{\g}\g_{\h_{\ast}L}^{-1}\zeta^{\delta_{\h}L}_{\h^{-1}}\g_{\h_{\ast}L})(F)\\
        &\overset{\mathclap{\mathrm{(cr)}}}{=}\quad
        (\zeta_{\h^{-1}}^{\delta_{\h}L})^{-1}\zeta_{\g}^{\delta_{\h}L}(\zeta_{\h^{-1}}^{\delta_{\h}L})^{\g}(F)\ ,
    \end{split}
\end{equation}
where ``(cr)'' stands for the cocycle relation \eqref{eq:cocycle} in $\Rc(M,\h_{\ast}L)$.
\item[$(ii)$] 

For $\kk,\g\in\G_c(M)$ the assumption \eqref{eq:equivalence} together with \eqref{eq:cocycle-naturality} implies
\be
\zeta_\g^{\delta_\kk L}=(\kk\zeta)_\g=(\zeta_{\kk^{-1}}^{\delta_{\kk}L})^{-1}\zeta_{\g}^{\delta_{\kk}L}
(\zeta_{\kk^{-1}}^{\delta_{\kk}L})^{\g}\,\overset{\mathclap{\mathrm{(cr)}}}{=}\,(\zeta_{\kk^{-1}}^{\delta_{\kk}L})^{-1}
\zeta_{\kk^{-1}\g}^{\delta_{\kk}L}\ .
\ee
Using that $Z_1^VZ_2^V=(Z_1Z_2)^V$ (for $Z_1,Z_2\in\Rc(M,L)$, $V\in\mathrm{Int}(M,L)$) we see that $\G_c(M)\ni\g\mapsto\zeta_\g$
is a representation
\begin{equation}\label{eq:cocycle-representation}
    \begin{split}
    \zeta_{\kk^{-1}}\zeta_{\g}=\zeta_{\kk^{-1}\g}=\zeta_{\g}(\zeta_{\kk^{-1}})^{\g}\ .
    \end{split}
\end{equation}
Let now $\h\in\HL$, and let $\g\in\G_c(M)$, $F\in\Floc(M,L)$. Choose $\h'\in\G(\h,\Oc)$ for some relatively compact, causally convex region $\Oc\subset M$ with $\supp\g,\,\supp F\subset\tilde\h(\Oc)$. Then according to Proposition \ref{prop:h-zeta} we have 
\begin{equation}
    (\h\zeta)_{\g}(F)=\zeta_{(\h')^{-1}}^{-1}\zeta_{\g}\zeta_{(\h')^{-1}}^{\g}(F)\ .
\end{equation}
The claim now follows from \eqref{eq:cocycle-representation}. 
\end{itemize}
\end{proof}
The assumption of full naturality would exclude important examples of anomalies in perturbation theory. Instead we require that the equivalence class of $\zeta$ is natural for the categorical structure described above.

\begin{remark}
The algebra $\fA(M,L,\zeta)$ is well-defined for any choice of the cocycle  $\zeta$ satisfying 
the defining conditions given in formula \eqref{eq:cocycle} and directly before it. In particular, we may choose
$\zeta_\g=\mathrm{id}_{\Floc(M,L)}$ for all $\g\in\G_c(M)$.
However this would not be an optimal choice for a model $(M,L)$ with a nontrivial anomaly (as one knows e.g.~from perturbative computations) since this might prevent the existence of
physical states (as e.g.~the vacuum in Minkowski space) for $\fA(M,L,\zeta)$.
\end{remark}

\section{The unitary AMWI in perturbation theory}\label{sec:pQFT}

\subsection{Perturbation theory}\label{subsec:Pt}
The unitary anomalous MWI (see \eqref{eq:uni-anom-MWI}) is an additional axiom. In this subsection,
we show that an \emph{essentially} equivalent condition holds in perturbation theory. In perturbation theory (in terms of formal power series in the coupling constant $\lambda$), for a free (\ie quadratic) Lagrangian $L$,  S-matrices are defined as time-ordered exponentials  
\begin{equation}\label{eq:simple-S}
S(\lambda F)\equiv e^{i\lambda F}_{T}\equiv 1+\sum_{n=1}^\infty \frac{(i\lambda)^n}{n!}F\cdot_T\ldots \cdot_T F\,,\quad F\in\Floc(M)\,.
\end{equation}

Here $\cdot_T$, the renormalized time-ordered product, is a binary, commutative and associative product \cite{FredenhagenR13} on a subspace (which contains the local functionals) of the space of all functionals on the configuration space (not restricted to solutions of the field equation, \ie in the ``off-shell formalism''\footnote{This differs from the formalism often used in physics, where the time ordered products are identified with Fock space operators, corresponding to  functionals which vanish on solutions of
the free field equation (``on-shell formalism''). Since the functionals vanishing on solutions do not form an ideal for the off-shell time-ordered product, the on-shell time-ordered product is less well behaved.}). $S(\lambda F)$ is
realized as formal power series of functionals. 
$\Floc(M)$ is the vector space of polynomial local functionals on $M$.  
The formal parameter $\lambda$ serves only as a bookkeeping device, so for the simplicity of notation we are going to omit it and write formal power series as infinite sums.

The perturbative S-matrix $S$ is supposed to satisfy the causal factorization condition \eqref{Caus} (see Remark \ref{rm:Caus}). 
It can be constructed by the Epstein-Glaser method \cite{EG73}, generalized to curved space times \cite{BF00,HW01,HW02}. It is not 
uniquely determined by the condition \eqref{Caus}, but according to Stora's {\it Main Theorem of Renormalization} \cite{PopineauS16,Pinter01,DF04} any other solution $\hat{S}$ is of the form 
\be\label{eq:MainT}
\hat{S}=S\circ Z
\ee
where $Z$ is a formal power series 
\be\label{eq:Z-perturbative}
Z(F)=\sum_{n=0}^{\infty}\frac{1}{n!}Z_n(F^{\ox n})
\ee
of linear symmetric maps $Z_n:\Floc(M)^{\ox n}\to\Floc(M)$, and the composition has to be understood in the sense of insertions of a formal power series into another one. The formal power series $Z$ are invertible (in the sense of formal power series) and generate a group, the renormalization group in the sense of Stückelberg and Petermann \cite{StuPet53}, see also \cite{DF04}.  Applying \eqref{eq:MainT} to the local functional $F=0$, we immediately get
$Z(0)=Z_0=0$. However, for a better agreement with Def.~\ref{def1} of $\Rc(M,L)$ and with Prop.~\ref{prop:automorphism-by-Z}, 
we may add a constant functional $c\in\RR$: in terms of $\tilde Z\doteq Z+c$ (hence $\tilde Z(0)=c$) the Main Theorem 
formula \eqref{eq:MainT} can equivalently be written as
\begin{equation}\label{eq:MainT1}
    \hat{S}(F)=\beta_{\tilde Z}(S(F))\ ,\quad F\in\Floc(M)\ ,
\end{equation}
with $\beta_{\tilde Z}\doteq\beta_{\tilde Z}^{\mathrm{ret}}$ or $\beta_{\tilde Z}\doteq\beta_{\tilde Z}^{\mathrm{adv}}$ (both choices coincide) 
and $\beta_{\tilde Z}^{\mathrm{ret/adv}}$ defined in Prop.~\ref{prop:automorphism-by-Z}.

The ambiguity on the choice of $S$ can be restricted by renormalization conditions which then also restrict the renormalization group to a subgroup. For our purposes we need unitarity, field equation (FE) (\ie the Schwinger-Dyson equation \eqref{SD}, see \cite[Sect.~7]{BDFR21}), Action Ward Identity and field independence and denote the corresponding subgroup of the renormalization group by $\Rc_0$. (For details see \cite{DF04} or \cite[Chap.~3.1]{Due19}.) We do not impose covariance conditions since we want to study the behavior of $S$ under symmetry transformations $\g\in\G_c(M)$ which would destroy these conditions. Since they are compactly supported we are interested in the subgroup $\Rc_c$ of $\Rc_0$
of renormalization group elements with compact support (in the sense of \eqref{eq:suppZ}). The explicit definition of $\Rc_c$
is given in Appendix \ref{app:Delta-X}.

Note that \eqref{eq:simple-S} implies $\frac{d}{d\lambda}\vert_{\lambda=0}S(\lambda F)=iF$ and hence $Z_1=\mathrm{id}$.
For a generally covariant formalism, however, this choice is too restrictive \cite{HW02}. We take this generalization into account by admitting nontrivial, but still invertible $Z_1$. For convenience, we continue to use the formulation for the S-matrix as in formula \eqref{eq:simple-S} and obtain the more general S-matrices by composition with the renormalization group map $Z$ as in equation \eqref{eq:MainT}.

We use as an input the anomalous MWI 
from Brennecke and one of us \cite[Thm.~7]{Brennecke08}, see also \cite[Chap.~4.3]{Due19}. This identity is equivalent to a renormalized version of the Quantum Master Equation in the BV formalism (see \cite{FredenhagenR13}) and can be understood as an infinitesimal version of the unitary AMWI introduced in our paper,
the latter turns out from Theorem \ref{th:AMWI-uAMWI} below.

$\G_c(M)$ and $\Rc_c$ can be equipped with appropriate topologies and made into infinite-dimensional Lie groups modelled on locally convex topological vector spaces (for more details, see e.g. \cite{Neeb06,Neeb,Michor} and [Rej16] for review in the context of perturbative AQFT).  For $\G_c(M)$ this is relatively easy, since $\mathrm{Diff}_c(M)$ is a standard example of an infinite dimensional Lie group, while $\mathcal{C}^{\infty}_c(M,\mathrm{Aff}(\RR^n))$ is just the gauge group of point-wise affine transformations (smooth functions with values in the finite-dimensional Lie group of affine transformations on $\RR^n$), hence also standard. The action of diffeomorphisms on the group of affine transformations is smooth, so the Lie group structure on the semidirect product follows.
 As for $\Rc_c$, note that this is a subspace of the space of formal power series with values in multilinear maps from $\Floc(M)$ to $\Floc(M)$, which in itself is a topological vector space, so $\Rc_c$ can be equipped with the induced topology.

Note that we avoid the more subtle aspects of infinite dimensional differential geometry, since here we work with very explicit formulas and the precise choice of a setting for infinite dimensional calculus is not that relevant.

Let $\LieGc$ denote the Lie algebra of $\G_c(M)$ (Def.~\ref{def:G}) and $\LieRc$ the Lie algebra of $\Rc_c$.
The anomalous MWI states that there exists a linear map $\Delta:\LieGc\to\LieRc$ such that
\be\label{eq:anomMWI}
e_T^{iF}\cdot_T\bigl(\partial_X F+\partial_X L-\Delta X(F)\bigr)=\int d^4x\ \bigl(e_T^{iF}\cdot_T X\phi(x)\bigr)\frac{\delta L}{\delta \phi(x)}
\ee
for all $F\in\Floc(M)$, $X\in\LieGc$ and with the action
\be\label{eq:dXG}
\partial_X G[\phi]\doteq \int d^4x\ \frac{\delta G}{\delta \phi(x)}X\phi(x)
\ee
for $G\in\Floc(M)$, and with $\partial_X L\doteq \partial_X L(f)$ with $f\equiv 1$ on $\supp X$. This defines a map from $\LieGc$ to $\Gamma(T\Ec(M,\RR^n))$, the space of vector fields on $\Ec(M,\RR^n)$, by $X\mapsto \partial_X$.
The relation to the renormalized BV Laplacian $\triangle$ of \cite{FredenhagenR13} is given by $\triangle_F(\partial_X)\equiv \Delta X(F)$.

\begin{remark}
The statement that $\Delta X\in\LieRc$ is not given in the original formulation of the AMWI (in \cite[Thm.~7]{Brennecke08}
or \cite[Chap.~4.3]{Due19}), but apart from this the geometrical formulation of the AMWI \eqref{eq:anomMWI} is \emph{equivalent} to the original one,
as one sees from the following translation (for illustration see \cite[Example 4.2.1]{Due19} and \cite{DPR21}). In the mentioned references
the AMWI is formulated in terms of the local functional $A\doteq\int d^4x\,\,h(x)Q(x)\,\frac{\delta L(f)}{\delta\phi(x)}$ (where
$h\in\Dc(M)$, $Q$ is a polynomial in $\phi$ and its partial derivatives and $f\vert_{\supp h}=1$), the pertinent 
derivation\footnote{In \cite{Due19} the derivation $\delta_A$ of \cite{Brennecke08} is denoted by $\delta_{hQ}$.}
$\delta_A\equiv\delta_{hQ}\doteq\int d^4x\,\,h(x)Q(x)\,\frac{\delta}{\delta\phi(x)}$ and the anomaly map 
$F\mapsto\Delta_A(e_\ox^F)\equiv\Delta(e_\ox^F;hQ)$. Now, $A$ and $\delta_A$ can be obtained as follows: for a smooth curve
$(\g^\lambda)$ in $\G_c(M)$ with $\g^0=\e$ and with tangent vector $X\in\LieGc$ at $\lambda=0$ we set 
\be
h(x)Q(x)\doteq X\phi(x)=\frac{d}{d\lambda}\Big\vert_{\lambda=0}(\g^\lambda\phi)(x)
\ee
and consequently
\be\label{eq:dA-dX}
\delta_A G\equiv\delta_{hQ}G=\partial_X G
\ee
for $G\in\Floc(M)$, in particular $A=\partial_X L$. In this paper we write $-\Delta X(F)$ for the anomaly map $\Delta(e_\ox^F;hQ)$.

The statement that $\Delta X\in\LieRc$ for $X\in\LieGc$ is proved in Appendix \ref{app:Delta-X}. 
It expresses that the St\"uckelberg-Petermann group has an additional purpose (besides describing the non-uniqueness of the perturbative
$S$-matrix by the Main Theorem \eqref{eq:MainT}): it also characterizes the anomalies of the MWI, restricted to fields $Q$ which are of first order in $\phi$. This is analogous to the
two different purposes of the (non-perturbative) renormalization group $\Rc(M,L)$ (Def.~\ref{def1}) worked out in this paper: 
the automorphisms $\beta_Z^{\mathrm{ret}/\mathrm{adv}}$ of $\fA(M,L)$
(with $Z\in\mathcal{R}(M,L)$, see Prop.~\ref{prop:automorphism-by-Z}) and the characterization of the anomaly map 
$\g\to\zeta_\g$ (belonging to the unitary AMWI \eqref{eq:uni-anom-MWI}) by $\zeta_\g\in\Rc(M,L)$.
\end{remark}

For perturbative, scalar QED, the  equivalence of the on-shell MWI and the unitary MWI\footnote{The unitary MWI is also an on-shell statement.} 
expressing global $U(1)$-symmetry is proved in \cite[Thm.~4.1]{DPR21}. We now generalize this result 
in different ways, we prove it for the \emph{anomalous} MWI and for \emph{arbitrary models} and \emph{arbitrary symmetry transformations} $\g\in\G_c(M)$. 

\begin{remark}
In the following Theorem we use the simple form of the S-matrix as in \eqref{eq:simple-S}. Using instead a general S-matrix by composition with a renormalization group map $Z$ would require in the result only a change to an equivalent cocycle, as explained in 
Definition~\ref{def:Cocycle Equivalence} and Proposition~\ref{prop:betaZ-cocycle}. (The proofs of Lemma~\ref{lem:Zg-in-R} and
Proposition~\ref{prop:betaZ-cocycle} apply also to the perturbative framework considered here, because the definition of $\Rc_c$ \ref{def:Rc}
is analogous to the Definition of $\Rc(M,L)$ \ref{def1} reduced to the special case $V=0$.)
\end{remark}

\begin{theorem}\label{th:AMWI-uAMWI}
In formal perturbation theory (\ie the equations hold in the sense of formal power series) the on-shell AMWI and the unitary AMWI are equivalent in the following sense:
\begin{itemize}
    \item[$(i)$] The on-shell AMWI ({\it i.e.}, \eqref{eq:anomMWI} $\mathrm{mod}\ \frac{\delta L}{\delta\phi}$) implies the unitary AMWI,
\be\label{eq:uAMWI}
S\circ \g_L(F)=S\circ\zeta_\g(F)\ \mathrm{mod}\ \frac{\delta L}{\delta\phi}\quad\text{for all}\quad\g\in\G_c(M)\ ,\,\,F\in\Floc(M)\ ,
\ee
with a cocycle $\zeta$ taking values in  $\Rc_c$ and with $\supp\zeta_\g\subset\supp\g,\g\in\G_c(M)$.
\item[$(ii)$] The unitary AMWI \eqref{eq:uAMWI} implies the on-shell AMWI \eqref{eq:anomMWI} where $X$ is the tangent vector at $\lambda=0$ of a smooth curve $(\g^\lambda)$ in $\G_c(M)$ with $\g^0=\e$ and the anomaly map $\Delta X$ is given by $\Delta X=\frac{d}{d\lambda}\big\vert_{\lambda=0}\zeta_{\g^\lambda}$. It holds  that $\supp\Delta X\subset\supp X$.
\end{itemize}
\end{theorem}
\begin{proof} $(i)$
Let $\g\in\G_c(M)$. We choose a smooth curve $\lambda\mapsto\g^{\lambda}\in\G_c(M)$ with $\g^0=\e$ and $\g^1=\g$ and let 
$X^{\lambda}\in\LieGc$ be defined by $\frac{d}{d\lambda}\g^{\lambda}=X^{\lambda}\g^{\lambda}$. In the spirit of
Remark \ref{rm:uAMWI-alternative}
we search for a smooth curve $\lambda\mapsto\zeta_{\g^\lambda}^{-1}\in\Rc_c$ with $\zeta_\e^{-1}=\mathrm{id}$ and 
\be\label{eq:Z(1-lambda)-1}
\frac{d}{d\lambda}S\bigl(\g^{\lambda}_L\zeta_{\g^\lambda}^{-1}(F)\bigr)=0\ \mathrm{mod}\ \frac{\delta L}{\delta\phi}\ .
\ee
Inserting then $\lambda =0$ and $\lambda =1$ into $S\circ\g^{\lambda}_L\circ \zeta_{\g^\lambda}^{-1}$, we obtain the unitary AMWI \eqref{eq:uAMWI}.

We get
\begin{align}
\frac{d}{d\lambda}\g^\lambda_\ast G[\phi]& =\frac{d}{d\lambda'}\Big\vert_{\lambda'=\lambda}\,(\g^{\lambda'}{\g^{\lambda}}^{-1})_{\ast}\g^{\lambda}_{\ast}G[\phi]\nonumber\\
&=\frac{d}{d\lambda'}\Big\vert_{\lambda'=\lambda}\,\g^{\lambda}_{\ast}G[\g^{\lambda'}{\g^{\lambda}}^{-1}\phi]\nonumber\\
&=\int d^4x\ \frac{\delta \g^{\lambda}_{\ast}G}{\delta\phi(x)}[\phi]\,X^{\lambda}\phi(x)
=\partial_{X^{\lambda}}\g^{\lambda}_\ast G[\phi]
\end{align}
(for $G\in\Floc(M)$); by using this result we perform the differentiation and obtain the condition
\be
S\bigl(\g^{\lambda}_L \zeta_{\g^\lambda}^{-1}(F)\bigr)\cdot_T\Bigl(\partial_{X^{\lambda}}\g^{\lambda}_L \zeta_{\g^\lambda}^{-1}(F)+
\partial_{X^{\lambda}}L+\g^{\lambda}_{\ast}\frac{d}{d\lambda}\zeta_{\g^\lambda}^{-1}(F)\Bigr)=0\ \mathrm{mod}\ \frac{\delta L}{\delta\phi}\ .
\ee
We insert the anomalous MWI \eqref{eq:anomMWI} and find 
\be
S\bigl(\g^{\lambda}_L \zeta_{\g^\lambda}^{-1}(F)\bigr)\cdot_T
\Bigl(\Delta X^{\lambda}(\g_L^{\lambda}\zeta_{\g^\lambda}^{-1}(F))+\g^{\lambda}_{\ast}\frac{d}{d\lambda}\zeta_{\g^\lambda}^{-1}(F)\Bigr)
=0\ \mathrm{mod}\ \frac{\delta L}{\delta\phi}\ .
\ee
We thus get the wanted family $\lambda\mapsto \zeta_{\g^\lambda}^{-1}$ as the unique solution of the differential equation
\be\label{eq:DGL-Z-1}
\frac{d}{d\lambda}\zeta_{\g^\lambda}^{-1}=-({\g_{\ast}^{\lambda}})^{-1}\Delta X^{\lambda}\g_L^{\lambda}\zeta_{\g^\lambda}^{-1}\ ,
\ee 
with the initial condition $\zeta_{\g^0}^{-1}=\mathrm{id}$. Since 
\be\label{eq:inR}
{\g_{\ast}^{\lambda}}^{-1}\Delta X^{\lambda}\g^{\lambda}_L\in\LieRc\ ,
\ee
as explained in the next paragraph, it follows that $\zeta_{\g^\lambda}^{-1}\in\Rc_c$, in particular $\zeta_\g\in\Rc_c$.

The claim \eqref{eq:inR} follows from $\Delta X^{\lambda}\in\LieRc$. In detail, let $\mu\mapsto\Ups^\lambda_\mu$ be a smooth curve in $\Rc_c$
with $\Ups^\lambda_0=\mathrm{id}_{\Floc}$ and $\Delta X^{\lambda}=\frac{d}{d\mu}\big\vert_{\mu =0}\Ups^\lambda_\mu$. Then, analogously  to Lemma \ref{lem:Zg-in-R}, using the formulas from Appendix \ref{app:Delta-X}, 
also ${(\g^{\lambda}}^{-1})_L\Ups^\lambda_\mu\, \g^{\lambda}_L=\delta_{{\g^\lambda}^{-1}}L+{\g_{\ast}^{\lambda}}^{-1}\Ups^\lambda_\mu\,\g^{\lambda}_L$ lies
in $\Rc_c$, and  \eqref{eq:inR} follows by applying $\frac{d}{d\mu}\vert_{\mu=0}$.

To prove the statement on the support of $\zeta_\g$ we choose $\g^{\lambda}$ such that $\supp\g^{\lambda}\subset\supp\g$, $0\le\lambda\le1$ and thus $\supp \Delta X^{\lambda}\subset\supp X^{\lambda}\subset\supp\g$. Let $F,G\in\Floc(M)$ with $\supp F\cap\supp\g=\0$.  We have $\supp (Z(F+G)-Z(G))\subset\supp F$ for $Z\in\Rc_c$ (see \cite[formula (6.3)]{BDF09} or Prop.~\ref{prop:inv-suppZ}) and hence 
\be
\supp \g^{\lambda}_{\ast}(\zeta_{\g^{\lambda}}^{-1}(F+G)-\zeta_{\g^{\lambda}}^{-1}(G))\subset\supp F\ .
\ee
Thus with $\g^{\lambda}_L(\bullet)=\g^{\lambda}_{\ast}(\bullet)+\delta_{\g^\lambda}L$ we find
\be
\Delta X^{\lambda}\g_L^{\lambda}\zeta_{\g^{\lambda}}^{-1}(F+G)=\Delta X^{\lambda}\bigl(\g^{\lambda}_L\zeta_{\g^{\lambda}}^{-1}(G)+\g^{\lambda}_{\ast}(\zeta_{\g^{\lambda}}^{-1}(F+G)-\zeta_{\g^{\lambda}}^{-1}(G))\bigr)
=\Delta X^{\lambda}\g^{\lambda}_L\zeta_{\g^{\lambda}}^{-1}(G)
\ee
and 
\be
\zeta_{\g^{-1}}(F+G)=F+G+\int_0^1d\lambda\,\,\frac{d}{d\lambda}\zeta_{\g^{\lambda}}^{-1}(G)=F+\zeta_\g^{-1}(G)
\ee
\ie $\supp\zeta_\g=\supp\zeta_\g^{-1}\subset\supp\g$.

It remains to show that $\zeta$ satisfies the cocycle identity. As shown in Proposition \ref{prop:AMWI-cocyclerelation}, 
the cocycle identity
\[
S\circ\zeta_{\g\h}(F)=S\circ\zeta_\h\circ(\zeta_\g)^\h(F)
\]
holds after evaluation of both sides on on-shell configurations.
However, $G\mapsto S(G)$ is not injective as a map from local functionals to functionals on-shell.
We solve this problem in the following way: we add a source term $\langle\phi,q\rangle$ 
to the free Lagrangian $L$. 
This does not change
the time ordered 
product. From \eqref{eq:anomMWI} we see that also the anomaly map $\Delta X$ is not changed; hence, this holds also for
the pertinent $\zeta$ satisfying the unitary AMWI. By Prop.~\ref{prop:AMWI-cocyclerelation} we conclude
\be
S\circ\zeta_{\g\h}(F)=S\circ\zeta_\h\circ(\zeta_\g)^\h(F)\ \mathrm{mod}\ \frac{\delta L}{\delta\phi}+q
\ee
for all sources $q$, hence for all field configurations. Since the  \emph{off-shell} S-matrix is injective,
we obtain the cocycle relation.

$(ii)$ The assertion is obtained by applying $\frac{d}{d\lambda}\vert_{\lambda=0}$ to the unitary AMWI \eqref{eq:uAMWI} for $\lambda\mapsto \g^\lambda$. The statement on the support follows from the fact that for any neighborhood $U$ of $\supp X$ we can find a smooth curve $\g^{\lambda}$ with $\supp\g^{\lambda}\subset U$ and $\frac{d}{d\lambda}\vert_{\lambda=0}\g^{\lambda}=X$. Thus $\supp \Delta X\subset U$ for all neighborhoods of $\supp X$ and hence $\supp\Delta X\subset\supp X$.
\end{proof}

If the unitary AMWI \eqref{eq:uAMWI} is anomaly free, then this holds also for the AMWI \eqref{eq:anomMWI}; explicitly, if $\zeta_{\g^\lambda}F=F$
for a certain $F\in\Floc(M)$ and in a neighbourhood of $\lambda =0$, then $\Delta X(F)=0$.
From the proof we see that the reversed statement holds in
the following sense: if, for a certain $F\in\Floc(M)$ and a suitable choice of the curve $(\g^\lambda)$, it holds that 
$\Delta X^{\lambda}(\g_L^{\lambda} F)=0$ for all $\lambda\in [0,1]$, then $S(\g_L F)=S(F)\, \mathrm{mod}\ \frac{\delta L}{\delta\phi}$.

\subsection{Scaling anomaly}\label{subsect:comp-anomalies}
We determine the anomaly in a special case -- the dilations in 4-dimen\-sio\-nal Minkowski space $\MM$, which are a combination of
a structure preserving embedding $\rho$ with an affine field redefinition $\Phi$. We study the massless real scalar field case,
{\it i.e.}, $L\doteq\frac{1}2\,(\partial\phi)^2$, and we work with the additional renormalization condition that the time-ordered product scales almost
homogeneously (see \cite{DF04}). Let $\RR\ni\lambda\mapsto \g^{\lambda}$ be a 1-parameter subgroup of $\G_c$ with generator $X\in\mathrm{Lie}\G_c$ which acts on field configurations as
\be\label{eq:Xphi}
X\phi(x)=\frac{d}{d\lambda}\Big\vert_{\lambda=0}\g^\lambda\phi(x)=\beta(x)\,(1+x^\mu\partial_\mu)\phi(x)
\ee

with $\lambda\in\RR$, $\beta\in\Dc(\MM,\RR)$. 
Let
\be\label{eq:assume}
\beta\vert_U=b\in\RR\ ,
\ee
for some open convex neighbourhood $U$ of the origin. Then for $x\in e^{-|b|}U$ we have
\be
\g^\lambda \phi(x)\doteq \phi(x\,e^{\lambda b})\,e^{\lambda b}\ ,\ \lambda\in [-1,1]\ .
\ee


In view of the AMWI we compute
\be
\partial_{X}L\overset{\eqref{eq:dXG}}{=}-\int d^4x\,\,\square\phi(x)\,X\phi(x)\overset{\eqref{eq:Xphi}}{=}
-\int d^4x\,\,\beta(x)\partial_{\mu}j^{\mu}(x)
\ee
with the dilation current
\be
j^{\mu}(x)=\bigl(\phi(x)+x^{\nu}\partial_{\nu}\phi(x)\bigr)\partial^{\mu}\phi(x)-\frac12x^{\mu}\partial_{\nu}\phi(x)\partial^{\nu}\phi(x)\ .
\ee

We want to compute $\Delta X(F)=\sum_{n=0}^\infty(\Delta X)^{(n)}(F^{\otimes n})/n!$ (understood as formal power series in $F$)
for the particular local functional
\be\label{eq:F}
F[\phi]\doteq\int d^4x\,\,f(x)\phi(x)^2/2\quad\text{with}\quad\supp f\subset e^{-|b|}U\ .
\ee
Note that $\supp \g^\lambda_\ast F\subset U$ for $\lambda\in[-1,1]$. By using \eqref{eq:dXG} we obtain
\be
\partial_XF[\phi]=b\int d^4x\,\,\frac{\delta F}{\delta\phi(x)}\,(1+x^\mu\partial_\mu)\phi(x)
=b\int d^4x\,\,f(x)\bigl(1+\frac12x^{\mu}\partial_{\mu}\bigr)\phi(x)^2\ .
\ee

The sequence $(\Delta X)^{(n)}(F^{\otimes n})$ can be computed by solving the unitary AMWI \eqref{eq:anomMWI} by induction on $n$,
see \cite[formula (5.15)]{Brennecke08} or \cite[formula (4.3.10)]{Due19}.
In 0th order we use the conservation law for the dilation current in the massless theory and conclude that $\Delta^{(0)}=0$, in agreement with the general results of \cite{Brennecke08}. In first order we obtain 
\be\label{eq:Delta-A1}
(\Delta X)^{(1)}(F)=\partial_X F+i\,F\cdot_T \partial_X L+i\int d^4 x\,\bigl(F\cdot_T X\phi(x)\bigr)\square\phi(x)
\ee
and in second order
\begin{multline}
(\Delta X)^{(2)}(F^{\otimes 2})=2i\,F\cdot_T\partial_X F+2i\,F\cdot_T (\Delta X)^{(1)}(F)-F\cdot_TF\cdot_T\partial_XL\\
-\int d^4 x\,\bigl(F\cdot_TF\cdot_T X\phi(x)\bigr)\square\phi(x)\ .
\end{multline}
Since the MWI holds true in classical field theory, that is, for tree diagrams,
the only potentially nonvanishing contributions come from \emph{local} terms of
the loops $F\cdot_T \partial_XL[0]$, $F\cdot_T\partial_XF[0]$ and $F\cdot_T F\cdot_T \partial_XL[0]$. Integrating by parts we note that
\begin{equation}
    \partial_X L=\int d^4x\,\,\partial_{\mu}\beta(x) j^{\mu}(x)\,,
\end{equation}
so due to our assumption on $\beta$ and $f$ given in \eqref{eq:assume} and \eqref{eq:F}, the only nonvanishing contribution to $\Delta X(F)$ comes from $F\cdot_T\partial_X F[0]$.

With $D^2$ being the renormalized fish diagram,{\it i.e.},
\be
D^2(z)=\bigl(D_F(z)\bigr)^2\big\vert_{\mathrm{renormalized}}=D^2(-z)
\ee
(where $D_F$ is the Feynman propagator of the massless scalar field), and by using the Action Ward Identity, we compute
\be
      F\cdot_T\partial_X F[0]=b\int d^4xd^4y\,\, f(x)f(y)\bigl(1+\frac12\,y^\mu\partial^y_\mu\bigr)D^2(x-y)\ .
\ee
We may symmetrize, that is, we may replace $(1+\frac12\,y^\mu\partial^y_\mu)D^2(x-y)$ by
\be
\bigl(1+\frac14(x^\mu\partial^x_\mu+y^\mu\partial^y_\mu)\bigr)D^2(x-y)=
\bigl(1+\frac14(x^\mu-y^{\mu})\partial^x_\mu\bigr) D^2(x-y)=\frac14\,\partial^x_{\mu}(x-y)^\mu D^2(x-y)\ .
\ee
Using $D_F(z)=\frac{-1}{4\pi^2}\,\frac{1}{z^2-i0}$ we obtain 
\begin{equation}\label{eq:identity-DF2}
    z^\mu D^2(z)=
    \frac{1}{8\pi^2}\partial^\mu D_F(z)\quad\text{for}\quad z\not=0\ .
\end{equation}
Moreover, since both sides have unique extensions from $\Dc'(\RR^4\setminus\{0\})$ to
$\Dc'(\RR^4)$ which preserve almost homogeneous scaling, they agree everywhere.  
Inserting this result, we end up with
\be
    \Delta X(F)=\frac12\,(\Delta X)^{(2)}(F^{\otimes 2})=i\,F\cdot_T\partial_X F[0]
    =\frac{b}{32\pi^2}\int d^4x\,\,(f(x))^2\ .
\ee

Here we have taken into account that $(\Delta X)^{(n)}(F^{\otimes n})=0$ for all $n\geq 3$. Proceeding by induction on $n$ this can be seen as follows
(for details see \cite{Brennecke08,DF04} and \cite[Chaps.~3.1 and 4.3]{Due19}): writing 
$$\int d^4y\prod_{j=1}^n \bigl(d^4x_j\,f(x_j)\bigr)\,\Delta^{(n)}(x_1,\ldots,x_n;y)\doteq(\Delta X)^{(n)}(F^{\otimes n})$$ 
the terms contributing to $\Delta^{(n)}(x_1,\ldots,x_n;y)$ are (up to constant prefactors)
\begin{align}\label{eq:Delta-n}
    &\bigl(\phi^2(x_1)\cdot_T\ldots\cdot_T\phi^2(x_{n-1})\cdot_T\phi(y)\,X\phi(y)\bigr)\,\delta(x_n-y),\quad
    \bigl(\phi^2(x_1)\cdot_T\ldots\cdot_T\phi^2(x_n)\cdot_T\square\phi(y)\,X\phi(y)\bigr)\nonumber\\
&\bigl(\phi^2(x_1)\cdot_T\ldots\cdot_T\phi^2(x_n)\cdot_T X\phi(y)\bigr)\,\square\phi(y),\quad
\bigl(\phi^2(x_1)\cdot_T\ldots\cdot_T\phi^2(x_{n-2})\bigr)\,b\delta(x_{n-1}-y,x_n-y)\ ,
\end{align}
note that the last term is coming from 
$$\bigl(F^{\cdot_T(n-2)}\cdot_T(\Delta X)^{(2)}(F^{\ox 2})\bigr)=\bigl(F^{\cdot_T(n-2)}\bigr)\,(\Delta X)^{(2)}(F^{\ox 2})\ .$$
Now we use that 
\be\label{eq:supp-Delta}
\supp\Delta^{(n)}(x_1,\ldots,x_n;y)\subset\Delta_{n+1}\doteq\{(x_1,\ldots,x_n,y)\in \MM^{n+1}\,\vert\,x_1=\ldots=x_n=y\}\ ,
\ee
that is, on $\MM^{n+1}\setminus\Delta_{n+1}$ the contributions to 
$\Delta^{(n)}(x_1,\ldots,x_n;y)$ cancel out. This holds even on the whole manifold $\MM^{n+1}$.
To wit, looking at the causal Wick expansion of the terms given in \eqref{eq:Delta-n}, all coefficients (which are $\CC$-valued distributions depending 
on the relative coordinates) scale almost homogeneously with a degree smaller than $4n$, 
as one verifies by power counting.
Since we require that the extension of the time ordered products to the thin diagonal preserves almost homogeneous scaling, terms 
proportional to $\partial^a\delta(x_1-y,\ldots,x_n-y)$ cannot be produced.

\medskip

\begin{proposition}\label{prop:scale-anomaly}
The anomaly $\zeta$ associated to the 1-parameter group $(\g^{\lambda})$ acts on \[F[\phi]=\int d^4x f(x)\phi(x)^2/2\] as
\begin{equation}\label{eq:scale}
     \zeta_{\g^{\lambda}}(F)=F+\frac{\lambda b}{32\pi^2}\int d^4x\,\,(f(x))^2 +c(\lambda)  \ ,\ \lambda\in[-1,1] 
\end{equation}
 with constant functionals $c(\lambda)$ not depending on $F$ with $c(0)=0$ and $c'(0)=0$.
 Moreover $c(\lambda)$ can be removed by a suitable renormalization.   
\end{proposition}
\begin{proof}
Obviously, the formula \eqref{eq:scale} for $\zeta_{\g^\lambda}$ satisfies the necessary condition $\frac{d}{d\lambda}\vert_{\lambda=0}\zeta_{\g^\lambda}(F)=\Delta X(F)$.
We have to show that there exist a $c(\lambda)$ (with the mentioned properties) such that $\zeta_{\g^{\lambda}}(F)$ satisfies the 
differential equation \eqref{eq:DGL-Z-1}, that is,
\begin{equation}\label{eq:DGL-zeta-1}
   -\frac{d}{d\lambda}\zeta_{\g^{\lambda}}^{-1}(F)=({\g_{\ast}^{\lambda}})^{-1}\Delta X(\g^{\lambda}_L\zeta_{\g^{\lambda}}^{-1}(F))\ 
\end{equation}
by using that $X^\lambda=X$. 

First note that, since $Z(G+a)=Z(G)+a$ for any $Z\in\Rc_c$, $G\in\Floc(M)$ and $a\in\CC$, the assertion \eqref{eq:scale} 
can equivalently be written as
\begin{equation}\label{eq:scale-1}
     \zeta_{\g^{\lambda}}^{-1}(F)=F-\frac{\lambda b}{32\pi^2}\int d^4x\,\,(f(x))^2 -c(\lambda)  \ ,\ \lambda\in[-1,1] \ .
\end{equation}

Now we insert \eqref{eq:scale-1} into the r.h.s.~of \eqref{eq:DGL-zeta-1} and use the facts that $\Delta X(G+C)=\Delta X(G)$ for $G\in\Floc(\MM)$ 
and any constant functional $C$ (as one easily sees from \eqref{eq:anomMWI}), that $\supp\delta_{\g^{\lambda}}L\cap\supp\g^{\lambda}_{\ast}F=\0$ and that $\Delta X$ is additive on functionals with disjoint support (property $(iii)$ of $z\in\LieRc$ in Appendix \ref{app:Delta-X}) and obtain
\begin{equation}
    ({\g_{\ast}^{\lambda}})^{-1}\Delta X(\g^{\lambda}_L\zeta_{\g^{\lambda}}^{-1}(F))=({\g_{\ast}^{\lambda}})^{-1}\Delta X(\delta_{\g^{\lambda}}L)+({\g_{\ast}^{\lambda}})^{-1}\Delta X(\g^{\lambda}_{\ast}F)\ .
\end{equation}
We have $\g^{\lambda}_{\ast}F=\int d^4x\,\,\bigl(e^{-2\lambda b}f(e^{-\lambda b}x)\bigr)\,\phi(x)^2/2$, hence
\begin{equation}
    ({\g_{\ast}^{\lambda}})^{-1}\Delta X(\g^{\lambda}_{\ast}F)=
    \frac{b}{32\pi^2}\int d^4x\,\, e^{-4\lambda b}f(e^{-\lambda b}x)^2
    =\frac{b}{32\pi^2}\int d^4x\,\, f(x)^2 \ ,
\end{equation}
Inserting \eqref{eq:scale-1} also into the l.h.s.~of \eqref{eq:DGL-zeta-1}, we see that \eqref{eq:scale-1} satisfies \eqref{eq:DGL-zeta-1} iff
\be
\frac{d}{d\lambda}c(\lambda)=({\g_{\ast}^{\lambda}})^{-1}\Delta X(\delta_{\g^{\lambda}}L)\ ;
\ee
taking also into account that $c(0)=0$, $c(\lambda)$ is uniquely fixed. 

To remove $c$ we use a renormalization $Z\in\Rc_0$ with
\begin{equation}\label{eq:Z-c}
    Z(F)=F\ ,\quad Z(\delta_{\g^{\lambda}}L)=\delta_{\g^{\lambda}}L-c(\lambda)\ .
\end{equation}
Since $F$ is quadratic in $\phi$, we may use the result of Prop.~\ref{prop:betaZ-cocycle} that the cocycle belonging to the renormalized 
time ordered product is the equivalent cocycle $\zeta'_{\g^{\lambda}}=Z^{-1}\zeta_{\g^{\lambda}}Z^{\g^{\lambda}}$. For the latter we find
\begin{equation}
\begin{split}
    c'(\lambda)\overset{\eqref{eq:scale}}{=}\zeta'_{\g^{\lambda}}(0)&=Z^{-1}\zeta_{\g^{\lambda}}Z^{\g^{\lambda}}(0)\\
    &=Z^{-1}\zeta_{\g^{\lambda}}({\g^{\lambda}_L})^{-1}Z(\delta_{\g^{\lambda}}L)\\
    &=Z^{-1}\zeta_{\g^{\lambda}}({\g^{\lambda}_L})^{-1}\bigl(\delta_{g^{\lambda}}L-c(\lambda)\bigr)\\
    &=Z^{-1}\zeta_{\g^{\lambda}}\bigl(\delta_{(\g^{\lambda})^{-1}}L+({\g^{\lambda}_{\ast}})^{-1}\delta_{\g^{\lambda}}L-c(\lambda)\bigr)\\
    &=Z^{-1}\zeta_{\g^{\lambda}}\bigl(-c(\lambda)\bigr)=0
\end{split}
\end{equation}
where we used
\begin{equation}
    \delta_{(\g^{\lambda})^{-1}}L+({g_{\ast}^{\lambda}})^{-1}\delta_{\g_{\lambda}}L=0\quad\text{and}\quad 
    \zeta_{\g^{\lambda}}\bigl(0-c(\lambda)\bigr)=\zeta_{\g^{\lambda}}\bigl(0\bigr)-c(\lambda)=0\ .
\end{equation}
This concludes the proof.
\end{proof}

We can also compute the action of global scaling transformations using the concepts of Section \ref{sec:RGF}. Namely, let $\h^{\lambda}\in\HL$ act on configurations $\phi$
as
\begin{equation}
    \h^\lambda \phi(x)=e^{\lambda b}\phi(xe^{\lambda b})\ ,\ x\in\MM\ .
\end{equation}
Then, given $f$ and $\lambda$, we choose $U$ such that $\supp f\subset e^{-|\lambda b|}U$ and find for all such $U$
\begin{equation}
  \theta_{\h^{\lambda}}(F)\overset{\eqref{eq:theta-h}}{=}\zeta_{\g^{\lambda}}(F)-\zeta_{\g^{\lambda}}(0)=F+\frac{\lambda b}{32\pi^2}\int d^4x\,\,f(x)^2\ .  
\end{equation}
The induced change of the Lagrangian vanishes, 
since for $g\in\Dc(\MM,\RR)$ we have 
\begin{equation}
    \delta_{\zeta,\h^{\lambda}}L(g)[\phi]\overset{\eqref{eq:d-zeta-h-L}}{=}
    \zeta_{\g^{\lambda}}(0)[g\phi]-\zeta_{\g^{\lambda}}(0)[0]=0\ .
\end{equation}
Finally, we may look at the anomalous Noether Theorem \ref{theorem:anomalousnoethertheorem} which simplifies in our case since the transformation $\g^{\lambda}$ does not change the causal structure. Following the proof of that theorem, we split 
\begin{equation}
    \delta_{\g^{\lambda}}L=Q_++Q_-
\end{equation}
such that $\supp Q_+\cap J_-(e^{|\lambda b|}U)=\0$ and $\supp Q_-\cap J_+(e^{|\lambda b|}U)=\0$ and find for $\supp f\subset U$ and $\beta(x)=b$ for $x\in e^{|\lambda b|}U$
\begin{equation}
S(\theta_{\h^{\lambda}}(F))=S(Q_-)^{-1}S(h^{\lambda}_{\ast}F)S(Q_-) \ .
\end{equation}

\subsection{Axial anomaly}
Another famous example of an anomaly is the axial anomaly. For a massless Dirac field $\psi$ in  $4$-dimensional Minkowski space the axial current
\be
j^a_{\mu}\doteq\overline{\psi}\gamma_{\mu}\gamma^5\psi
\ee
is conserved as a consequence of the Dirac equation. It is the Noether current corresponding to the symmetry
\be
\g\psi(x)\doteq e^{i\alpha(x)\gamma^5}\psi(x),\quad \g\overline{\psi}(x)\doteq\overline{\g\psi(x)}=\overline{\psi}(x)\,e^{i\alpha(x)\gamma^5},
\quad\alpha\in\Dc(\MM,\RR) \ ,
\ee
namely with the free Lagrangian
\be
L\doteq i\overline{\psi}\gamma^\mu\partial_\mu \psi
\ee
we have
\be\label{eq:dgL}
\delta_{\g}L=\int d^4x\,\, \alpha(x)\partial^{\mu}j^a_{\mu}(x)\ .
\ee
We compute the anomaly $\zeta_{\g}$ on $F=\int d^4x\,\, j_{\mu}(x)A^{\mu}(x)$ with the vector current $j_{\mu}=\overline{\psi}\gamma_{\mu}\psi$ and an external electromagnetic potential $A^{\mu}\in\Dc(\MM,\RR^4)$. 
As shown in \cite{BDFR21} the formalism of the present paper can be extended to Fermi fields by adding external Grassmann parameters in an appropriate way. For quadratic expressions in the basic Dirac field these parameters are not needed.

As in the case of scaling we choose a 1-parameter group $\g^{\lambda}$ with 
$\g^{\lambda}\psi(x)\doteq e^{i\lambda\alpha(x)\gamma^5}\psi(x)$. Its generator $X$ acts on $\psi$ and $\overline{\psi}$ as
\be
X\psi(x)=\frac{d}{d\lambda}\Big\vert_{\lambda=0}\g^{\lambda}\psi(x)=i\alpha(x)\gamma^5\psi(x)\ ,\quad
X\overline{\psi}(x)=\overline{\psi}(x)i\alpha(x)\gamma^5\ .
\ee
Since $\g^\lambda_\ast F=F$ we obtain $\partial_X F=\frac{d}{d\lambda}\vert_{\lambda=0}\,\g^\lambda_\ast F=0$ and, by using \eqref{eq:dgL}, we get
$\partial_X L=\frac{d}{d\lambda}\vert_{\lambda=0}\,\delta_{\g^\lambda}L=\delta_{\g}L$.
Since $\g^{\lambda_1}\g^{\lambda_2}=\g^{\lambda_1+\lambda_2}$ it holds that
\be\label{eq:X-lambda}
X^\lambda\g^\lambda\phi(x)=\frac{d}{d\lambda'}\Big\vert_{\lambda'=0}\g^{\lambda'}\g^{\lambda}\phi(x)=X\g^\lambda\phi(x),\quad
\text{that is,}\quad X^\lambda=X.
\ee

As in the case of scaling only the divergent loop graphs contribute to the anomaly. We consider the distributions $D_{\bullet}\in\Dc'(\RR^{4n})$
\be
D_{\mu_1,\dots,\mu_n;\nu}(x_1-y,\dots,x_n-y)\doteq \bigl(j_{\mu_1}(x_1)\cdot_T \dots\cdot_T j_{\mu_n}(x_n)\cdot_T j^a _{\nu}(y)\bigr)^c[0]\ , 
\ee
the upper index ``c'' means that we select the contribution of all connected diagrams.
By Furry's theorem (which is a consequence of charge conjugation invariance of the time ordered product) $D_{\mu_1,\dots,\mu_n;\nu}$ vanishes 
for $n$ odd. By using the inductive Epstein-Glaser construction of the time ordered product, one shows
that the divergence of $D_{\mu_1,\dots,\mu_n;\nu}$ with respect to $\nu$ and $y$ vanishes outside of the origin and is therefore a derivative of the $\delta$-function
\be
\sum_{i=1}^n\partial^{\nu}_{x_i}D_{\mu_1,\dots,\mu_n;\nu}\doteq p_{\mu_1,\dots,\mu_n}(\partial)\delta
\ee
where $p_{\bullet}$ is a family of homogeneous polynomials of the partial derivatives $\partial^{\nu_i}_{x_i}$ with degree $3n+4-4n=4-n$. It is symmetric under permutations of the index $i$ and odd under parity. The only nontrivial case is $n=2$ where 
$p_{\mu_1\mu_2}(\partial)=c\epsilon_{\mu_1\mu_2\rho\sigma}\partial^{\rho}_{x_1}\partial^\sigma_{x_2}$
with some constant $c\in\RR$. An explicit calculation yields $c=\frac{1}{2\pi^2}$ 
(for a derivation of this result based on the Epstein-Glaser method, 
see \cite{DKS91}), under the condition that the Ward identities for the vector current are satisfied,
{\it i.e.}, $\partial^{\mu_i}_{x_i}D_{\mu_1,\dots,\mu_n;\nu}=0$ for all $1\leq i\leq n$.

We compute $\Delta X$ as in the previous section and find
\be
\begin{split}
    \Delta X(F)&=-\frac12 F\cdot_T F\cdot_T\partial_XL[0]\\
    &=-\frac12\int d^4x_1d^4x_2d^4y\,\,A^{\mu_1}(x_1)A^{\mu_2}(x_2)\alpha(y)\,\partial_y^\nu D_{\mu_1,\mu_2;\nu}(x_1-y,x_2-y)\\
    &=\frac12\int d^4x_1d^4x_2d^4y\,\,A^{\mu_1}(x_1)A^{\mu_2}(x_2)\alpha(y)\,\,c\epsilon_{\mu_1\mu_2\rho\sigma}\partial^\rho_{x_1}
    \partial^{\sigma}_{x_2}\delta(x_1-y)\delta(x_2-y)\\
    &=\frac{c}{2}\int d^4y\,\,\alpha(y)(\partial^{\rho}A^{\mu_1})(y)(\partial^{\sigma}A^{\mu_2})(y)\epsilon_{\mu_1\mu_2\rho\sigma}\\
    &=-\frac{1}{16\pi^2}\int \alpha\, dA\wedge dA
\end{split}
\ee
with the 1-form $A\doteq A_{\mu}dx^{\mu}$.

In order to compute $\zeta_{\g^{\lambda}}(F)$ we also have to determine $\Delta X(\g^{\lambda}_L F)$. We first study
$\Delta X(\delta_{\g^{\lambda}}L)$. Also here the only possible contribution comes from the triangle diagram
\begin{equation}\label{eq:axial-triangle}
    \begin{split}
        \Delta X(\delta_{\g^{\lambda}}L)&=-\frac12(\delta_{\g^{\lambda}}L\cdot_T \delta_{\g^{\lambda}}L\cdot_T\partial_X L)^c[0]\\
        &=\lambda^2\int d^4xd^4yd^4z\,\,\alpha(x)\alpha(y)\alpha(z)\,\partial^\mu_x\partial^\nu_y\partial^\rho_z\bigl(j_\mu^a(x)\cdot_T j_\nu^a(y)\cdot_Tj_\rho^a(z)\bigr)^c[0]
    \end{split}
\end{equation}
since for an even number of factors $j^a$ the correlation functions can be renormalized to coincide with those of the vector current, and thus their divergences vanish, and for 5 factors the divergence of the correlation function with respect to one factor is of the form
\be
a_{\mu\nu\rho\sigma}\delta
\ee
with a tensor of 4th order which is symmetric with odd parity, hence $a_{\bullet}=0$.
The divergence of the triangle diagram does not vanish, actually it is of the form
\begin{equation}
   \partial^\rho_z\bigl(j_\mu^a(x)\cdot_T j_\nu^a(y)\cdot_Tj_\rho^a(z)\bigr)^c[0]=-\frac{1}{6\pi^2}\epsilon_{\mu\nu\sigma_1\sigma_2}\partial^{\sigma_1}_{x}\partial^{\sigma_2}_y\delta(x-z,y-z) 
\end{equation}
(see, \eg, \cite{DKS91}). Inserting it into \eqref{eq:axial-triangle} we see that also this term does not contribute, \ie $\Delta X(\delta_{\g^{\lambda}}L)=0$.

Also mixed terms which possibly could contribute to the value of $\Delta X$ on the sum  $G\doteq\g^{\lambda}_L F=F+\delta_{\g^{\lambda}}L$ vanish.
To explain this first note that $\partial_X G=\partial_X \delta_{\g^{\lambda}}L=0$, since $\g^\lambda_\ast j^a=j^a$. 
In second order in $G$ the statement follows 
from $\bigl(j\cdot_T j^a\cdot_T \partial j^a\bigr)^c[0]=\bigl(j\cdot_T j\cdot_T \partial j\bigr)^c[0]=0$
by Furry Theorem. For $n>2$ we obtain
\be
(\Delta X)^{(n)}(G^{\ox n})[0]=(\Delta X)^{(n)}(G^{\ox n})^c[0]=i^n\bigl(G^{\cdot_T n}\cdot_T\partial_X L\bigr)^c[0]\ ,
\ee
where we first use that $(\Delta X)^{(n)}$ is supported on the thin diagonal (see \eqref{eq:supp-Delta})
and then \eqref{eq:anomMWI}. In addition we have taken into account that 
(since $(\Delta X)^{(2)}(F^{\ox 2})\in\RR$) there is no connected diagram contributing to 
$\bigl(G^{\cdot_T (n-2)}\cdot_T(\Delta X)^{(2)}(F^{\ox 2})\bigr)=\bigl(G^{\cdot_T (n-2)}\bigr)\,(\Delta X)^{(2)}(F^{\ox 2})$.
By suitable renormalization (respecting the mentioned renormalization conditions) one can reach that
\be
\bigl(j(x_1)\cdot_T\ldots\cdot_T j(x_k)\cdot_T j^a(x_{k+1})\cdot_T\ldots\cdot_T j^a(x_n)\cdot_T \partial j^a(y)\bigr)^c[0]=0\quad\text{for}\quad n>2\ .
\ee
Hence,
\be
\Delta X(G)=\frac{1}2(\Delta X)^{(2)}(F^{\ox 2})=\Delta X(F)\ .
\ee
So we obtain
\begin{equation}
    {\g^{\lambda}_{\ast}}^{-1}\Delta X\g^{\lambda}_L(F)=\Delta X(F)\ .
\end{equation}

Proceeding analogously to the scaling anomaly (Prop.~\ref{prop:scale-anomaly}) we can now 
solve the differential equation for $\zeta_{\g^{\lambda}}(F)$ and obtain
\be
\zeta_{\g^\lambda}(F)=F-\frac{\lambda }{16\pi^2}\int \alpha \, dA\wedge dA\ .
\ee

\section{Conclusions and outlook}
In the program of constructing algebraic quantum field theories \cite{BF19} we succeeded in incorporating both aspects of causality: the causal independence of spacelike separated regions as well as a dynamical law by which future and past in a region of causal dependence are fixed. We have proven the time-slice axiom and constructed the general expression for the relative Cauchy evolution. We have used the latter to obtain the stress-energy tensor as an unbounded operator, improving on the results of \cite{BFV03}, where only the derivation obtained as the commutator with the stress-energy tensor could be reconstructed.

In addition to a classical Lagrangian which fixes the dynamics only in the case of the free theory and the subalgebra of Weyl operators, we introduced a cocycle on a group of classical symmetries with values in the renormalization group. Together with the Lagrangian this specifies the dynamics. Moreover, it describes whether classical symmetries of the Lagrangian are unbroken in the quantized theory, and allows a direct characterization of the renormalization group flow induced by anomalies. This means that anomalies appear when a classical symmetry is broken in the process of quantization and the departure from the classical expression is quantified in terms of a certain renormalization group cocycle. The transformation  of the S-matrix obtained this way may be interpreted as a \emph{quantum} symmetry arising from  the classical symmetry modified by the cocycle. Such quantum symmetries can then be unitarily implemented, which is the content of our \emph{anomalous Noether theorem}. We have also shown that in perturbation theory, the derivative of our cocycle is related to the BV Laplacian or the anomaly term in the perturbative anomalous Master Ward identity. This emphasizes the fact that our formulation indeed allows one to upgrade classical symmetries to quantum symmetries and their relative difference is reflected by the presence of anomalies.

There is one essential point missing in our construction, namely the implication of the spectrum condition, related to the existence of a vacuum, stability of states \etc \cite{Borchers}. There have been various attempts to understand this implication for the structure of the algebra, starting from Sergio Doplicher's “algebraic spectrum condition" \cite{Dop65}, including Rainer Verch's approach to an algebraic concept of wave front sets \cite{Verch}, but it is fair to say that there is not yet a fully satisfactory answer. From our experience with perturbation theory we know that the spectrum condition imposes constraints on the choice of cocycles which lead to the occurence of anomalies. Ignoring the slight difference between $\Rc(M,L_0)$ and the St\"uckelberg-Petermann renormalization group (see Remarks \ref{rm:SP} 
and \ref{rm:main-thm}), we strongly presume that, due to the main theorem of renormalization \cite{PopineauS16,DF04,BDF09} and 
Prop.~\ref{prop:betaZ-cocycle}, the equivalence class of the cocycle is uniquely fixed.
We may therefore formulate the remaining open problem in
the algebraic construction of quantum field theories as the problem to determine this equivalence class. 

In this paper we treated only scalar theories, but we included also an example of computation of an anomaly for fermions along the line of \cite{BDFR21}. It would be desirable to cover also gauge theories.
We plan to return to this problem in future work.

\begin{appendix}
\section{Functionals and generalized fields}\label{app:gen-field}
\begin{definition}\label{eq:def-F-A(f)}
The (functional) support of a map $F:\Ec(M)\to\CC$
is the smallest closed subset $N$ of $M$ such that $F[\phi+\psi]=F[\phi]$
for all $\phi,\psi\in\Ec(M)$ with $\supp\psi\cap N=\emptyset$.

A local functional $F\in\Floc(M)$ is a map $F:\Ec(M)\to\RR$ with compact support which satisfies the Hammerstein relation
\be\label{eq:Hammer}
F(\phi+\chi+\psi)=F(\phi+\chi)-F(\chi)+F(\chi+\psi)
\ee
for $\phi,\chi,\psi\in\Ec(M)$ with $\supp\phi\cap\supp\psi=\0$.

A generalized field is a map $A:\Dc(M)\to\Floc(M)$ with $\supp A(f)\subset\supp f$ such that
\be
A(f+g+h)=A(f+g)-A(g)+A(g+h)
\ee
whenever $\supp f\cap\supp h=\0$. Two generalized fields $A, A'$ are equivalent if
\be
\supp (A-A')(f)\subset\supp(f-1)\quad\forall f\in\Dc(M)\  .
\ee
The support of a generalized field $A$ is defined by
\be
\supp A=\{x\,|\,x\in\supp A(f)\ \forall \ f\equiv 1 \text{ near }x\}
\ee
\end{definition}

Note that $x\in\supp A(f)$ for every $f$ with $f\equiv 1$ near $x$ if $x\in\supp A(f')$ for some $f'$ with $f'\equiv 1 $ near $x$. 
Namely we can split $f'=f_0+f_1+f_2$ with $f=f_0+f_1$ and $x\not\in\supp f_{1,2}$ and $\supp f_2\cap\supp f_0=\0$. Then
\be\label{eq:A(f')}
A(f')=A(f)-A(f_1)+A(f_1+f_2)
\ee
and 
\begin{align*}
\supp A(f')&\subset\supp A(f)\cup\supp A(f_1)\cup\supp A(f_1+f_2)\\
&\subset \supp A(f)\cup\supp f_1\cup\supp f_2\ ,
\end{align*}
hence $x\in\supp A(f)$ if $x\in\supp A(f')$.

Obviously, for a generalized field $A$ it holds that
\be
\supp A\subset\overline{\bigcup_{f\in\Dc(M)}\supp A(f)}\ .
\ee

\begin{proposition}
Equivalent generalized fields have the same support.
\end{proposition}

\begin{proof}
Let $A'$ be equivalent to $A$ and $x\not\in\supp A$. Then there exists an $f_0\equiv1$ near $x$ 
such that $x\not\in \supp A(f_0)$. Since
$$
\supp A'(f_0)\subset\supp A(f_0)\cup\supp(A'-A)(f_0))\subset\supp A(f_0)\cup\supp(f_0-1)\ ,
$$
we see that $x\not\in\supp A'(f_0)$, hence $x\not\in\supp A'$.
\end{proof}

We can characterize equivalence classes of generalized fields by their relative action.
\begin{definition}
Let $A$ be a generalized field. The relative action of $A$ is a map $\delta A:\Dc(M)\to\Floc(M)$ defined by
\be
\delta A(\psi)[\phi]=A(f)[\phi+\psi]-A(f)[\phi] \ ,\ f\equiv 1 \text{ on }\supp \psi\ .
\ee
\end{definition}

To see that this definition does not depend on the choice of $f$, let $f'$ be another choice. We split
$f'=f_0+f_1+f_2$ with $f=f_0+f_1$ and $\supp\psi\cap\supp f_{1,2}=\emptyset$ and $\supp f_2\cap\supp f_0=\0$. Then,
the relation \eqref{eq:A(f')} holds true; and since $\supp A(f_1)\subset\supp f_1$ and 
$\supp A(f_1+f_2)\subset\supp f_1\cup\supp f_2$, we have $A(f_1)[\phi+\psi]=A(f_1)[\phi]$ and 
$A(f_1+f_2)[\phi+\psi]=A(f_1+f_2)[\phi]$. This yields the assertion.

\begin{proposition}
Two generalized fields are equivalent iff their relative actions are equal.
\end{proposition}
\begin{proof}
Let $A,A'$ be generalized fields. If they are equivalent, then $\supp (A-A')(f)\subset\supp(f-1)$ for all $f$. 
Let $f\equiv 1$ on $\supp\psi$. We then have
\be
\delta A(\psi)-\delta A'(\psi)=A(f)[\bullet+\psi]-A(f)-A'(f)[\bullet+\psi]+A'(f)
=(A-A')(f)[\bullet+\psi]-(A-A')(f)=0
\ee 
since $\supp (f-1)\cap\supp \psi=\0$. 

Let on the other side $\delta A=\delta A'$, $f\in\Dc(M)$ arbitrary and $x\not\in\supp (f-1)$. 
There exists a neighborhood $V$ of $x$ such that $f\equiv1$ on $V$.
Then for all $\psi$ with $\supp\psi\subset V$
\be
(A-A')(f)[\bullet+\psi]=\delta A(\psi)-\delta A'(\psi)+(A-A')(f)=(A-A')(f)
\ee
hence $x\not\in\supp(A-A')(f)$, \ie $A$ and $A'$ are equivalent.
\end{proof}
The proposition leads to a criterion for the support of a generalized field:

\begin{proposition}\label{prop:suppA=suppdA}
Let $A$ be a generalized field and let $\supp\delta A$ be 
the smallest closed subset $N$ of $M$ such that $\delta A(\psi)=0$ for all $\psi\in\Dc(M\setminus N)$. 
Then it holds that  
\be
\supp A=\supp\delta A\ .
\ee
\end{proposition}
\begin{proof}
Let $x\in\supp A$. Then for any neighborhood $V$ of $x$ and any $f\equiv 1$ on $V$ there exists some $\psi$ with $\supp \psi\subset V$ such that 
\be
0\neq A(f)[\bullet+\psi]-A(f)=\delta A(\psi)\ .
\ee
The opposite inclusion follows by essentially the same argument.
\end{proof} 
We now look at generalized fields with compact support and set
\be
A(1)[\phi]=\delta A(\chi\phi)[0]\ ,\ \chi\in\Dc(M)\ ,\ \chi\equiv1\text{ on }\supp A
\ee
$A(1)$ does not depend on $\chi$, namely for $\chi'\equiv1$ on $\supp A$ we have
\be
\delta A(\chi'\phi)[0]=A(f)[\chi'\phi]-A(f)[0]=A(f)[\chi'\phi]-A(f)[\chi\phi]+\delta A(\chi\phi)[0]\ .
\ee
(where $f\equiv 1$ on $\supp\chi\cup\supp\chi'$) and  $\supp(\chi'-\chi)\cap\supp A=\0$, hence 
$A(f)[\chi'\phi]-A(f)[\chi\phi]=\delta A((\chi'-\chi)\phi)[\chi\phi]=0$. We also have $A(1)=A'(1)$ for $A,A'$ equivalent. 

We have found a map $A\mapsto A(1)$ from generalized fields with compact support to local functionals. 
On the other hand, we can also consider a map from local functionals 
to generalized fields
\be\label{eq:AF}
F\mapsto A_F\ ,\ A_F(f)[\phi]=F[f\phi]\ .
\ee 
The associated relative action can be written as
\be
\delta A_F(\bullet)[\phi]=F[\phi+\bullet]-F[\phi]
\ee 
by choosing $f\equiv 1$ on $\supp F$; hence $\supp A_F=\supp\delta A_F=\supp F$.
We also find
\be
A_F(1)[\phi]=\delta A_F(\chi\phi)[0]=F[\chi\phi]-F[0]=F[\phi]\ .
\ee
with $\chi\equiv1$ on $\supp F$. 
Moreover, the generalized field $A_{A(1)}$ built from the local functional $A(1)$ is equivalent to $A$,
since $\delta A_{A(1)}(\psi)=A(1)[\bullet+\psi]-A(1)[\bullet]=\delta A(\psi)$.
%
%
\section{Interpolating metrics}\label{sec:interpolatingmetrics}
In this Appendix we show that, given two Lorentz metrics
$g_0,g_1$ on the manifold $M$ for which the manifold becomes globally hyperbolic, there exists a sequence of 5 metrics, starting with 
$g_0$ and ending with $g_1$, such that for each neighbouring pair all pointwise convex combinations are 
Lorentz metrics for which $M$ is globally hyperbolic, see \eqref{eq:g-convex-comb}.

We choose 
time functions $t_i$ associated to $g_i$, $i=0,1$,
with timelike differentials $dt_i$ such that $\supp (t_1-t_0)$ is compact and such that the convex combinations $\lambda dt_1+(1-\lambda) dt_0$ nowhere vanish, $0\le\lambda\le 1$ (always possible in more than 2 dimensions).

There exists a vector field $X$ with $\langle dt_i,X\rangle=1$, $i=0,1$. In a first step we define metrics $g_i'$ with larger lightcones (\ie $g_i'\ge g_i$) for which $X$ is timelike. In the second step we define a metric $g_{01}$ whose lightcone contains $X$ and is contained in the lightcones of both metrics. We obtain a sequence of 5 metrics,
\begin{equation}
   g_0\le g_0'\ge g_{01}\le g_1'\ge g_1 
\end{equation}
such that for each neighbouring pair all convex combinations are globally hyperbolic Lorentz metrics.

In detail we construct the metrics $g'_0,\,g'_1$ and $g_{01}$ as follows:
we define Riemannian metrics $\gamma_{i}$ by
\begin{equation}\begin{split}
    &\gamma_i(Y,Y)=2a_i\langle dt_i,Y\rangle^2-g_i(Y,Y)\ ,\ i=0,1\ .
    \end{split}
\end{equation}
with $a_i=(g_i^{-1}(dt_i,dt_i))^{-1}$.
Let $d_i=\gamma_i(X,X)$. 
We set
\begin{equation}
    g_i'\doteq k dt_i^2-\gamma_i
\end{equation}
with $k>1+d_0,1+d_1,2a_0,2a_1$ to be determined later, and find that 
$X$ is timelike for $g'_0$ as well as for $g'_1$ and that $g_i'\ge g_i$, $i=0,1$.
We then construct $g_{01}$ by
\begin{equation}
    g_{01}\doteq b\, dt_0dt_1-c(dt_0^2+dt_1^2)-\gamma_0-\gamma_1
\end{equation}
where $b,c>0$ are chosen such that $X$ is timelike for $g_{01}$ and such that each vector field $Y$ which is timelike for $g_{01}$ and future directed with respect to $t_0+t_1$ is also timelike and future directed for $g_i'$, $i=0,1$.

The condition that $X$ is timelike for $g_{01}$ requires $b>2c+d_0+d_1$. Let now $Y$ be a future directed vector field with respect to $g_{01}$ and $t_0+t_1$. We see immediately that then $\langle dt_i,Y\rangle>0$ for both values of $i$. We then use the inequality
\begin{equation}
   2\langle dt_0,Y\rangle\langle dt_1,Y\rangle\le(\lambda \langle dt_0,Y\rangle^2+\lambda^{-1}\langle dt_1,Y\rangle^2)\ ,\lambda>0 
\end{equation}
with $\lambda b=2c$ as well as with $\lambda^{-1}b=2c$. We obtain the inequalities
\begin{equation}
    0<g_{01}(Y,Y)\le \left(\frac{b^2}{4c}-c\right)dt_i^2-\gamma_0(Y,Y)-\gamma_1(Y,Y)\ ,\ i=0,1\ .
\end{equation}
We therefore choose $k$ such that
\begin{equation}
  \frac{b^2}{4c}-c\le k\ .  
\end{equation}
and get $g_i'(Y,Y)>0$, $i=0,1$.

\section{Proof of \texorpdfstring{$\Delta X\in\LieRc $}{DeltaX}}\label{app:Delta-X}

In this appendix we prove that the anomaly map $\Delta X$ of the perturbative AMWI \eqref{eq:anomMWI} lies in $\LieRc$.
In a first step we list the defining properties of the subgroup $\Rc_c$ of the St\"uckelberg Petermann renormalization group $\Rc_0$
determined by the renormalization conditions given at the beginning of Sect.~\ref{sec:pQFT}  and obtain the defining properties of $\LieRc$;
we also give an explicit formula for the Lie bracket. 
In a second step we verify that the latter are satisfied by $\Delta X(F)\equiv\Delta(e_\ox^F;hQ)$ 
by using the structural results for $\Delta(e_\ox^F;hQ)$ derived in \cite[Sect.~5.2]{Brennecke08} (see also\cite[Chap.~4.3]{Due19}).

\begin{definition}\label{def:Rc}
The compactly supported subgroup $\Rc_c$ of the St\"uckelberg-Petermann renormalization group is the set of formal power series 
$Z=\sum_{n=0}^\infty\frac{1}{n!}Z_n$,
with $n$-linear symmetric maps $Z_n$ of local functionals to local functionals (cf.~\eqref{eq:Z-perturbative}), 
with the following properties:
\begin{itemize}
    \item[(1)] $Z_1$ is invertible,
    \item[(2)] $Z(F+G)=Z(F)-Z(0)+Z(G)$ for $\supp F\cap\supp G=\0$, $F,G\in\Floc(M)$,
    \item[(3)] $Z(F^{\psi}+\delta L(\psi))=Z(F)^{\psi}+\delta L(\psi)$ for $\psi\in\Dc(M,\RR^n)$,
    \item[(4)] $Z(F)^{\psi}=Z(F^{\psi})$ for $\psi\in\Dc(M,\RR^n)$,
    \item[(5)] $\supp Z$ is compact where the support is defined analogously to \eqref{eq:suppZ}.
\end{itemize}
\end{definition}
Mind the difference: the $n$-fold time ordered product $F_1\cdot_T\,\,\cdots\,\,\cdot_T F_n$ is a $\CC$-valued functional, 
but $Z(F)$ is an $\RR$-valued functional (by definition of $\Floc$) -- this is the reason why there is no defining property for $\Rc_c$ 
corresponding to the renormalization condition unitarity for the $S$-matrix. 
An immediate consequence of the property (4) is that $\supp Z(F)=\supp F$ for all $F\in\Floc(M)$ and thus $Z(0)=Z_0=\const$.
Note that condition (2) implies due to multilinearity of the coefficients of $Z$ the general locality condition
\begin{itemize}
    \item[(2')] $Z(F+G+H)=Z(F+H)-Z(H)+Z(G+H)$ with $F,G$ as above and $H\in\Floc$
\end{itemize}
(see \cite[Appendix B]{BDF09}); and, assuming the validity of (4), (3) can equivalently be written in the simpler form
\begin{itemize}
    \item[(3')] $Z(F+\delta L(\psi))=Z(F)+\delta L(\psi)$ for $\psi\in\Dc(M,\RR^n)$.
\end{itemize}

We also point out that each $Z_k$, $k\in\NN$, satisfies
the properties (2)-(5) individually -- this observation is crucial for the proof of the Main Theorem \cite{DF04}. 
It is obvious for (2), (4) and (5); and (3') can be verified as follows: 
\begin{align}
Z_k\bigl(F+\delta L(\psi)\bigr)=&\frac{d^k}{d\lambda^k}\Big\vert_{\lambda=0}\, Z\bigl(\lambda(F+\delta L(\psi))\bigr)\nonumber\\                        =&\frac{d^k}{d\lambda^k}\Big\vert_{\lambda=0}\, Z\bigl(\lambda F+\delta L(\lambda\psi)+c_\psi(\lambda)\bigr)\nonumber\\
=&\frac{d^k}{d\lambda^k}\Big\vert_{\lambda=0}\, \Bigl(Z(\lambda F)+\delta L(\lambda\psi)+c_\psi(\lambda)\Bigr)\nonumber\\
=&Z_k(F)+\delta_{k1}\,\delta L(\psi),
\end{align}
where $c_\psi(\lambda)$ is a constant and we use that $\delta L(\psi)=\langle\phi,K\psi\rangle+\tfrac{1}2\langle\psi,K\psi\rangle$ and that 
$Z(F+c)=Z(F)+c$ (since $\supp c=\emptyset$). 
In particular note that $Z_1\bigl(\delta L(\psi)\bigr)=\delta L(\psi)$. 

 As mentioned in Subsection \ref{subsec:Pt}, we do not require that $Z_1=\mathrm{id}$, in contrast to \cite{DF04,BDF09}, but in agreement with \cite{HWRG} . The reason is that the functional $F$ as the argument of $S$ has to be understood as related to a quantum observable $\hat{F}$ by normal ordering. A symmetry transformation then acts on $F$ in two ways: via its action on field configurations and via its action on the normal ordering prescription. The latter action is incorporated in the action of the renormalization group in the unitary AMWI.

Usually a normal ordering prescription 
is of the form
\be\label{eq:normal-ordering}
\hat F=\alpha_H^{-1} F\quad\text{with}\quad\alpha_H\doteq e^{\frac{1}{2}\,\langle H,\frac{\delta^2}{\delta\phi^2}\rangle}
\ee
where $H$ coincides up to smooth terms with an admissible $2$-point function of the free theory (\ie satisfying microlocal spectrum condition). 
We point out that the algebra of quantum observables $\hat F$ is 
\emph{abstractly} defined and the formula \eqref{eq:normal-ordering} has to be understood in a formal sense, because $\alpha_H^{-1} F$ does, in general, not exist as a functional,
see \cite{BDF09} and \cite{FredenhagenR15}.
If the normal ordering is changed one should keep $\hat{F}$ fixed and therefore has to admit changes of $F$.
More explicitly, if the new normal ordering is characterized by $H'$ we find
\be
\hat{F}=\alpha_{H'}^{-1}Z_1(F)\quad \text{and hence}\quad Z_1=\alpha_{H'-H}\ .
\ee
Note that $Z_1$ is well defined on polynomial functionals, since $H'-H$ is smooth.
The set of admissible $H$ is restricted by the above condition 
(5) on $Z_1$. ((1)-(4) 
are obviously satisfied, because $Z_1=\alpha_{H'-H}$ is invertible, linear, acts trivially on $\delta L(\psi)=\langle\phi,K\psi\rangle+\tfrac{1}2\langle\psi,K\psi\rangle$
and commutes with functional derivatives.) Explicitly, since $Z_1$ must have compact support, also
$H'-H$ and all its derivatives restricted to the diagonal must have compact support. Note that $H$ appearing in \eqref{eq:normal-ordering} is a solution of the 
free field equation, but all further admissible $H'$ do not necessarily have this property. 
 
 

The associated Lie algebra $\LieRc$ is defined as follows: it is
the set of formal power series $z=\sum_{n=0}^\infty \frac{1}{n!}z_n$, with $n$-linear symmetric maps $z_n$ of local functionals to local functionals, with the properties
\begin{itemize}
    \item[$(i)$]$\mathrm{id}+\lambda z_1$ is invertible for $\lambda$ sufficiently small,
    \item[$(ii)$] $z(F+G)=z(F)-z(0)+z(G)$ for $\supp F\cap\supp G=\0$, $F,G\in\Floc(M)$,
    \item[$(iii)$]$z(F^{\psi}+\delta L(\psi))=z(F)^{\psi}$ for $\psi\in\Dc(M,\RR^n)$,
     \item[$(iv)$]$z(F)^{\psi}=z(F^{\psi})$ for $\psi\in\Dc(M,\RR^n)$,
     \item[$(v)$] the support of $z$, \begin{align}\label{eq:suppz}
&\supp z\doteq \{x\in M\,|\, \text{ for every neighborhood }U\ni x\,\text{ there exist}\ F,G\in\Floc(M),\nonumber \\
              &\text{ with }\supp F\subset U\text{ such that } z(F+G)\neq z(G)\},
\end{align} is compact.
\end{itemize}
To obtain an explicit formula for $[z^a,z^b]_{\LieRc}$, we use that it is connected to the product 
in the Lie group $\Rc_c$ by the expansion
\be\label{eq:Lie-bracket}
Z^a(\lambda)Z^b(\lambda)Z^a(\lambda)^{-1}Z^b(\lambda)^{-1}=\mathrm{id}+\lambda^2\,[z^a,z^b]_{\LieRc}+\mathcal{O}(\lambda^3)\ ,
\ee
where
\be
Z^i(\lambda)\doteq\mathrm{id}+\lambda\,z^i\quad,\ i=a,b\ .\ee
By taking into account that
\be\label{eq:z(F+cG)}
z(F+\lambda G)=z(F)+\lambda \langle z'(F),G\rangle+\mathcal{O}(\lambda^2)\ ,
\ee
where 
\be
\langle z'(F),G\rangle\doteq\frac{d}{d\lambda}\Big\vert_{\lambda=0}z(F+\lambda G)=\sum_{n=1}^\infty\frac{1}{(n-1)!}\,z_n(F^{\otimes(n-1)}\otimes G)
\ee
is linear in $G$, we obtain
\be\label{eq:Z-inverse}
Z(\lambda)^{-1}(F)=F-\lambda\, z(F)+\lambda^2\,\langle z'(F),z(F)\rangle+\mathcal{O}(\lambda^3)\ .
\ee
By using \eqref{eq:z(F+cG)} and \eqref{eq:Z-inverse}, we see by a straightforward computation that $Z^a(\lambda)Z^b(\lambda)Z^a(\lambda)^{-1}Z^b(\lambda)^{-1}$ is indeed
of the form given in \eqref{eq:Lie-bracket}, and we can read off that
\be\label{eq:Lie-bracket-1}
[z^a,z^b]_{\LieRc}(F)=\langle(z^a)'(F),z^b(F)\rangle-\langle(z^b)'(F),z^a(F)\rangle\ .
\ee


\begin{remark}The renormalization group $\Rc_c$ may be considered as a group of diffeomorphisms of the space of local functionals. Accordingly the associated Lie algebra corresponds to a Lie algebra of vector fields equipped with the usual Lie bracket.
\end{remark}

Finally we are going to prove that $\Delta X\in\LieRc$.
In first order, $\Delta X$ is a second order differential operator, due to the two contractions in the time ordered product 
$F\cdot_T \partial_XL$ (see \eqref{eq:Delta-A1}), terms with more contractions do not occur since $\partial_XL$ is of second order in $\phi$.
Hence $\mathrm{id}+\lambda (\Delta X)_1$ is invertible on polynomial local functionals, so condition $(i)$ holds. 
$\Delta X$ evidently also satisfies conditions $(ii)$ and $(v)$. Condition $(iv)$ follows from the condition field independence
(which implies formula (5.27) of \cite{Brennecke08}), due to the fact that $\G_c(M)$ contains only affine field redefinitions, hence the term in equation (5.26) of \cite{Brennecke08} does not contribute.


$(iii)$ is not explicitly mentioned in \cite{Brennecke08}. It follows from the following calculation. We use the notation $\psi_L$ and $\psi_{\ast}$ introduced in the proof of Lemma \ref{lem:Zg-in-R}, so that $(iii)$
assumes the form $z\psi_L=\psi_{\ast}z$, and find
\begin{align*}
S(\psi_{L}F)\cdot_T\Delta X({\psi_{L}F})\quad\overset{\mathclap{\eqref{eq:anomMWI}}}{=}
&\quad\frac{d}{id\lambda}\Big\vert_{\lambda=0}   S\bigl(\g^{\lambda}_{L}(\psi_{L}F)\bigr)\ &\mathrm{mod}\ \frac{\delta L}{\delta \phi}\\
\overset{\mathclap{\eqref{eq:gL-psiL}}}{=}&\quad \frac{d}{id\lambda}\Big\vert_{\lambda=0}S\bigl(({\g'}^\lambda\psi)_{L}\g^\lambda_{L}F\bigr)&\mathrm{mod}\ \frac{\delta L}{\delta \phi}\\
\overset{\mathclap{(\mathrm{FE})}}{=}&\quad \frac{d}{id\lambda}\Big\vert_{\lambda=0}S\bigl(\g^\lambda_{L}F\bigr)&\mathrm{mod}\ \frac{\delta L}{\delta \phi}\\
 \overset{\mathclap{\eqref{eq:anomMWI}}}{=}&\quad \frac{d}{id\lambda}\Big\vert_{\lambda=0}S\bigl(F+\lambda\Delta X(F)\bigr)&\mathrm{mod}\ \frac{\delta L}{\delta \phi}\\
 \overset{\mathclap{(\mathrm{FE})}}{=}&\quad \frac{d}{id\lambda}\Big\vert_{\lambda=0}S\bigl(\psi_{L}(F+\lambda \Delta X(F))\bigr)&\mathrm{mod}\ \frac{\delta L}{\delta \phi}\\
 =&\quad  S(\psi_{L}F)\cdot_T(\psi_{\ast}\Delta X(F))&\mathrm{mod}\ \frac{\delta L}{\delta \phi} 
\end{align*}
We now use the fact that $\Delta X(G+c)=\Delta X(G)$ for $G\in\Floc(M)$ and a constant functional $c$ and add a source term $\langle\phi,q\rangle$ to the Lagrangian, thus replacing $L$ by $L_q\doteq L+\langle\phi,q\rangle$. Neither the time ordered product nor 
$\Delta X$ nor $\Delta X(\psi_{L_q}F)$ depend on $q$, and $S(\psi_{L_q}F)=e^{i\langle\psi,q\rangle}S(\psi_LF)$.
Hence the equation 
\be 
S(\psi_{L}F)\cdot_T\Delta X({\psi_{L}F})=S(\psi_{L}F)\cdot_T(\psi_{\ast}\Delta X(F))
\ee
holds everywhere. Since the off-shell $S$-matrix is invertible with respect to $\cdot_T$, we arrive at condition $(iii)$.
\end{appendix}

\end{document}